\documentclass[12pt]{article}
\usepackage[utf8]{inputenc}
\usepackage[T1]{fontenc}
\usepackage[sc]{mathpazo}
\usepackage[english]{babel}
\usepackage{microtype}
\usepackage{datetime}
\usepackage{geometry}
\usepackage{setspace}
\usepackage{graphicx}
\usepackage{caption}
\usepackage{subcaption}
\usepackage{booktabs}
\usepackage[flushleft]{threeparttable}
\usepackage{tabularx}
\usepackage{ragged2e}
\usepackage{multirow}
\usepackage{rotating}
\usepackage{amsmath,amssymb,amsthm}
\usepackage{dsfont}
\usepackage{bm}
\usepackage{hyperref}
\usepackage[open,openlevel=2,atend,numbered,level=4]{bookmark}
\usepackage[nameinlink]{cleveref}
\usepackage[authoryear]{natbib}
\usepackage{enumitem}
\usepackage[dvipsnames]{xcolor}
\usepackage{tikz}
\usepackage{titling}

\definecolor{deepblue}{RGB}{0,0,110}
\definecolor{darkgreen}{rgb}{0.0,0.5,0.0}

\hypersetup{
    pdftitle={Allocating Students to Schools: Theory, Methods, and Empirical Insights},
    pdfauthor={Yeon-Koo Che, Julien Grenet, YingHua He},
    colorlinks=true,
    linkcolor={deepblue},
    citecolor={deepblue},
    urlcolor={deepblue},
    hypertexnames=false,
    pdfpagemode=UseNone,
    pdfdisplaydoctitle=true
}

\setlist[enumerate]{leftmargin=1.5cm,rightmargin=0.5cm,noitemsep, topsep=2pt}

\appto\TPTdoTablenotes{\labelsep0.0em\footnotesize}

\allowdisplaybreaks

\theoremstyle{definition}
\newtheorem{theorem}{Theorem}

\crefname{equation}{}{}
\crefname{theorem}{Theorem}{Theorems}
\crefname{figure}{Figure}{Figures}
\crefname{table}{Table}{Tables}
\crefname{section}{Section}{Sections}

\DeclareMathOperator{\argmax}{arg\,max}

\def\half{\frac{1}{2}}
\newcommand{\Prob}{\mathbb{P}}
\newcommand{\z}{\mathbf{z}}
\newcommand{\DA}{\text{DA}}
\newcommand{\TTC}{\text{TTC}}

\makeatletter
\renewcommand*{\@fnsymbol}[1]{%
  \ifcase#1 \or *\or 1\or 2\or 3\fi%
}
\makeatother

\title{%
    \textbf{Allocating Students to Schools:\\
    Theory, Methods, and Empirical Insights}\thanks{
    We are grateful to Christopher Campos, Aram Grigoryan, Dongwoo Hahm, Adam Kapor, Tomás Larroucau, Michael Lee, Minseon Park, and Olivier Tercieux for their comments. Our dear friend and coauthor, YingHua He, passed away in July 2024. We dedicate this chapter to his memory.
}
}

\author{
    Yeon-Koo Che\thanks{Department of Economics, Columbia University, USA. Email: \href{mailto:yeonkooche@gmail.com}{\texttt{yeonkooche@gmail.com}}.}\qquad
    Julien Grenet\thanks{CNRS and Paris School of Economics, France. Email: \href{mailto:julien.grenet@psemail.eu}{\texttt{julien.grenet@psemail.eu}}.}\qquad
    YingHua He\thanks{Department of Economics, Rice University, USA. Email: \href{mailto:yinghua.he@rice.edu}{\texttt{yinghua.he@rice.edu}}.}
}

\date{\normalsize%
    Chapter~4 of the \emph{Handbook of the Economics of Matching}, Volume~2 (Elsevier)\\
    edited by Yeon-Koo Che, Pierre-André Chiappori, and Bernard Salanié\\[1ex]
    DOI: \href{https://doi.org/10.1016/bs.hesmat.2025.10.004}{10.1016/bs.hesmat.2025.10.004}\\[4ex]%
    August 2025
}

\begin{document}

\setstretch{1.22} 

\thispagestyle{empty}

\maketitle

\begin{abstract}\singlespacing%
\noindent This chapter surveys the application of matching theory to school choice, motivated by the shift from neighborhood assignment systems to choice-based models. Since educational choice is not mediated by price, the design of allocation mechanisms is critical. The chapter first reviews theoretical contributions, exploring the fundamental trade-offs between efficiency, stability, and strategy-proofness, and covers design challenges such as tie-breaking, cardinal welfare, and affirmative action. It then transitions to the empirical landscape, focusing on the central challenge of inferring student preferences from application data, especially under strategic mechanisms. We review various estimation approaches and discuss key insights on parental preferences, market design trade-offs, and the effectiveness of school choice policies.\\[1ex]
\end{abstract}


\newpage

\begingroup
\footnotesize
\setstretch{1.18} 
\tableofcontents
\endgroup


\section*{Introduction}
\addcontentsline{toc}{section}{Introduction}

The landscape of primary and secondary education has been significantly reshaped by the rise of school choice, a movement born out of concerns with traditional geographically based school assignment systems. For decades, the prevailing model automatically enrolled students in the public school nearest to their home. While convenient, this neighborhood assignment system frequently led to significant issues, including a lack of parental choice, particularly for families in areas with underperforming schools, and inequities stemming from socioeconomic disparities mirrored in school resources. Furthermore, the strong link between housing markets and school allocation under neighborhood assignment often resulted in allocational inefficiencies, as housing choices heavily dictated school placements. These concerns fueled advocacy for policies empowering parents with greater control over their children's educational paths.

Early school choice initiatives often centered on magnet schools and open enrollment programs, which allowed students to attend public schools outside their designated districts.\footnote{Throughout, we identify school choice with open enrollment, thus excluding other initiatives such as charter schools and voucher programs, which some authors include as part of the school choice movement.} Minnesota, for instance, enacted the first statewide open enrollment law in 1988, a model subsequently adopted by other states. Major urban centers like New York, Boston, and Seattle have also embraced such programs, with a growing focus on designing allocation mechanisms that are both efficient and fair. From a classical economics perspective, expanding choice in education is anticipated to enhance efficiency, much like broadening the scope of markets. However, a critical distinction arises: school choice is not typically mediated by price. Consequently, whether school choice fulfills its promise of improved allocation and other benefits hinges crucially on how these systems are organized and implemented.

It is precisely at this juncture that matching theory has become indispensable, offering powerful tools for both the theoretical design and empirical evaluation of school choice mechanisms. Theoretically, matching theory has been instrumental in analyzing the inherent trade-offs between competing desiderata such as fairness, efficiency, and strategy-proofness—--the property that individuals have no incentive to misrepresent their preferences. It has also provided frameworks for optimizing tie-breaking procedures in situations of coarse priorities (where many students may have the same priority level), achieving desirable cardinal welfare outcomes that consider the intensity of preferences, and implementing affirmative action policies effectively.

Empirically, the application of matching theory has enabled researchers to infer student preferences from school choice data and to evaluate the welfare implications of various assignment systems. The increasing availability of administrative data from school choice programs has spurred a rich body of empirical work focused on estimating student preferences and simulating the effects of counterfactual policies. A central challenge in this empirical agenda is inferring true preferences from observed ranking submissions, especially under mechanisms where strategic reporting might be prevalent.

While this chapter is primarily focused on school choice, we also draw selectively on insights from the closely related literature on centralized college admissions, particularly in settings where applicants are ranked according to exogenous priority criteria that closely parallel the school choice problem. The goal is not to provide a comprehensive review of all education-related matching markets---such as childcare assignment, foster-care placements, teacher assignment, or other emerging applications---but rather to highlight methodological contributions relevant for the analysis of school choice and to incorporate empirical evidence that informs its broader understanding. Within this scope, the chapter will navigate the multifaceted contributions of matching theory to the study of school choice. It is organized into three main parts.

\Cref{sec:theory} explores the theoretical underpinnings, beginning with the simple economics of school choice to illustrate the shortcomings of geographic assignment and the potential benefits of choice-based systems. It then establishes the formal model framework for school choice, discusses assignment methods used before major reforms, and explores desirable mechanisms like Deferred Acceptance and Top Trading Cycles, paying close attention to the trade-offs between efficiency and stability, the implications of coarse priorities, the pursuit of cardinal welfare, and the design of affirmative action policies.

\Cref{sec:methods} transitions to the empirical landscape, examining methods for estimating student preferences under both strategic and non-strategic (strategy-proof) mechanisms. This section will cover various approaches, including those that model preference reporting as a choice over lotteries, those grounded in the stability of matching outcomes, and those dealing with incomplete models of behavior, alongside a discussion of how to select among alternative preference hypotheses.

Building on these foundations, \Cref{sec:empirical_insights} synthesizes empirical lessons from the implementation of centralized school choice systems around the world, organizing the discussion around three central themes. It first examines how families form preferences and make application decisions, with particular attention to information frictions and behavioral deviations from full rationality. It then turns to the design trade-offs in school choice systems, including the degree of centralization, the choice of assignment mechanisms, and the use of decision-support tools. Finally, it assesses the broader impact of school choice reforms on market-level outcomes, equity, and related behaviors such as residential mobility. Together, these empirical insights illustrate how the tools of matching theory have deepened our understanding of real-world school choice and informed the development and evaluation of assignment policies across diverse institutional contexts.


\section{Theories of School Choice}
\label{sec:theory}

\subsection{A Simple Economics of School Choice}
\label{sec:toy}

To illustrate the potential shortcomings of geographic assignment and the potential benefits of school choice, consider a simplified model. We have a unit mass of families, each with one child seeking to attend a public school. There are two districts, $A$ and $B$, each capable of accommodating half the population, and each district has one public school: $a$ in $A$ and $b$ in $B$. 

We represent the preferences of each family with a point $(x, y)$ within $[-1/2, 1/2]^2$, as shown in \Cref{fig:preference_square}. The horizontal axis ($x$) captures the family's preference for residing in district~$A$ versus $B$, with higher values of $x$ indicating a stronger preference for $A$. Similarly, the vertical axis ($y$) represents the family's preference for school~$a$ over $b$, with higher $y$ values signifying a stronger preference for $a$. Assume for now that family preferences $(x, y)$ are uniformly distributed on the square $[-1/2, 1/2]^2$.

In an ideal allocation, a family would reside in district~$A$ if and only if $x > 0$ (they prefer $A$), and their child would attend school~$a$ if and only if $y > 0$ (they prefer $a$). The diagram visually depicts these ideal choices, with the light red area representing where school~$a$ is preferred and the light blue area showing where district~$A$ is preferred. The overlap, shaded in light purple, highlights the families who ideally desire both to live in district~$A$ and to have their child attend school~$a$.

\begin{figure}[!ht]
\begin{center}
\begin{tikzpicture}
\begin{scope}[scale=8]

\node[below left] at (0,0) {\footnotesize$(-\half,-\half)$};
\node[below right] at (1,0) {\footnotesize$( \half,-\half)$};
\node[above left] at (0,1) {\footnotesize$(-\half,\half)$};
\node[above right] at (1,1) {\footnotesize$(\half,\half)$};

\fill[red!30] (0.5,0) rectangle (1,0.5);
\node[red] at (0.75,0.25) {$A$ selected};

\fill[blue!30] (0,0.5) rectangle (0.5,1);
\node[blue] at (0.25,0.75) {$a$ selected};

\fill[violet!30] (0.5,0.5) rectangle (1,1);

\draw[thick] (0,0) rectangle (1,1);

\draw[fill] (0.5,0.5) circle [radius=0.001];

\draw[->] (0,0.5) -- (1.1,0.5) node[right] {\footnotesize $x$ (Residence Preference)};
\draw[->] (0.5,0) -- (0.5,1.1) node[above] {\footnotesize $y$ (School Preference)};

\end{scope}
\end{tikzpicture}
\caption{Ideal Residential and School Choices}
\label{fig:preference_square}
\end{center}
\end{figure}
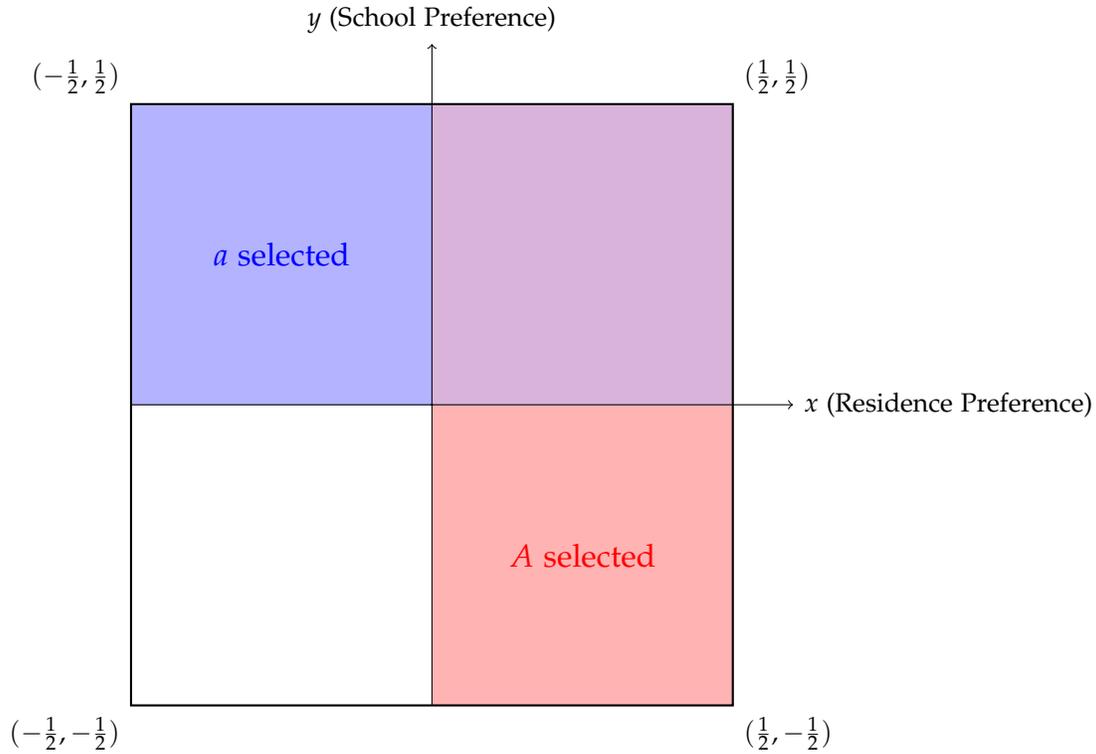

Under geographic assignment, families choose their child's school indirectly by choosing their residence. Essentially, a family chooses a bundle of a residential district and a school, facing housing rental prices $r_A$ and $r_B$, respectively, in districts~$A$ and $B$. Assuming a model where the utility difference between the bundles $(A,a)$ and $(B, b)$ is $x + y$, the housing market will reach equilibrium when the price difference between the two districts ($\delta = r_A - r_B$) is 0. The resulting assignment is depicted in \Cref{fig:preference_square2}. 

\begin{figure}[!ht]
\begin{center}
\begin{tikzpicture}

\begin{scope}[scale=8]

\node[below left] at (-0.5,-0.5) {\footnotesize$(-\half,$-$\half)$};
\node[below right] at (0.5,-0.5) {\footnotesize$( \half,-\half)$};
\node[above left] at (-0.5,0.5) {\footnotesize$(-\half,\half)$};
\node[above right] at (0.5,0.5) {\footnotesize$(\half,\half)$};

\draw[thick] (-0.5,0.5) -- (0.5,-0.5);

\draw[dotted] (-0.5,0) -- (0.5,0);
\draw[dotted] (0,-0.5) -- (0,0.5);


\fill[blue!30] (-0.5,0.5) -- (0,0) -- (0,0.5) -- cycle; 
\node[font=\scriptsize] at (-0.2,0.4) {$A$ instead of $B$};

\fill[red!30] (-0.5,0) -- (0,0) -- (-0.5,0.5) -- cycle;
\node[font=\scriptsize] at (-0.3,0.1) {$b$ instead of $a$};

\fill[green!30] (0.5,-0.5) -- (0,0) -- (0.5,0) -- cycle;
\node[font=\scriptsize] at (0.3,-0.1) {$B$ instead of $A$};

\fill[orange!30] (0,0) -- (0.5,-0.5) -- (0,-0.5) -- cycle;
\node[font=\scriptsize] at ( 0.2,-0.4) {$a$ instead of $b$};

\draw[thick] (-0.5,-0.5) rectangle (0.5,0.5);

\draw[->] (-0.5,0) -- (0.6,0) node[right] {\footnotesize $x$ (Residence Preference)};
\draw[->] (0,-0.5) -- (0,0.6) node[above] {\footnotesize $y$ (School Preference)};

\end{scope}

\end{tikzpicture}
\caption{Choices under Geographic School Assignment}
\label{fig:preference_square2}
\end{center}
\end{figure}
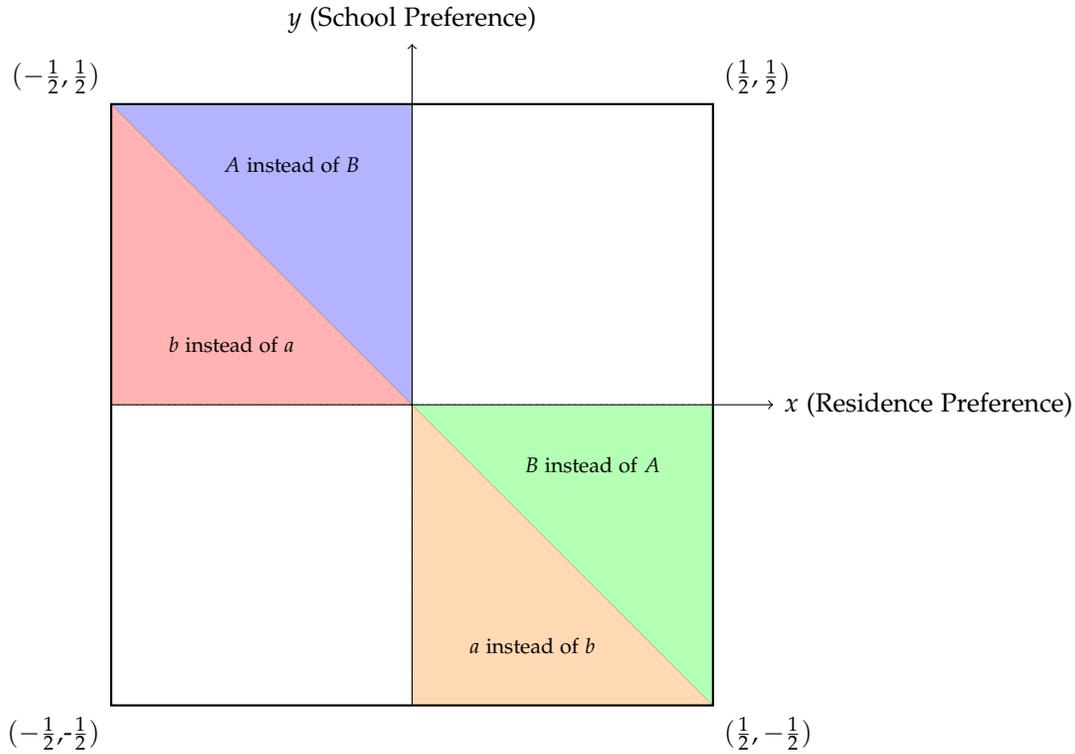

Geographic assignment, by tethering school choice to residence, compels the housing market to reconcile two distinct preferences: where families wish to live and where they want their children to attend school. This forced coupling generates distortions, visually represented in the figure by the four colored regions. Families in the blue region, for instance, would ideally reside in district~$B$ but choose $A$ solely to access school~$a$. This might stem from a stronger preference for $B$'s amenities or a parent's workplace location. Under geographic assignment, they sacrifice their residential preference to secure their educational one. Similarly, families in the red region, who ideally desire residence in $B$ and school~$a$, are forced to compromise their educational choice. The green and orange regions illustrate analogous inefficiencies, where families sacrifice either their residential or school preferences due to the constraints of geographic assignment. These inefficiencies underscore how this system can lead to suboptimal outcomes, forcing families to make trade-offs between their ideal living situations and their children's education. 

School choice can mitigate these inefficiencies by decoupling families' residential choices from their school choices. In our example, if families are allowed to send their children to either school, regardless of their district, the two choices become independent. Families can now choose their residence based solely on amenities and housing preferences, and select the ideal school for their children independently. Under this system, families will choose to reside in district~$A$ if and only if $x \ge \delta$, where the price difference $\delta$ adjusts to clear the housing market. They will also (try to) send their children to school~$a$ if and only if $y \ge 0$. Again, assuming a uniform distribution of $(x,y)$, the equilibrium price will be $\delta = 0$, and children will be assigned to their preferred schools. In short, school choice, in this stylized example, can achieve the first-best allocation. 

While the previous example highlights the potential benefits of school choice in mitigating inefficiencies caused by geographic assignment, real-world implementations face additional challenges. Consider a scenario where the preference for schools is not uniformly distributed; instead, school~$a$ is more popular, leading to oversubscription. If tuitions could be adjusted to clear the market, an efficient allocation might emerge. However, this contradicts the principle of free public education. Therefore, alternative methods are necessary to address oversubscription.

One common approach is a lottery system, where seats at school~$a$ are assigned randomly among those who prefer it ($y > 0$). However, this allocation ignores the intensity of preference, potentially leading to students with weak preferences for $a$ being admitted while those with strong preferences are excluded. Alternatively, granting geographic priority to district~$A$ residents for school~$a$ essentially recreates a smaller-scale geographic assignment system, inheriting its associated inefficiencies, albeit to a lesser extent.

In reality, school choice mechanisms often fall short of achieving the theoretical first-best allocation. The lack of monetary transfers in public education necessitates alternative assignment methods based on elicited preferences. The design of these mechanisms becomes crucial in balancing efficiency, fairness, and the practical constraints of public schooling. In the following sections, we will explore how various school choice mechanisms attempt to achieve desirable allocations by strategically linking assignment to reported preferences. 

\subsection{Model Framework}

A school choice model is very close to the many-to-one matching model introduced in Chapter~1. While there is an important difference (highlighted later), the model primitives follow from many-to-one matching, also known as the College Admissions problem. 

A \textbf{school choice problem} consists of a set of \textbf{students}, $\mathcal{I}=\{i_1,\ldots,i_{I}\}$, and a set of \textbf{schools}, $\mathcal{J}=\{s_1,\ldots,s_{J}\}$. Each school~$s\in \mathcal{J}$ has \textbf{capacity}~$q_s$, the maximum number of students the school can enroll. Each student~$i\in \mathcal{I}$ has a strict \textbf{preference ordering} $P_i$ over schools and the option to be unassigned. Each school~$j\in \mathcal{J}$ has a strict \textbf{priority ordering} $\pi_j$ over the students. In the context of the public school system, one may not view priorities as schools' preferences; rather, they can be viewed as reflecting the policy objectives of the school system and other societal considerations. Examples include priority rankings based on sibling attendance, the proximity of residence to schools (e.g., ``walk zone''), and/or test scores and grade point averages of the preceding schools. While priorities are often coarse, with many students falling within the same priority class (e.g., walk zone), it is convenient to begin with strict priorities as our baseline model; we will address the implications of coarse priorities in \Cref{subsec:coarse_priorities}.

An \textbf{assignment} is a mapping $\mu: \mathcal{I}\cup \mathcal{J}\to \mathcal{J}\cup 2^{\mathcal{I}}$ such
that for each $i\in \mathcal{I}$, $\mu(i)\in \mathcal{J} \cup \{i\}$ and, for each $s\in\mathcal{J}$, $\mu(s)\subseteq \mathcal{I}$; i.e., each student is assigned a school or himself/herself (not assigned to any school), and each school is assigned to a subset of students. Further, $\mu(i)=s$ if and only if $i\in \mu(s)$, and $|\mu(s)|\leq q_s$ for each $s$ so that the capacity constraint is satisfied.

There are several normative objectives we consider. First, we say $\mu$ is \textbf{efficient} if no other assignment $\mu'$ \emph{Pareto dominates} $\mu$, i.e., making all students weakly better off and some strictly better off than $\mu$. Note that schools' priority orderings do not factor into this notion, in keeping with the idea that the schools are the resources to be consumed by students, rather than being welfare-relevant entities in their own right, as would be the case in two-sided matching. 

At the same time, schools' priorities reflect policy and social considerations. There is a familiar way to incorporate them if one views the school assignment as a two-sided matching. An assignment is \textbf{stable} if it is (i)~\textbf{individually rational}: for each student~$i\in \mathcal{I}$, $\mu(i)$ is weakly preferred to the option of being unassigned; it is (ii)~\textbf{non-wasteful}: for each student~$i\in \mathcal{I}$, $s P_i \mu(i) \Rightarrow |\mu(s)|=q_s$; and there is (iii)~\textbf{no justified envy}: if a student~$i$ prefers a school~$s$ to his/her assignment, then all students matched to school~$s$ must have a higher priority than student~$i$. While stability is justified in the two-sided matching context as an equilibrium notion, here it is justified as some basic efficiency desiderata (namely, (i) and (ii)) and a fairness desideratum (namely, (iii)). 

In practice, students' preferences are unobserved by the school system, so the above normative properties cannot be satisfied unless the assignment process also provides students with incentives to reveal their preferences. Define an \textbf{assignment mechanism} to be a mapping from all possible profiles of student preferences to an assignment. We say that an assignment mechanism is \textbf{strategy-proof} if each student has an incentive to report his/her preference truthfully, for any preference profile reported by the other students. Strategy-proofness makes an assignment mechanism straightforward for participants, and enables other normative properties, such as efficiency and stability, to be achievable reliably and robustly.
 
\subsection{School Assignment Methods before the Reforms}

Before considering desirable school assignment methods, it is instructive to examine a couple of assignment methods used prior to the reforms.

\subsubsection{New York City School Assignment: Decentralized Matching}

Before the 2004 redesign, the New York City public high school assignment system was highly decentralized and complex. \citet{Abdulkadiroglu_Pathak_Roth(2005)AER:PP} and \citet{Abdulkadiroglu_Agarwal_Pathak(2017)AER} describe the system as follows: Rising NYC high school students applied to five out of more than 600 school programs; they could receive multiple offers and be placed on waitlists. Students were allowed to accept only one school and one waitlist offer, and the cycle of offers and acceptances was repeated twice. Most students not assigned in these rounds went through an administrative process that manually placed them at schools close to their residences. The lack of centralized coordination meant that the number of offers and acceptances was limited, a problem known as \emph{congestion} by \citet{Roth_Xing(1997)JPE}.

How the congestion problem adversely impacts the outcome can be illustrated by a simple example from \citet{Che_Koh(2016)JPE}. Suppose there are only two students, $1$ and $2$, applying to schools $a$ and $b$. Each school has one seat to fill (and faces a prohibitively high cost of enrolling two students). Student~$i$ has priority measured by score $t_i$, $i=1,2$, where $0<t_2<t_1<2t_2$, for both schools. The priority score $t_i$ is a school's payoff from admitting student~$i$. Each student has an equal probability of preferring either school, which is private information (unknown to the other student and to the schools). The applications are free of cost, and the timing of the decentralized process is as follows. First, students apply to schools simultaneously. Then, the two schools simultaneously offer admissions to sets of students. Finally, students decide which offer (if any) to accept. A school's seat remains empty if no applicant accepts its offer. While extremely simple, this model captures the lack of centralized coordination as well as the congestion (i.e., each school has only one round of offers).\footnote{While the NYC mechanism allowed for more flexibility in terms of waitlist management, the coordination challenges are higher due to the large number of students and school programs. One can thus view the simple model as a fair stylization of the system.}

Given the large cost of over-enrollment, each school admits only one student. Their payoffs are described as follows:
\medskip
\begin{center}
\renewcommand{\arraystretch}{1.22}
\begin{tabular}{|c|c|c|}
\hline
$a$'s strategy \textbackslash{} $b$'s strategy & $b$ admits 1 & $b$ admits 2\tabularnewline\hline
$a$ admits 1 & $\half t_1$, $\half t_1$ & {$t_1, t_2$}\tabularnewline\hline
$a$ admits 2 & {$t_2, t_1$} & $\half t_2$, $\half t_2$ \tabularnewline\hline
\end{tabular}
\end{center}
\medskip
This game has a battle of the sexes structure (with asymmetric payoffs), so there are two different types of equilibria. First, there are two asymmetric pure-strategy equilibria in which one school admits student~1 and the other admits student~2. There is also a mixed-strategy equilibrium in which each school admits 1 with probability $\gamma:=\frac{2t_1-t_2}{t_1+t_2}>1/2$ and admits 2 with probability $1-\gamma$, where $\gamma$ is chosen such that the other school is indifferent over whom to admit. Both types of equilibria show the pattern of strategic targeting. In the pure-strategy equilibria, schools manage to avoid competition by targeting different students. The mixed-strategy equilibrium may involve schools using some ``spurious'' measure of student merit as a randomization device; schools following these strategies will appear to value such measures despite there being no intrinsic value to them.

This example, while extremely simple, reveals problems with decentralized matching in terms of welfare and fairness. First, the student with the highest score (student~1) may attend a less preferred school (in both types of equilibria), even though both schools prefer that student; that is, justified envy arises. Second, it could be the case that student~1 prefers~$a$ and student~2 prefers $b$, but the former is assigned $b$ and the latter is assigned $a$, showing that the equilibrium outcome is inefficient among students. Lastly, the mixed-strategy equilibrium is wasteful because both schools may admit the same student, thus wasting the seat of the unmatched school. 

\subsubsection{Boston School Assignment: Immediate Acceptance}
\label{subsubsec:IA}

Unlike NYC, school systems in many cities such as Boston, Cambridge, Minnesota, and Seattle were using a centralized assignment process; but the assignment algorithm used in these districts, known as the \textbf{Immediate Acceptance (IA) mechanism} or simply \textbf{Boston mechanism}, was criticized on the normative ground mentioned above. 

IA takes as input the rank-order list of schools submitted by students and the strict priority ordering of students for each school, based on sibling attendance and walk zone criteria, with ties broken by lotteries. Then, the assignment proceeds in multiple rounds. In the first round, students are assigned to their top-ranked school based on their priority order, starting with the highest priority down to the second highest, and so on. If a school's capacity is filled, it rejects the remaining students. In round two, rejected students are considered for the next highest-ranked school if seats are available at these schools; otherwise, they are rejected. This continues until all students have been assigned to a school or have exhausted their acceptable options. 

To illustrate IA, consider the following example: Let $\mathcal{I}=\{1,2,3,4\}$ be the set of students and $\mathcal{J}=\{a,b,c\}$ be the set of schools, with school capacities $q_a=2$ and $q_b=q_c=1$. The student preferences and school priorities are described in \Cref{fig:IA}.

\begin{figure}[!ht]
\centering
\renewcommand{\arraystretch}{1.22}
\begin{tabular}{c|c|c|c p{1.5cm} c|c|c}
\multicolumn{4}{c}{} & & \multicolumn{3}{c}{} \\[-1.5ex]
$P_1$ & $P_2$ & $P_3$ & $P_4$ & & $\succ_a$ & $\succ_b$ & $\succ_c$ \\
\cline{1-4} \cline{6-8}
$a$   & $a$   & $b$   & $b$   & & $4$ & $3$ & $1$ \\
$b$   & $b$   & $a$   & $a$   & & $1$ & $4$ & $2$ \\
      &       &       & $c$   & & $2$ & $1$ & $3$ \\
\multicolumn{4}{c}{}          & & $3$ & $2$ & $4$ \\
\end{tabular}
\caption{Immediate Acceptance}
\label{fig:IA}
\end{figure}

The outcome of IA in this example is $\mu=(1a, 2a,3b,4c)$ if the students report their preferences truthfully. 

The example illustrates several problems with IA. First, the matching is not stable, assuming that students report their preferences truthfully. For instance, student~1 is assigned school~$a$, but school~$a$ prefers student~4, who also prefers school~$a$ to her assigned school~$c$. The failure of stability could lead to outcomes that violate school priorities. This means that a student may be assigned to a school even if another student with a higher priority for that school remains unassigned. Such priority violations raise concerns about fairness and can even result in legal challenges from dissatisfied students and parents.

The instability in turn implies that the mechanism is not strategy-proof: student~4 could obtain her second most-preferred school~$a$ by misrepresenting her preferences to $(a,b,c)$. This failure of strategy-proofness is intuitive. If a student's most preferred school is highly popular and thus difficult to get into, but her second most-preferred is not much worse and is less competitive, then strategically ranking the latter for her top spot is quite rational.\footnote{At a conference organized by the Federal Reserve Bank of Chicago in 1994 (``Midwest approaches to school reform''), Meyer and Glazerman report:
\begin{quote}
It may be optimal for some families to be strategic in listing their school choices. For example, if a parent thinks that their favorite school is oversubscribed and they have a close second favorite, they may try to avoid ``wasting'' their first choice on a very popular school and instead list their number two school first.
\end{quote}
In a meeting of the West Zone Parents Group of the city of Boston, it was said:
\begin{quote}
One school choice strategy is to find a school you like that is undersubscribed and put it as a top choice, or, second a school that you like that is popular and put it as a first choice and find a school that is less popular for a ``safe'' second choice.
\end{quote}} 
A lab experiment conducted by \citet{Chen_Sonmez(2006)JET} shows that about 80\% of the subjects misreport their preferences under IA. The failure of strategy-proofness raises multiple issues. First, how participants strategize is not straightforward, requiring sophistication in strategic thinking. Second, participants may differ in the sophistication of their strategic thinking and may put the more sophisticated at an advantage.\footnote{\label{fn:Pathak-Sonmez} \citet{Pathak_Sonmez(2008)AER} apply the full-information Nash equilibrium analysis of \citet{Ergin_Sonmez(2006)JPubE} to study how sincere applicants who reveal their preferences truthfully fare against sophisticated ones who play a best response, and formalize the sense in which naivety disadvantages the former.} Third, the strategic nature of the mechanism makes it more challenging for one to decipher the data to estimate preferences accurately. We turn to these issues in the empirical part of this survey.

How can we predict the outcome of IA and its welfare implications? Predicting the outcome of a strategic mechanism is challenging since one must rely on some notion of equilibrium; it is often unclear whether participants play the chosen notion of equilibrium. A case in point is \citet{Ergin_Sonmez(2006)JPubE}, which studies Nash equilibria of the IA mechanism \emph{under the assumption that each student observes the preferences and priorities of all other students}; they show that the set of Nash equilibria coincides with that of stable matching outcomes.\footnote{The argument for this is rather simple. Take any stable matching; say $\mu=(1a, 2c, 3b, 4a)$. It is an equilibrium for all students to report only the assigned school under $\mu$ to be acceptable; a unilateral deviation can never enable a deviator to obtain a better school (since there is no blocking pair). Meanwhile, any unstable matching can never be an equilibrium, since one can unilaterally deviate to obtain a block (either by herself or with a school).} Such an analysis sheds some light on the difference between DA and IA, yet the full-information Nash equilibrium concept is very restrictive and unlikely to hold in practice.
 
\subsection{Desirable Mechanisms: Efficiency vs.\ Stability}

The challenges encountered in both NYC and Boston school assignment systems highlight the need to have an assignment mechanism that is straightforward for the participants to play; namely, the assignment mechanism should be strategy-proof. Beyond the strategy-proofness, one could aim for two desiderata: stability (less ambitiously, no justified envy) and efficiency.

\paragraph{Deferred Acceptance and Top Trading Cycles.}

As it turns out, there are strategy-proof mechanisms that attain these objectives. Stability can be attained by the Student-Proposing Deferred Acceptance (DA) algorithm which runs in multiple rounds \citep[see][]{Gale_Shapley(1962)AmMathMon, Abdulkadiroglu_Sonmez(2003)AER}: in round~1, students apply to their most-preferred acceptable schools, and the schools accept tentatively the students up to their capacities based on their priority orderings and reject the rest; in the subsequent rounds, rejected students apply to the most-preferred schools that have not rejected them, and the schools again tentatively accept the students according to their priority orderings up to the capacity, and reject the rest, possibly including those whom they may have accepted before; this process continues until there are no more rejected students.\footnote{The reader is referred to the formal description of DA in Chapter~1.} The DA mechanism is strategy-proof: namely, students have no incentive to misrepresent their preference orderings when they are used to implement DA.

Meanwhile, the Top Trading Cycles (TTC) mechanism guarantees efficient assignment among students. TTC operates by recursively forming and clearing cycles of students and schools \citep{Abdulkadiroglu_Sonmez(2003)AER}. Initially, each school is assigned a counter to track the number of remaining seats. In step~1, 
each student points to her most preferred school, and each school points to the student with the highest priority. This creates at least one cycle. Students in a cycle are assigned their pointed-to schools, removed from the process, and the school's counter is decremented. In the subsequent steps, the process repeats with the remaining students and schools. Students point to their most preferred school among those with available seats, and schools point to their highest-priority remaining student. Cycles are formed, students are assigned, and counters are updated. The algorithm ends when all students are assigned or all submitted choices have been considered. TTC is strategy-proof, meaning students have no incentive to misrepresent their preference orderings when they are used to implement TTC. 

Unfortunately, stability and efficiency may not be both satisfied. To see this, suppose there are three students $\mathcal{I}=\{1,2,3\}$ and three schools $\mathcal{J}=\{a,b,c\}$, each with unit quota. The students' preferences and priorities are described in \Cref{fig:long}.
\begin{figure}[!ht]
\centering
\renewcommand{\arraystretch}{1.22}
\begin{tabular}{c|c|c p{1cm} c|c|c}
\multicolumn{3}{c}{} & & \multicolumn{3}{c}{} \\[-1.5ex]
$P_1$ & $P_2$ & $P_3$ & & $\succ_a$ & $\succ_b$ & $\succ_c$ \\
\cline{1-3} \cline{5-7}
$b$   & $a$   & $a$   & & $1$ & $2$ & $1$ \\
$a$   & $b$   & $b$   & & $3$ & $1$ & $2$ \\
$c$   & $c$   & $c$   & & $2$ & $3$ & $3$ \\
\end{tabular}
\caption{DA versus TTC}
\label{fig:long}
\end{figure}

DA produces an assignment $\mu_{DA}=(1a, 2b, 3c)$. This assignment is not Pareto efficient since 1 and 2 would be strictly better off, with 3 remaining the same, if 1 and 2 swapped their assignments. Indeed, the resulting ``efficient'' assignment $\mu_{TTC}=(1b, 2a, 3c)$ is attained by TTC, where a cycle $1\to b\to 2\to a\to 1$ is formed and cleared in the first round. Note that this assignment is unstable: 3 has justified envy toward 2.

This example illustrates the trade-off between stability and efficiency. While the objectives may conflict with each other, if a school district seeks either objective, there is a case for choosing between DA and TTC. The next result shows that student-proposing DA implements a stable matching that is most efficient among all stable matchings in the following sense: 

\begin{theorem} [\citealp{Gale_Shapley(1962)AmMathMon}] \label{thm:student-optimal}
DA assignment is student-optimal in the sense that it Pareto dominates all other stable assignments for the students.\footnote{See the proof in Chapter~1.}
\end{theorem}

This means that DA achieves stability with minimal efficiency loss. The intuition behind this result is explored in Chapter~1, but roughly, it is related to the fact that the set of stable matchings forms a complete lattice in terms of relative welfare of students versus school priority orderings, and the student-optimal DA selects the extremal stable matching in favor of students. 

How well does TTC perform in minimizing the stability loss? The next result establishes an analogous result to \Cref{thm:student-optimal} for TTC in this regard, but in a much weaker and nuanced sense:\footnote{\citet{Dogan_Ehlers(2022)AEJ:Micro} show the robustness of (i) to several measures of relative stability.}

\begin{theorem}[\citealp{Abdulkadiroglu_et_al(2020)AER:Insights}]
\label{thm:ttc}
\begin{description}
    \item[(i)] In the one-to-one matching setting (i.e., schools have unit quota), TTC is ``justified envy minimal:'' no other efficient and strategy-proof mechanism entails ``fewer'' incidences of justified envy, in the set inclusion sense.\footnote{One may be interested in dropping strategy-proofness, namely a mechanism that is efficient and justified envy undominated (in the set containment sense). \citet{Abdulkadiroglu_et_al(2020)AER:Insights} show that such a mechanism is NP-hard to find.}
    \item[(ii)] In the general setting with multiple school quotas, when students' preferences and priorities are generated in the i.i.d.\ fashion, TTC is less likely to create justified envy than random serial dictatorship (RSD).\footnote{See Chapter~1 for the exact definition.}
\end{description}
\end{theorem}

\begin{proof} We only provide the proof for~(i).
Let $\varphi $ be a Pareto efficient and strategy-proof mechanism. We show that if $\varphi $ admits less justified envy (in the sense of admitting a subset of blocking pairs) than $\TTC$ at any priority orderings $\succ$, then $\varphi (P)=\TTC(P)$ for all $P$, where $\TTC(P)$ is the TTC assignment given preferences $P$ and priorities $\succ$. (We suppress the dependence of the matching on $\succ$.) Suppose, to the contrary, there exists $P$ such that
$\varphi(P)\neq \TTC(P).$
Let $I_{k}(P)$ be the set of agents who are matched in step $k$ of $\TTC(P)$, and let $\ell$ be the smallest $k$ such that $I_k(P)$ contains an agent who is assigned differently between $\TTC$ and $\varphi$.
By definition, for some $i\in I_{\ell}(P)$, $\varphi (P)(i)\neq \TTC(P)(i)$.
Let $c=\{s_{k},i_{k}\}_{k=1,\dots,K}$ be the cycle such that $s_k$ points to $i_k$, and $i_k$ points to $s_{k+1}$, with $s_{K+1}\equiv s_1$. Suppose $i$ is matched with $\TTC(P )(i)$ and $i=i_{K}$.

Consider the preference relation:
$P_{i_{K}}^{\prime }:s_{1},s_{K},..$.
Since we have only altered the preferences of $i_{K}$ in $c$ and $\TTC(P)(i_{K})=s_{1}$ has become her first choice, the $\TTC$ matching remains the same: $\TTC(P_{i_{K}}^{\prime },P_{-{i_{K}}} )=\TTC(P )$. 
Since
$\varphi (P)(i_{K})\neq \TTC(P)(i_{K}) =s_{1},$ and $\varphi (P)(i_{K})$ is still available at step $\ell $, we obtain
\begin{align}
\label{eq:sp1}
s_{1}P_{i_{K}}\varphi(P)(i_{K})\text{.}
\end{align}
But since $\TTC(P_{i_{K}}^{\prime },P_{-\{ i_{K}\}})(i_{K})=s_{1}P_{i_{K}}^{\prime }s_{K}$, $(i_{K},s_{K})$ does not block $\TTC(P_{i_{K}}^{\prime },P_{-\{i_{K}\}} )$.

Since $\varphi $ has less justified envy than $\TTC$ at $\succ $, $
(i_{K},s_{K})$ should not block $\varphi (P_{i_{K}}^{\prime
},P_{-\{i_{K}\}} )$. Since $s_{K}$ points to $i_{K}$ in cycle $c$, $%
i_{K}\succ _{s_{K}}j$ for all $j\neq i_{K}$ who are still unassigned at step
$\ell $. Given that $s_{K}$ will be assigned one of the agents still unassigned at step
$\ell $, by construction, we must then have
$
\varphi (P_{i_{K}}^{\prime },P_{-\{i_{K}\}} )(i_{K})\in
\{s_{1},s_{K}\}.
$
Given \Cref{eq:sp1}, strategy-proofness of $\varphi $ implies
$
\varphi (P_{i_{K}}^{\prime },P_{-\{i_{K}\}} )(i_{K})=s_{K}
$;
otherwise $i_{K}$ would be able to manipulate $\varphi $ in economy $%
(P )$ by submitting $P_{i_{K}}^{\prime }$ to obtain $s_{1}$. Hence, $%
i_{K}$ obtains her second choice under $P_{i_{K}}^{\prime }$.
Now, we have $\varphi (P_{i_{K}}^{\prime },P_{-\{i_{K}\}} )(i_{K-1})\neq
s_{K}=\TTC(P_{i_{K}}^{\prime },P_{-\{i_{K}\}} )(i_{K-1})$, and so we can apply
the same argument for $i_{K-1}$.

By iterating the argument for $i_{K-1},\dots,i_1$, and for every
agent in the cycle $c$, we obtain that
$
\varphi(P_{c}^{\prime},P_{-c})(i_{k})=s_{k},
$
where $P_{c}^{\prime}=\{P_{i_{k}}^{\prime}\}_{i_{k}\in c}$ and $%
P_{i_{k}}^{\prime }:s_{k+1},s_{k},\dots$. Applying
the argument iteratively, we find $\varphi (P_{c}^{\prime },P_{-c}
)(i_{k})=s_{k}$ for all $k$. But this contradicts the Pareto efficiency of $%
\varphi $ because every agent in the cycle will be better off if every $%
i_{k} $ is matched with $s_{k+1}$ (modulo $k$) without changing the matching
of other agents, establishing the claim.
\end{proof}

The sense in which TTC performs also ``well'' in stability is weak. First, the constrained-stability result holds only in the one-to-one matching setting. Even in this case, DA's constrained efficiency holds more powerfully. Second, even for the general matching setting, TTC's stability performance could be shown only relatively better than RSD. Still, these results are meaningful since RSD has been the standard method for implementing an efficient assignment. 
 
The logic of why TTC may do better in terms of stability than other efficient mechanisms, such as RSD, is that TTC prioritizes the assignment of students based on their priority standing: those with high priorities at schools that are highly demanded are more likely to be part of the cycles at the early rounds, enabling them to claim their preferred schools. If their preferred schools happen to be those they have high priorities for, then no other students will have justifiable envy of those students. By contrast, RSD ignores students' priorities in its assignment, so it will entail justifiable envy in this case.

To illustrate, consider \Cref{fig:short}, which has the same  priorities as \Cref{fig:long} but different student preferences. 
\begin{figure}[!ht]
\centering
\renewcommand{\arraystretch}{1.22}
\begin{tabular}{c|c|c p{1cm} c|c|c}
\multicolumn{3}{c}{} & & \multicolumn{3}{c}{} \\[-1.5ex]
$P_1$ & $P_2$ & $P_3$ & & $\succ_a$ & $\succ_b$ & $\succ_c$ \\
\cline{1-3} \cline{5-7}
$a$   & $b$   & $a$   & & $1$ & $2$ & $1$ \\
$b$   & $a$   & $b$   & & $3$ & $1$ & $2$ \\
$c$   & $c$   & $c$   & & $2$ & $3$ & $3$ \\
\end{tabular}
\caption{TTC Admits Short Cycles}
\label{fig:short}
\end{figure}

In this case, TTC yields $(1a, 2b, 3c)$, since $1$ and $a$ form a cycle of length~1, or a \emph{short cycle}, and likewise 2 and $b$ form a short cycle. This means that the assignment is not only efficient but also stable, admitting no justified envy.
 
However, suppose the preferences are as in \Cref{fig:long}. Then, recall that the assignment $\mu_{TTC}=(1b, 2a, 3c)$ resulted from clearing a cycle of length 2, or a \emph{long cycle}, $1\to b\to 2\to a\to 1$, in the first round. In this case,  3 has justifiable envy for 2. In fact, whenever a student obtains her assignment via a long cycle, she, \emph{a priori}, has no reason to have a higher priority at the assigned school than any other student, including those who may envy her, just as in the case of RSD. The logic of \citet{Abdulkadiroglu_et_al(2020)AER:Insights} is that TTC may reduce justified envy relative to the RSD insofar as students are assigned via short rather than long cycles. 
 
Still, \Cref{thm:ttc} provides some justification for the use of TTC in case the policymaker cares about efficiency. In practice, however, when centralized school choice assignment has been adopted worldwide, most of the school systems have adopted DA. The only exception is the New Orleans School System, which had used TTC for a year before switching to DA.

One reason for the infrequent use of TTC is that it is difficult for school systems to explain to parents.\footnote{See \citet{Leshno_Lo(2021)Restud} for some recent attempts to make the mechanism explainable.} The other issue is that parents may not tolerate violations of their priorities that occur under TTC. Despite \Cref{thm:ttc}, there are two reasons why the latter loss is significant.

\citet{Che_Tercieux(2024)WP} suggests that the stability benefit of TTC may be small as the market size grows large. Consider the one-to-one matching in which students' preferences and priorities are drawn identically and independently. \Cref{fig:TTC-large}, from Che and Tercieux, shows that as the market size $n$ (i.e., the number of students, which equals the number of schools) increases, the proportion of incidences of justified envy among the incidences of envy under TTC converges to that under RSD.\footnote{Incidences of envy are the (random) set of students and the (random) set of schools for each student in the set such that each student prefers the set of schools over her assigned school. Note that the incidences of envy are the same between TTC and RSD, due to the equivalence result by \citet{Abdulkadiroglu_Sonmez(1998)ECTA} and \citet{Knuth(1996)JoA}. The normalization serves as a scaling device that keeps the equivalence nontrivial. }
\begin{figure}
\centering
\includegraphics[width=0.8\linewidth]{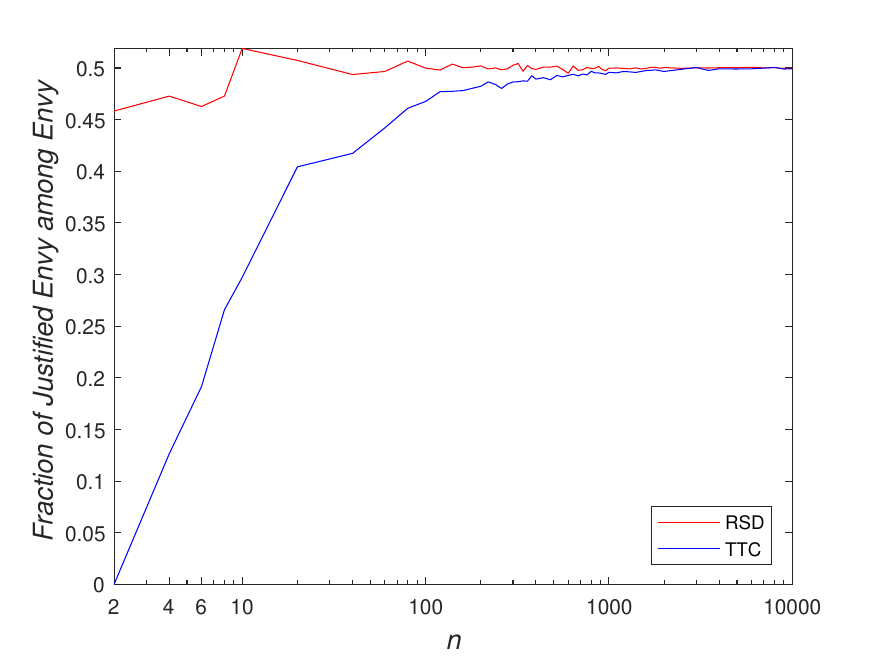}
\caption{The Fraction of Incidences of Justified Envy out of all Incidences of Envy}
\label{fig:TTC-large}
\begin{threeparttable}
\begin{tablenotes}
\item \emph{Source:} \citet[Figure~5]{Che_Tercieux(2024)WP}.
\end{tablenotes} 
\end{threeparttable}
\end{figure}
The reason for this is that as the market size $n$ grows, the fraction of the students assigned via long cycles converges to~1.

Second, \citet{Che_Tercieux(2019)JPE} consider whether the trade-off between DA and TTC vanishes in a large market if one considers approximate versions of stability and efficiency. They consider an environment in which students have cardinal preferences over schools that consist of common and idiosyncratic terms of payoffs for schools, and schools have cardinal payoffs corresponding to their priorities. They then define an assignment to be \textbf{asymptotically efficient} if, for any $\epsilon>0$, the proportion of students who may benefit more than $\epsilon$ from a Pareto-improving reassignment vanishes in probability as $n\to \infty$, and \textbf{asymptotically stable} if the proportion of ``blocks'' that benefit the deviating student and school by more than $\epsilon$, out of all possible blocks, vanishes in probability as $n\to \infty$. When common payoff terms are important, more precisely in the sense that some schools are uniformly better than others, regardless of the idiosyncratic payoffs, they show that the TTC fails to be asymptotically stable, suggesting that the stability loss remains significant (see Theorem~8 of Chapter~2, and Theorem~2 of \citealp{Che_Tercieux(2019)JPE}).

Combined, these two results suggest that as long as respecting priorities remains an important policy objective, the TTC remains problematic, possibly explaining the infrequent use of the TTC (or RSD).

\paragraph{Compromise between Efficiency and Stability.}

\citet{Che_Tercieux(2019)JPE} also show that when common payoff terms are important in the sense mentioned above, DA fails to be asymptotically efficient, meaning the efficiency loss of DA does not vanish as the market grows large. Combining with the asymptotic instability of TTC, this suggests that if the policymaker cares about both efficiency and stability/no justified envy, the two well-known mechanisms may entail significant drawbacks. It is also not clear why one should satisfy one desideratum exactly with the minimal loss of the other. In principle, one may prefer to achieve both desiderata approximately. \citet{Che_Tercieux(2019)JPE} show that such a compromise is indeed possible. 

To explain their proposed solution, it is helpful to understand why DA fails to achieve asymptotic efficiency. Whenever students compete for limited seats at a school, they are rationed based on their priorities at the school, independently of their relative preferences about the school. For example, it is plausible that a student who ranks the school at the bottom of her preference orderings may win against the student for whom the school is most preferred. In other words, whenever students have conflicting preferences and compete for a popular school, the conflicts are resolved in a way that is likely to undermine their welfare. 

Their proposal, called \textbf{Deferred Acceptance with Circuit Breaker (DACB)}, modifies DA so that the tentative assignments are finalized periodically whenever some student reaches a pre-specified cutoff number of applications for the first time. The process then repeats with the application counter refreshed after each assignment. See \citet{Che_Tercieux(2019)JPE} and Chapter~2 for the precise definition of DACB. They show that with the cutoff number of applications suitably calibrated, DACB can achieve asymptotic efficiency and asymptotic stability. While the mechanism is not strategy-proof, it is asymptotically incentive compatible---namely, truth-telling is an $\epsilon$-Bayesian Nash equilibrium: for each $\epsilon>0$, truth-telling gives a student within $\epsilon$ of the highest possible payoff given all others report truthfully if the number~$n$ of students is sufficiently large. A more detailed presentation of this result can be found in Chapter~2. 
 
\subsection{Coarse Priorities and Indifferences}
\label{subsec:coarse_priorities}

In real-world school choice scenarios, school priorities are often coarse, meaning they do not strictly rank all students. For example, in the Boston Public Schools (BPS) system, priorities are based on factors like sibling attendance, walk zone, and whether the student currently attends the school. This coarseness leads to many ties in priority rankings. Formally, we extend the priority rankings $\succ$ to $\succsim$, which allows for possible ties.

In practice, the assignment algorithm requires strict rankings, which are obtained by breaking ties in various ways. There are two common methods for tie-breaking:
\begin{itemize} 
\item Single Tie-Breaking (STB): A single lottery is used to create a priority ranking of students that applies to all schools.
\item Multiple Tie-Breaking (MTB): A separate lottery is used for each school, resulting in independent rankings across schools.
\end{itemize}
It is well known that, with indifference, a student-optimal stable matching---a stable matching that Pareto dominates all other stable matchings---may not be well defined. Hence, a more reasonable objective is to achieve a \textbf{constrained-efficient} stable matching, defined as a stable matching that is Pareto-undominated by other stable matchings. Unfortunately, this more reasonable goal may not be achieved by DA with standard tie-breaking rules such as MTB and STB.

To illustrate the inefficiency of MTB, consider the example depicted in \Cref{fig:long}, except that schools have no priorities, meaning schools are indifferent over all students. Now, suppose the MTB were used to break ties. Then, the priorities $\succ$ in \Cref{fig:long} could be realized with positive probability (more precisely, equal to $1/27$). In that case, a matching $\mu_{MTB}=(1a, 2b, 3c)$ would be realized, which we know is Pareto-dominated by $\mu^*=(1b, 2a, 3c)$. Moreover, the latter matching is stable. In other words, the MTB yields a stable matching that is not constrained-efficient, or \emph{Pareto-undominated by other stable matchings}.

Suppose instead that STB were used to break ties. Then, effectively, the mechanism becomes the random serial dictatorship, which we know would be efficient. Further, given that there are no priorities, the efficient assignment would be stable. Hence, in this case, DA-STB implements a student-optimal stable matching. Indeed, if schools have no priorities, then DA-STB reduces to RSD, so it produces an efficient and student-optimal stable assignment.

More generally, however, even DA-STB cannot consistently select a constrained-efficient stable matching. To illustrate this, suppose the students have the same preferences and priorities as in \Cref{fig:long}, except that school~$a$ is indifferent between $3$ and $2$. In this case, a STB may yield the priorities orderings precisely as in \Cref{fig:long} with probability $1/2$, so that the assignment $\mu_{STB}=(1a, 2b, 3c)$ arises, which again is Pareto-dominated by $\mu^*=(1b, 2a, 3c)$. Importantly, $\mu^*$ is stable, so again $\mu_{STB}$ fails to be constrained-efficient. 

These examples illustrate the possibility that DA with a standard tie-breaking rule leaves room for Pareto improvement without sacrificing stability. We shall discuss how this can be achieved shortly. Before proceeding, however, it is important to note that such an improvement cannot be achieved without a price. The next theorem shows that any such Pareto improvement can only be achieved at the expense of strategy-proofness.

\begin{theorem}[\citealp{Abdulkadiroglu_Pathak_Roth(2009)AER}]
\label{thm:apr}
Assume that all students are acceptable to schools. For any tie-breaking rule, there is no strategy-proof mechanism that Pareto dominates DA with that tie-breaking rule.
\end{theorem} 
\begin{proof} 
Consider any tie-breaking rule $\tau$. We fix schools' weak priority orderings $\succsim$ and an arbitrary tie-breaking rule $\tau$ that completes $\succsim$ to a strict order $\succ^{\tau}$. One can show, and we take as given, the fact: If $\nu$ Pareto dominates $\mu = \DA^{\tau}(P)$ for a given tie-breaking rule $\tau$, then the same set of students is matched in both $\nu$ and $\mu$.

Suppose that there exists a strategy-proof mechanism $\varphi$ and tie-breaking rule $\tau$ such that $\varphi$ Pareto dominates $\DA^{\tau}$. There exists a profile $P$ such that
$\varphi_{i}(P) R_i \DA_{i}^{\tau}(P)$ for all $i$, and 
$\varphi_{j}(P) P_j \DA_{j}^{\tau}(P)$ for some $j$, where $aR_i b$ means either $aP_i b$ or $a=b$.
We will say that the matching $\varphi(P)$ Pareto dominates the matching $\DA^{\tau}(P)$, where $\DA^{\tau}(P)$ denotes the student-optimal stable matching for $(P; \succ^{\tau})$. Let $s_i = \DA_{i}^{\tau}(P)$ and $\hat{s}_i = \varphi_i(P)$ denote $i$'s assignment under $\DA^{\tau}(P)$ and $\varphi(P)$ respectively, where $\hat{s}_i P_i s_i$.

Consider profile $P' = (P_i', P_{-i})$ where $P_i'$ ranks $\hat{s}_i$ as the only acceptable school. Since $\DA^{\tau}$ is strategy-proof, $s_i = \DA_{i}^{\tau}(P) R_i \DA_{i}^{\tau}(P')$, and since $\DA_{i}^{\tau}(P')$ is either $\hat{s}_i$ or $i$, we conclude that $\DA_{i}^{\tau}(P') = i$. Then, since the set of assigned students is identical between the two mechanisms, we must have $\varphi_i(P') = i$. Suppose $P'$ are the actual preferences. Then, $i$ could state $P_i$ and be matched to $\varphi_i(P) = \hat{s}_i$ which under $P_i'$ she prefers to $\varphi_i(P') = i$. This shows that $\varphi$ is not strategy-proof, a contradiction.
\end{proof}
 
The theorem implies that there is no strategy-proof mechanism that delivers a Pareto improvement over a DA assignment with any exogenous tie-breaking rule. Put differently, if a school choice mechanism aims to improve upon the DA algorithm with a given tie-breaking rule by making some students better off without harming others, it cannot be strategy-proof. This highlights the inherent trade-off between strategy-proofness and potential welfare performance in school choice.

Indeed, the theorem applies broadly to any mechanism that may or may not yield a stable outcome. For example, recall that TTC and Serial Dictatorship are Pareto efficient and strategy-proof. The theorem implies that neither mechanism can Pareto dominate DA when all students are acceptable to all schools.

These findings underscore the importance of DA as a benchmark for evaluating other school choice mechanisms. While DA may not always produce the most efficient outcome, it offers a balance between strategy-proofness, fairness, and efficiency that is difficult to surpass without sacrificing one of these desirable properties.

\paragraph{Stable Improvement Cycles Algorithm.}

The discussion so far has focused on strategy-proof methods for breaking ties in priorities, where students have no incentive to misrepresent their preferences. However, in some contexts, ensuring truthful revelation may not be the primary concern. For instance, if there is a high degree of trust between students and the school district, or if students are well-informed about school priorities and the consequences of different choices, strategic behavior may be less prevalent.

In such scenarios, it may be possible to achieve improvements over the DA algorithm while maintaining stability. One approach is the \textbf{Stable Improvement Cycles (SIC)} algorithm.

The SIC algorithm starts with any stable matching $\mu$. For example, one can first run the student-proposing DA with some exogenous tie-breaking rule and find $\mu$. We say that student~$i$ \textbf{desires $s$ at $\mu$} if she prefers school $s$ to her assignment $\mu(i)$, and let $D_s$ be the set of highest $\succsim_s$-priority students among those who desire $s$. We can then form a \textbf {stable improvement cycle} that consists of distinct students $i_1,\ldots,i_n\equiv i_0$ ($n\geq 2$) such that each student~$i_l$, for $l=0,\ldots, n-1$, is in $D_{\mu(i_{l+1})}$.
 
Given a stable improvement cycle, define a new matching $\mu'$ by:
\begin{align*}
\mu'(j)=\left\{
\begin{array}{ll}
\mu(j) & \mbox{if } j \notin \{i_1,\ldots,i_n\} \\
\mu(i_{l+1}) & \mbox{if } j=i_l\\
\end{array}
\right.
\end{align*}
Observe that matching $\mu'$ continues to be stable and it Pareto dominates $\mu$.

\begin{theorem}[\citealp{Erdil_Ergin(2008)AER}]
\label{thm:erdil-ergin}
Fix $(P,\succsim,q)$ and let $\mu$ be a stable matching. If
$\mu$ is Pareto-dominated by another stable matching, then $\mu$
admits a stable improvement cycle.
\end{theorem}
 
This theorem implies a method for finding a constrained-efficient stable matching using the SIC algorithm. Start with a stable matching. Repeatedly search for and execute stable improvement cycles. Terminate when no further stable improvement cycles are possible. The algorithm guarantees that we will eventually reach a constrained-efficient matching. This means that starting from a stable matching, by successive execution of stable improvement cycles, we will reach a point where no further Pareto improvements are possible while maintaining stability.

In light of \Cref{thm:apr}, one would expect that the SIC algorithm is not strategy-proof. This is indeed the case; a student may profitably misrepresent his/her preferences to receive a school that he/she may not like in the first-round DA in order to trade for a better school from the SIC round. Despite this problem, \citet{Erdil_Ergin(2008)AER} argue that the gains from misrepresenting preferences are rather limited: for any school one may receive from a misreported preference, there is some probability that one could get an even better school by truth-telling.

\paragraph{Alternative tie-breaking rules in the NYC school assignment.}

Summarizing the results so far, we have shown that DA with standard tie-breaking rules such as STB and MTB does not yield a constrained-efficient stable matching. Nevertheless, there is a sense in which DA-STB delivers a better efficiency outcome than DA-MTB. \citet{Abdulkadiroglu_Che_Yasuda(2015)AEJ:Micro} and \citet{Allman_Ashlagi_Nikzad(2023)TE} show that DA-STB allocates \emph{popular} schools---popular in the sense that the numbers of students top-ranking them exceed their capacities---ex-ante efficiently, whereas DA-MTB does not, in a continuum economy with no priorities. We have also shown that the SIC algorithm can produce a constrained-efficient stable matching although at the expense of strategy-proofness.

Can we quantify the efficiency losses arising from DA-STB and DA-MTB in the field? The question is ultimately an empirical one. Table~1 in \citet{Abdulkadiroglu_Pathak_Roth(2009)AER} provides some ideas about the magnitude of these losses. In line with the intuition provided above, it shows that DA-STB tends to deliver better outcomes than DA-MTB: for each $ k \le 4$, the number of students who receive $k$-th or better options is larger in the former than in the latter. Perhaps consistent with this outcome, the NYC school system adopted DA-STB rather than DA-MTB. Still, this comparison does not extend to those assigned to lower-ranked schools or unassigned.\footnote{\citet{Arnosti(2023)MNSC} provides a formal result that matches the observed single-crossing property.} The same table also suggests that the improvement from SIC on DA-STB is rather modest, particularly for the price of sacrificing strategy-proofness.

\subsection{Cardinal Welfare in School Choice}
\label{subsec:cardinal_welfare}

The previous sections focused on students' ordinal preferences for welfare evaluation, without considering their preference intensities. In practice, however, students may differ in their preference intensities even among schools they rank identically, and such differences may be relevant for evaluating the overall student welfare resulting from an assignment.

To illustrate this point, consider an example with three students $\mathcal{I}=\{1,2,3\}$ and three schools $\mathcal{J}=\{a,b,c\}$, each with one seat. The schools have no priorities in keeping with the feature mentioned in the previous section. All students have the same ordinal rankings with $aPbPc$. But their preference intensities described by the (normalized) von~Neumann--Morgenstern (vNM) values associated with alternative assignments are:\footnote{The values are normalized to have the same total sum across students.}
\medskip
\begin{center}
\renewcommand{\arraystretch}{1.22}
\begin{tabular}
[c]{|c|c|c|c|}\hline & $u_{j}^{1}$ & $u_{j}^{2}$ &
$u_{j}^{3}$\\\hline $j=a$ & $4$ & $4$ & $3$\\\hline $j=b$ & $1$ &
$1$ & $2$\\\hline $j=c$ & $0$ & $0$ & $0$\\\hline
\end{tabular}
\end{center}
\medskip
In this scenario, any assignment of students across the schools is ex-post Pareto efficient.  However, there is a sense in which student~3 would suffer less than either student~1 or~2 would from being assigned $b$ rather than $a$. One way to operationalize such a cardinal welfare concept is by measuring students' preferences toward lotteries of alternative assignments. Indeed, lotteries are quite natural in the school choice context due to coarse priorities, as noted in the previous section.

\paragraph{IA versus DA revisited.} 

To see how the ex-ante welfare concept helps us to operationalize cardinal preferences, recall the earlier discussion about Immediate Acceptance (IA) and Deferred Acceptance (DA) algorithms. Given identical ordinal preferences, one cannot compare the two based on either ex-post Pareto efficiency or stability. Both are indistinguishable based on these criteria.

To be specific, suppose first that DA is used to assign students, with ties broken by any lotteries, e.g., STB or MTB. Since it is strategy-proof, all three students submit true (ordinal) preferences, and they will be assigned to the schools with equal probabilities. This leads to an expected utility of $EU_{1}=EU_{2}=EU_{3}=\frac{5}{3}$ for the three students. 

Now, suppose IA is used to assign students with any lotteries as a tie-breaker. No matter how student~3 ranks the schools, students~1 and~2 prefer to rank truthfully, meaning truthful reporting is a dominant strategy for these students. Given this, if student~3 ranks truthfully, then all three students are assigned with equal probabilities across the schools just like under DA, so she will get the expected utility of $5/3$. Suppose student~3 ranks $(b,a,c)$. Then, she will be assigned $b$ with probability one, and receive an expected utility of 2, so she will indeed ``manipulate'' her preference. More interestingly, such a behavior not only increases her ex-ante welfare, but it also benefits the first two students; they are now assigned between $a$ and $c$ with equal probabilities, thus enjoying the expected utility of $2$ each. Remarkably, the IA assignment Pareto dominates the DA assignment in ex-ante welfare. This result is generalized for an arbitrary number of schools and capacities.\footnote{See also \citet{Miralles(2012)JET} who establishes a similar result in a large continuum economy model.}

\begin{theorem} [\citealp{Abdulkadiroglu_Che_Yasuda(2011)AER}]
Consider a school choice model $(\mathcal{I},\mathcal{J},q)$, where students have identical ordinal preferences but have vNM utilities $(u_{s})_{s\in \mathcal{J}}$ drawn as private information according to an arbitrary distribution with finite support, and the schools have no priorities and use symmetric lotteries to break ties. In any symmetric Bayes--Nash equilibrium of the IA algorithm, each type of student is
(at least weakly) better off than she is under DA.
\end{theorem}
\begin{proof}
Let $\mathcal{J}=\{s_1,\dots, s_m\}$, with capacity $q_s$ for each $s\in \mathcal{J}$. Without loss, $\sum_s q_s\le 
|\mathcal{I}|$, the number of students; and assume that assignment at any school is better than non-assignment. Each student draws a vNM $\mathbf{u}=(u_1, \dots, u_m)$, $u_1>\dots>u_m$, from a finite support according to probability mass function $\pi(\mathbf{u})$. Under DA, each student has a probability $q_s/|\mathcal{I}|$ of getting assigned $s$. 
 	
Now consider IA. Fix any symmetric Bayes--Nash equilibrium $(\sigma, \dots, \sigma)$, in which each player employs a mixed strategy $\sigma: \mathbf{u}\mapsto \Delta(\mathcal R)$, where $\mathcal R$ is the set of all rank-order lists of $\mathcal{J}$. (Such an equilibrium is well-defined by Nash's existence theorem.) Let $x_s(\mathbf{u})$ be the probability that a student of type $\mathbf{u}$ obtains $s\in \mathcal{J}$ in that equilibrium. In equilibrium, all seats are filled; otherwise, one can deviate to increase the assignment probability at the school with empty seats and reduce the chance of non-assignment, a contradiction. Consequently, for each $s\in \mathcal{J}$, $|\mathcal{I}|\sum_{\mathbf{u}}\pi(\mathbf{u}) x_s(\mathbf{u})=q_s.$

Consider a symmetric Bayes--Nash equilibrium of IA, and a student of any type~$\mathbf{u}$. Suppose the student deviates to randomizing over the strategies other students employ in equilibrium according to their relative ``frequencies'' $\pi$. Such a deviation will yield a random assignment $(x_s(\mathbf{u}'))_{s\in \mathcal{J}}$ with probability $\pi(\mathbf{u}')$, for each $\mathbf{u}'$. Since such a deviation cannot be profitable in equilibrium, 
\begin{align*}
U^{IA}(\mathbf{u})&:=\sum_{s\in \mathcal{J}} u_s x_s(\mathbf{u}) \ge \sum_\mathbf{u} \sum_{s\in \mathcal{J}} {u}_s (\pi(\mathbf{u})x_s(\mathbf{u}))\\
& = \sum_{s\in \mathcal{J}} {u}_s \sum_\mathbf{u} (\pi(\mathbf{u})x_s(\mathbf{u})) = \sum_{s\in \mathcal{J}} {u}_s \frac{q_s}{|\mathcal{I}|}=:U^{DA}(\mathbf{u})\text{.}
\end{align*} 
Since the argument applies to all $\mathbf{u}$, the desired result is proven.
\end{proof}

This theorem suggests the potential for IA to attain a more desirable allocation than DA in ex-ante welfare. In the above example, the first two students have a higher willingness to pay for $a$ than student~3, when the ``payment'' is the risk/the probability of getting assigned $c$, the worst school. IA activates such a trading possibility, while DA does not. The reason can be traced to the fact that the students are given more ``say'' under IA than DA, about how they wish to be treated when competing for a school. While the students, using their strategic ranking of schools, can improve their odds of assignment at a school, DA does not allow for this ability; any competition among students for a school is entirely resolved by how the school ranks them, not how the students rank the school.\footnote{This difference did not go unnoticed. In the wake of the redesign of the assignment system, a parent noted: ``I'm troubled that you're considering a system that takes away the little power that parents have to prioritize... what you call this strategizing as if strategizing is a dirty word...'' (Boston Public School Hearing, 2005).} 

Recall that the latter feature is precisely what allowed DA to be strategy-proof. The theorem suggests that there may be a silver lining with a strategic mechanism such as IA in eliciting cardinal preferences and using them to achieve a desirable outcome.

There are several reasons, however, that one should not overstate the case for IA based on the preceding argument. First, the preceding analysis is limited to the special case of the identical ordinal preferences and no priorities. \citet{Troyan(2012)GEB} shows that the interim Pareto dominance of IA over DA may not hold when schools have coarse priorities, but that IA Pareto dominates all strategy-proof mechanisms in the ex-ante sense, or utilitarian welfare sense, even in the presence of coarse priorities. It remains unknown, however, to what extent a similar result holds when students have heterogeneous ordinal preferences. Second, it is unclear whether parents or students may play equilibrium strategies; namely, we do not know how closely the theoretical prediction based on equilibrium analysis describes the outcome of a strategic mechanism such as IA. Quite possibly, individuals may not understand the mechanism well and may not know how best to strategize in the face of such a mechanism. It is also conceivable that some parents are better than others at strategizing, and if this confers an advantage to the former, there may be an issue of equity.\footnote{\label{fn:example_naive}However, it is worth noting that strategically sophisticated applicants may not harm naive students. Suppose in the above example that all students are initially naive and rank $(a,b,c)$, they all enjoy the expected utility of $5/3$. Suppose now only applicant~3 becomes sophisticated and ranks $(b,a,c)$; this benefits the two naive applicants by reducing the competition for $a$, which stands in contrast to \citet{Pathak_Sonmez(2008)AER} (see \cref{fn:Pathak-Sonmez}).} Third, IA may create different kinds of inequity, conferring a strategic advantage to applicants with high outside options (e.g., private schools). To see this, consider the same example as above, except that student~3 now has an outside option worth $1.5+\epsilon$ (whereas students 1 and 2 have zero outside options). With the insurance of $1.5+\epsilon$ for an arbitrarily small $\epsilon>0$, student~3 now prefers to take a risk and rank $(a,b,c)$, and with probability $1/3$ succeeds in taking $a$ away from those who arguably value it more. See \citet{Akbarpour_et_al(2022)JPubE}, \citet{Shorrer(2019)EC19}, and \citet{Calsamiglia_Martinez-Mora_Miralles(2021)EJ} for formalizing this intuition. All in all, there are compelling reasons to weigh strategy-proofness more favorably than the potential benefit from its violation. 

Of course, this caveat raises an interesting research question: can one promote cardinal welfare better than DA with a minimal sacrifice of its desirable strategy-proofness property?

\paragraph{Choice-Augmented Deferred Acceptance.}
 
The Choice-Augmented Deferred Acceptance (CADA) mechanism, developed by \citet{Abdulkadiroglu_Che_Yasuda(2015)AEJ:Micro}, offers a method for improving cardinal welfare with a modest sacrifice on the desirable property of strategy-proofness. It does so by modifying DA to allow students to influence how they are treated in ties while maintaining truthful revelation of ordinal preferences.

CADA consists of two main steps:
\begin{enumerate}
    \item [(i)] Ordinal Preference Submission and Target School Declaration: Students submit their ordinal preferences over schools, just like in DA. In addition, they declare a ``target'' school.
    \item [(ii)] Modified DA with Targeted Tie-Breaking: The mechanism then runs a modified version of the DA algorithm. In this modified version, any ties at each school are resolved by the STB, except that the realized lottery order at each school is reordered to prioritize the students who declare that school over the others who did not.\footnote{Specifically, two random priority lists, $\mathbf{T}$ and $\mathbf{R}$, of students are generated. For each school, the students who targeted that school are ranked first, according to $\mathbf{T}$, and then those who did not are ranked next according to $\mathbf{R}$. } 
\end{enumerate}

The DA algorithm is then run based on the ordinal preferences from step~(i) and the priority lists from step~(ii). It is easy to see that the CADA is ordinally strategy-proof, meaning it is a dominant strategy for each student to rank schools truthfully. The scope for the strategization is limited to the tie-breaking part of DA, which can be influenced by the ``strategic'' choice one makes about the target schools. Since the tie-breaking is subordinated to the ``intrinsic'' priorities, CADA, just like DA, respects the true priorities of students. In short, CADA produces a stable matching with no justified envy. In particular, when the priorities are strict, CADA reduces to DA. When the priorities are coarse (as is often the case), however, CADA can elicit cardinal preferences in a way that improves ex-ante welfare.

To see this, suppose CADA is used in the above example. All students report truthfully their ordinal preference $(a,b,c)$. For targeting, students 1 and 2 will name $a$ for the target school (in a dominant strategy); and given this, student~3 will name $b$ for the target school. Then, school~$a$ will prioritize 1 and 2 over 3, so either of the first two will win (based on random tie-breaking). In the second round, when the two rejected students, who include 3, apply to school~$c$, it will prioritize student~3, resulting in that student receiving $b$, and the other getting $c$. Hence, the ex-ante efficient assignment will be obtained.

\citet{Abdulkadiroglu_Che_Yasuda(2015)AEJ:Micro} consider a model with a mass of students and a finite number of schools with masses of seats and with no priorities. They then measure the efficiency of a mechanism by the \textbf{scope of efficiency}---formally, the maximal set $\mathcal{J}'\subset \mathcal{J}$ of schools whose allocation shares among students cannot be reallocated in a way that yields a Pareto improvement. It is as if a classic market exists, and delivers Pareto efficiency, but only with shares of $\mathcal{J}'$ being allowed to be traded. They show that CADA supports a larger scope of efficiency than DA in terms of set-inclusion or set-cardinality.

\subsection{Affirmative Actions: Quotas and Reserves}
\label{subsec:quotas}

The concept of affirmative action in public school assignments has been a subject of extensive debate, reflecting societal goals to address historical inequalities and promote diversity within educational institutions. School districts often implement controlled choice programs to offer parental choice while striving to maintain a balance across various student characteristics, such as race, ethnicity, or socioeconomic status. Before the advent of school choice policies, neighborhood assignment frequently led to socioeconomically segregated areas, as wealthier families relocated to neighborhoods with their preferred schools. This left families without such means constrained to neighborhood schools, regardless of their qualities or match values. To counteract these shortcomings, controlled school choice programs have become increasingly popular. This section will explore the common methods of incorporating affirmative action—maximum quotas and reserves—into school assignment mechanisms like Deferred Acceptance (DA).

\subsubsection{Maximum Quotas}

In practice, controlled school choice programs often enforce diversity objectives by setting hard upper bounds (``ceilings'')  and/or hard lower bounds (``floors'')  for different student types. An affirmative action policy may take the form of a maximum quota against the advantaged group and a minimum quota for a protected disadvantaged group. For example, the Jefferson County School District in Kentucky requires elementary schools to allocate between 15\% and 50\% of their seats to students from the geographic area with the highest concentration of beneficiaries of affirmative action policies. Similarly, in New York City, ``Educational Option'' (EdOpt) schools have to accept students across different ability ranges: 16\% must score above grade level on the standardized English Language Arts test, 68\% at grade level, and the remaining 16\% below grade level. 

\citet{Abdulkadiroglu_Sonmez(2003)AER} were among the first to formally incorporate type-specific quotas (upper bounds) into their seminal work on school choice mechanisms. They consider maximal quotas for non-overlapping types of students. They then modify DA to respect the type-specific quotas: in each step, each school accepts (tentatively) the students of each type at most up to its maximal quota and rejects the others based on its priority, and once the quota is filled for a given type, it only replaces the students in that type as better students apply.

This maximum-quota-modified DA, labeled DA$_{\text{MaQ}}$, retains the desirable properties of DA: it is strategy-proof, meaning students have no incentive to misrepresent their preferences, and it produces a student-optimal stable matching.

However, a critical question arises: Do maximum quotas always benefit the minority groups they are intended to protect? Surprisingly, this is not necessarily the case. As highlighted by \citet{Kojima(2012)GEB}, affirmative action policies can have ``perverse consequences.'' To illustrate, consider an example from \citet{Kojima(2012)GEB}:
Suppose there are three students, $1$, $2$ and $3$, and two schools, $a$ and $b$, with capacities $q_a=2$ and $q_b=1$. Students $1$ and $2$ belong to a majority group, and $3$ belongs to a minority group.

\begin{figure}[!ht]
\centering
\renewcommand{\arraystretch}{1.22}
\begin{tabular}{c|c|c p{1cm} c|c}
\multicolumn{3}{c}{} & \multicolumn{1}{c}{} & \multicolumn{2}{c}{} \\[-1.5ex]
$P_1$ & $P_2$ & $P_3$ & & $\succ_a$ & $\succ_b$ \\
\cline{1-3} \cline{5-6}
$a$   & $a$   & $b$   & & $1$ & $2$ \\
      & $b$   & $a$   & & $3$ & $1$ \\
\multicolumn{3}{c}{}  & & $2$ & $3$ \\
\end{tabular}
\caption{The Perverse Effects of Maximum Quotas}
\label{fig:kojima}
\end{figure}
 
Suppose first there are no maximum quotas against the majority group, then DA produces a matching $\mu=(1a, 2a, 3b)$ in which all students are assigned their first-best choices.

Next, school~$a$ introduces an affirmative action policy for the minority group via a maximum quota of one against the majority group $\{1, 2\}$. DA$_{\text{MaQ}}$ then produces a matching $(1a, 2b, 3a)$. The reason is that the maximum quota on the majority group forces school~$a$ to reject 2 when both 1 and 2 apply. Student~2 then applies to $b$, and displaces 3 at school~$b$ (since 2 has a higher priority than 3 at $b$). Not only is the new matching Pareto inferior to the old matching, but the new matching is strictly worse for the minority student~3, as she prefers $b$ over $a$.
 
\subsubsection{Minority Reserves}

To circumvent the inefficiencies and perverse effects caused by rigid majority quotas, \citet{Hafalir_Yenmez_Yildirim(2013)TE} propose a different interpretation of affirmative action policies based on ``minority reserves.'' The idea of minority reserves is that schools give higher priority to minority students up to a certain ``reserve'' number, but unlike a maximum quota against the majority group, it does not cap the seats allocated to the majority group. Minority reserves are extensively used in practice.\footnote{See \url{https://www.schools.nyc.gov/enrollment/enrollment-help/meeting-student-needs/diversity-in-admissions} and \citet{Dur_Pathak_Sonmez(2020)JET} for more details of the use of minority reserves in NYC and Chicago exam schools.} 

The Deferred Acceptance (DA) algorithm adapted for minority reserves, labeled DA$_{\text{MiR}}$, works as follows:
\begin{itemize}
    \item \textbf{Step~1:} Start with the matching in which no student is matched. Each student~$i$ applies to her first-choice school. Each school~$s$ first accepts as many as $r_s^m$ minority applicants with the highest priorities if there are enough minority applicants. Then it accepts applicants with the highest priorities from the remaining applicants until its capacity is filled or there are no applicants left. Any remaining applicants are rejected by~$s$.
    \item \textbf{Step~$k$:} Start with the tentative matching obtained at the end of Step~$k-1$. Each student~$i$ who got rejected at Step~$k-1$ applies to her next-choice school. Each school~$s$ considers the new applicants and students admitted tentatively at Step~$k-1$. Among these students, school~$s$ first accepts as many as $r_s^m$ minority students with the highest priorities if there are enough minority students. Then it accepts students with the highest priorities from the remaining students. The rest of the students, if any remain, are rejected by $s$. If there are no rejections, then stop.
\end{itemize}
The algorithm terminates when no rejection occurs and the tentative matching at that step is finalized. This mechanism is designed such that schools dynamically adjust their priorities: giving the highest priority to student types who have not filled their reserve, but no such priority once their reserves are filled. This dynamic priority structure ensures that the DA algorithm with minority reserves preserves its desirable properties, including being student-optimal, fair, non-wasteful, and group strategy-proof.

It is easy to see that DA$_{\text{MiR}}$ is more flexible than DA$_{\text{MaQ}}$. Suppose the maximum quota $q_s^M$ against the majority group is replaced by a minority reserve set at $r_s^m=q_s- q_s^M$. The minority group is prioritized over the $r_s^m$ reserved seats. However, if there are not enough minority students to fill the reserves, majority students can still be admitted to occupy those reserved seats, meaning that the number of majority students admitted can exceed what would be allowed under a strict quota if there are not enough minority students who prefer that school to fill the reserved slots.

\citet{Hafalir_Yenmez_Yildirim(2013)TE} demonstrate a crucial theoretical advantage of minority reserves:
\begin{theorem}
Consider majority quotas $q^{M}_s$ and minority reserves $r^{m}_s$ such that $r^{m}_s=q_s-q^{M}_s$, for each $s\in \mathcal{J}$. DA$_{\text{MiR}}$ with reserves $(r^{m}_s)$ produces a matching that Pareto dominates DA$_{\text{MaQ}}$ matching with quotas $(q^{M}_s)$.
\end{theorem}
This theorem implies that the student-proposing DA with minority reserves (DA$_{\text{MiR}}$) weakly Pareto dominates DA with majority quotas (DA$_{\text{MaQ}}$), suggesting that switching from majority quotas to minority reserves can make all students at least weakly better off.

This result can be illustrated easily in \Cref{fig:kojima} by \citet{Kojima(2012)GEB}. Suppose school~$a$ replaces its maximum quota of 1 against students $\{1,2\}$ with a reserve of 1 for student~3. DA$_{\text{MiR}}$ then produces $(1a, 2a, 3b)$,\footnote{The algorithm ends in step~1 where both 1 and 2 are assigned school~$a$, and 3 is assigned $b$. Since no minority student applies to $a$, accepting both majority students satisfies the policy.} which benefits both majority and minority students.

Even with minority reserves, however, some minorities may be made worse off, while others remain the same. Consider the following example from \citet{Hafalir_Yenmez_Yildirim(2013)TE}, in which there are three students 1, 2, and 3, with $\{2,3\}$ being minorities, and three schools, $a, b,$ and $c$, each with one seat. See \Cref{fig:minority_reserves}.
\begin{figure}[!ht]
\centering
\renewcommand{\arraystretch}{1.22}
\begin{tabular}{c|c|c p{1cm} c|c|c}
\multicolumn{3}{c}{} & & \multicolumn{3}{c}{} \\[-1.5ex]
$P_1$ & $P_2$ & $P_3$ & & $\succ_a$ & $\succ_b$ & $\succ_c$ \\
\cline{1-3} \cline{5-7}
$a$   & $c$   & $a$   & & $1$ & $1$ & $1$ \\
$c$   & $a$   & $b$   & & $2$ & $2$ & $2$ \\
$b$   & $b$   & $c$   & & $3$ & $3$ & $3$ \\
\end{tabular}
\caption{Minority Reserves Can Hurt Minorities}
\label{fig:minority_reserves}
\end{figure}

With no affirmative action, DA yields matching $(1a, 2c, 3b)$, which is Pareto efficient. Suppose (only) school~$a$ introduces a minority reserve of 1. Then, DA$_{\text{MiR}}$ yields $(1c, 2a,3b)$, which is Pareto inferior to the DA matching without affirmative action. The intuition is that 3 now outcompetes 1 at $a$ in step~1 based on her minority status, only to be displaced by 2 later, following 2's displacement by 1 at $c$.

Despite these specific examples, simulations by \citet{Hafalir_Yenmez_Yildirim(2013)TE} indicate that the positive welfare effects of minority reserves on minority students are ``substantial,'' with improvements for up to 30\% of minority students on average. Even more significant benefits are achieved for majority students, with up to 50\% being better off under DA$_{\text{MiR}}$ compared to DA$_{\text{MaQ}}$. These simulations, which account for correlations in student preferences and school priorities, demonstrate that minority reserves, on average, make minorities better off (while potentially making majorities worse off) than no affirmative action, and that DA$_{\text{MiR}}$ benefits both minorities and majorities significantly compared to DA$_{\text{MaQ}}$. Crucially, majority quota-based mechanisms are highly sensitive to quota size, especially for majority welfare, whereas minority reserve-based mechanisms temper the adverse effects on majorities.

\citet{Ehlers_Hafalir_Yenmez_Yildirim(2014)JET} further investigate the implications of ``hard bounds'' versus ``soft bounds'' in controlled school choice.\footnote{ \citet{Echenique_Yenmez(2015)AER} provide axiomatic characterizations of choice rules based on reserves and on quotas. } They formally define ``hard bounds'' as feasibility constraints that must be met, such as fixed upper and lower bounds for student types at schools. With hard bounds, they demonstrate serious limitations: a key finding is that assignments that are both ``fair across types'' and ``non-wasteful'' may not even exist. This non-existence problem extends even to assignments that are ``fair for the same types'' (where envy is only considered between students of the same type) if non-wastefulness is also required. This impossibility arises because lower bounds can create complementarities among students for schools, a known risk factor for nonexistence.

Given these difficulties with hard bounds, the paper introduces ``soft bounds'' as a novel interpretation, where control constraints dynamically regulate school priorities. This is precisely the mechanism behind minority reserves discussed earlier, where priorities adapt to fill floors or address ceilings. Under this ``soft bounds'' view, an assignment that is fair across types and non-wasteful is guaranteed to exist. The student-proposing DA, when operating with these soft bounds, achieves a student-optimal, fair, and non-wasteful assignment, and is also group strategy-proof. While the student-proposing DA maximizes student welfare among fair and non-wasteful assignments, it may not necessarily respect the controlled choice constraints (i.e., the floors and ceilings) in some cases. In contrast, the school-proposing DA with soft bounds minimizes violations of controlled choice constraints among fair and non-wasteful assignments.

In summary, \citet{Ehlers_Hafalir_Yenmez_Yildirim(2014)JET} highlight a fundamental trade-off: hard bounds are problematic due to existence issues and lack of strategy-proofness, whereas soft bounds (like minority reserves) restore desirable properties such as fairness, non-wastefulness, and truthfulness. This implies that if policymakers are not overly paternalistic, soft-bound policies are the superior choice, even though desired diversity might not be completely guaranteed.

A crucial aspect, often implicit in the mechanisms discussed so far, is the \textit{precedence order} in which different types of seats are filled when a student is eligible for more than one. So far, our discussion has implicitly assumed a certain precedence order for filling these seats: reserved seats are processed first for the protected group. However, recent research highlights that this seemingly minor implementation detail can profoundly alter the intended outcomes of affirmative action policies \citep{Dur_Kominers_Pathak_Sonmez(2018)JPE, Dur_Pathak_Sonmez(2020)JET, Sonmez_Yenmez(2022)ECTA}.

Specifically, the order in which a school fills its seats, whether it is through open competition first or through reserved categories first, can determine who gets in and from which group. For instance, in Boston, a policy designed to set aside 50\% of seats for local ``walk-zone'' students effectively yielded almost no advantage for them in practice \citep{Dur_Kominers_Pathak_Sonmez(2018)JPE}. This was because the district's mechanism filled the reserved walk-zone seats first. While this sounds intuitive, it meant that the strongest walk-zone students took those reserved spots. When it came to filling the ``open'' seats (available to all), the remaining walk-zone students, generally having lower lottery numbers or merit among the overall applicant pool, were then outcompeted by non-walk-zone students for those general seats. Had the open seats been processed first, more walk-zone students overall would have gained admission \citep[see][]{Dur_Kominers_Pathak_Sonmez(2018)JPE}. This revelation underscores that the sequencing of seat allocation is not a neutral technicality. It can create ``statistical targeting'' effects, where policies that appear neutral on paper (like equal reserve sizes) can implicitly favor or disadvantage certain groups based on their score distributions \citep[see][]{Dur_Pathak_Sonmez(2020)JET}. This highlights a fundamental trade-off: market designers have additional leverage beyond just setting quotas or reserves; they can use precedence orders to intentionally (or unintentionally) influence which groups are ultimately admitted, impacting both diversity and merit-based goals \citep[see][]{Dur_Pathak_Sonmez(2020)JET, Sonmez_Yenmez(2022)ECTA, Abdulkadiroglu_Grigoryan(2025)WP}. In particular, \citet{Abdulkadiroglu_Grigoryan(2025)WP} provide an axiomatic foundation based on constrained efficiency for pinning down a precedence order.


\section{Methods for Preference Estimation}
\label{sec:methods}

The rapid expansion of centralized school choice systems across the globe \citep[see][]{Neilson(2024)WP} has opened up new empirical opportunities for researchers to analyze how these procedures perform, how participants behave, and how alternative assignment mechanisms affect educational outcomes. A central ingredient of this empirical agenda is the use of rich administrative data generated by centralized matching platforms---including submitted rank-order lists (ROLs), priorities, school capacities, and final assignments---to recover information about applicants' underlying preferences. Such information is essential for evaluating the performance of different mechanisms, conducting counterfactual policy simulations, and assessing the welfare and distributional impacts of school choice reforms.

However, extracting reliable preference information from observed ROLs poses significant challenges. As highlighted in \Cref{subsubsec:IA}, in strategic mechanisms such as Immediate Acceptance, submitted rankings do not necessarily reflect applicants' true preferences. Even in strategy-proof mechanisms like Deferred Acceptance, reported preferences may not be truthful: for instance, applicants might omit schools they view as unattainable. In addition, list-length constraints, which are common in many systems, further distort observed rankings when applicants strategically exclude schools not because of lack of preference but to avoid going unmatched. These features make the interpretation of reported preferences inherently non-trivial.

The stakes of misinterpreting these data are high. Preference estimates feed directly into welfare analyses and policy recommendations. Misreading why disadvantaged students avoid ranking selective schools, for example, may lead to erroneous conclusions about their preferences and the design of policies aimed at closing opportunity gaps.

This section reviews the methodological approaches that have been developed to estimate preferences from the data generated by school choice mechanisms, highlighting both foundational frameworks and more recent advances. A key feature of most centralized school choice systems is that schools do not submit preferences but instead rank students according to pre-specified priority criteria. The resulting priorities are policy parameters chosen by the designer and, crucially for empirical work, are known functions of observable characteristics. This implies that only students' preferences must be recovered, distinguishing these settings from matching markets with non-transferable utility where both sides' preferences need to be estimated \citep[e.g.,][]{Agarwal(2015)AER, Diamond_Agarwal(2017)QE, He_Magnac(2022)EJ, Friedrich_et_al(2024)NBER, He_Sinha_Sun(2024)ECTA, Klein_Aue_Ortega(2024)GEB, Ederer(2025)WP}.\footnote{See \citet{Agarwal_Somaini(2023)chapter} for a recent survey of empirical models of non-transferable utility matching.}

\citet{Agarwal_Somaini(2020)AnnuRev} provide an excellent overview of many of the methods discussed in this section. We build on their survey by integrating recent developments and extending the scope to new empirical applications. While our review is mostly centered on preference estimation in school choice settings, we also draw on studies of centralized college admissions where programs rank  students using pre-specified rules, as the methodological approaches are largely common to both contexts. Our focus is on the identification and estimation of preferences, and we do not cover the complementary literature that exploits quasi-experimental variation embedded in centralized assignment mechanisms to evaluate school effectiveness \citep[e.g.,][]{Deming(2011)QJE, Pop-Eleches_Urquiola(2013)AER, Abdulkadiroglu_et_al(2011)QJE, Abdulkadiroglu_Angrist_Pathak(2014)ECTA, Abdulkadiroglu_Angrist_Narita_Pathak(2017)ECTA, Abdulkadiroglu_Angrist_Narita_Pathak(2022)ECTA, Angrist_Hull_Pathak_Walters(2017)QJE, Angrist_Hull_Walters(2023)Handbook}. Nonetheless, we return to some of the key insights from this literature in \Cref{sec:empirical_insights} of the chapter.

We begin by distinguishing between preference estimation approaches suited to IA and other inherently manipulable mechanisms (\Cref{subsec:preference_estimation_IA}), and those suited to DA (\Cref{subsec:preference_estimation_DA}), which is strategy-proof for applicants under certain conditions.\footnote{We do not discuss preference estimation under the Top Trading Cycles mechanism. At present, there are no prominent empirical implementations of TTC, and most of the methods covered in this section are not applicable. In particular, neither the optimal portfolio approach nor the stability-based approach can be used, as the assumptions required for their validity do not hold under TTC. Although the truth-telling approach discussed in \Cref{sec:TT_approach} could, in principle, apply in this context, it rests on strong assumptions that may not be satisfied in practice. If TTC sees wider adoption, developing preference estimation methods tailored to this mechanism would be a promising direction for future work.} Although there is some methodological overlap, each setting poses distinct identification challenges, which we address separately for expositional clarity. Beyond this distinction, we emphasize recent innovations that incorporate insights from behavioral market design, allowing for belief errors, strategic mistakes, and bounded rationality in estimation. We then turn to a more detailed discussion of identification and estimation techniques across empirical approaches (\Cref{subsec:identification_estimation}). For a summary of the applications discussed throughout this section, we refer the reader to \Cref{tab:survey}, which outlines the main settings, mechanisms, and estimation approaches used in the literature. 

\subsection{Preference Estimation under Immediate Acceptance and Other Strategic Mechanisms}
\label{subsec:preference_estimation_IA}

Early contributions in the empirical market design approach to school choice have focused on the challenges of recovering preferences under strategic mechanisms---particularly Immediate Acceptance (IA)---reflecting both their widespread use and the rich theoretical literature analyzing their properties.

\paragraph{Evidence of strategic behavior under manipulable mechanisms.} A key milestone in the empirical analysis of manipulable school choice mechanisms has been the consistent finding, both in lab experiments and real-world settings, that they induce pervasive strategic behavior, often at the expense of less informed or less sophisticated agents. In a seminal study, \citet{Chen_Sonmez(2006)JET} showed that participants in a school choice experiment were significantly more likely to misrepresent their preferences under IA than under either DA or TTC. In one treatment, fewer than 14\% of participants reported their true preferences under IA, compared to 72\% under DA and 50\% under TTC. The finding that IA results in lower rates of truthful reporting compared to DA or TTC has consistently been confirmed by nearly all subsequent school choice experiments \citep[for a comprehensive survey, see][]{Hakimov_Kubler(2021)ExpEcon}.

Field studies have corroborated these experimental findings, offering both direct and indirect evidence of preference manipulation under strategic assignment mechanisms. These include analyses of reported rankings under IA or closely related mechanisms \citep[e.g.,][]{Abdulkadiroglu_Pathak_Roth_Sonmez(2006)NBER, Abdulkadiroglu_Pathak_Roth_Sonmez(2006)WP, Agarwal_Somaini(2018)ECTA} and survey data eliciting students' true preferences \citep[e.g.,][]{Budish_Cantillon(2012)AER, Kapor_Neilson_Zimmerman(2020)AER, De_Haan_et_al(2023)JPE}. For example, \citet{Agarwal_Somaini(2018)ECTA} exploit geographic discontinuities in proximity-based priorities under Cambridge Public Schools' (CPS) Controlled Choice Plan, which uses a variant of IA, to show that students adjust their rankings strategically to retain priority, providing clear evidence of manipulation.

This growing body of evidence has been instrumental in challenging common misconceptions among school district officials, who often misinterpret the high proportion of students assigned to their top-listed school under IA as a sign of genuine satisfaction. In several cities, including Boston and Amsterdam, these insights were pivotal in prompting the adoption of DA \citep{Abdulkadiroglu_Pathak_Roth_Sonmez(2006)NBER, De_Haan_et_al(2023)JPE}. From a methodological perspective, this evidence has motivated economists to develop models that incorporate strategic behavior in estimating student preferences.

\paragraph{Preference reporting as a choice over lotteries.} The now-dominant approach to preference estimation in strategic assignment mechanisms was developed by \citet{Agarwal_Somaini(2018)ECTA}. They introduce a general framework for inferring student preferences from a broad class of mechanisms and apply it empirically to the CPS Controlled Choice Plan in Cambridge. This class, referred to as \emph{Report-Specific Priority + Cutoff} (RSP+C) mechanisms, is defined by the feature that an applicant's probability of admission to a given school depends on two elements: (i)~the applicant's priority at the school (which, under IA, may depend on the submitted report), and (ii)~the school's market-clearing cutoff, which reflects its capacity constraint. Most of the manipulable and non-manipulable mechanisms analyzed in the empirical literature, including IA and DA, fall within this class, with the notable exception of TTC.\footnote{\citet{Agarwal_Somaini(2018)ECTA} note that although \citet{Leshno_Lo(2021)Restud} derive a cutoff representation for TTC, it does not belong to the class of RSP+C mechanisms.}

To model strategic behavior in RSP+C mechanisms, Agarwal and Somaini assume that applicants submit ROLs that maximize their expected utility, given their beliefs about their admission probabilities at the different schools. In this framework, applicants hold private information about their own preferences over schools, but do not observe the preferences of other participants in the market. Beliefs about assignment probabilities depend on the submitted ROL, the priority at each school, and expectations about the strategies of other applicants.

Formally, let $i \in \mathcal{I} := \{1,\dots,I\}$ index students and $j \in \mathcal{J} := \{1,\dots,J\}$ index schools. Let $u_{ij}$ denote student~$i$'s indirect utility from being assigned to school~$j$, and $t_{ij}$ represent her priority at school~$j$. Priorities may be strict, as when they are determined by an exam score, or coarse, as when students fall into broad priority classes (e.g., based on walk zones). In the latter case, a random tie-breaker, denoted by $\tau_{ij}$, is used to rank students within the same priority class. Given a utility vector $\mathbf{u}_i := (u_{i1}, \dots, u_{iJ})$ and priority profile $\mathbf{t}_i := (t_{i1}, \dots, t_{iJ})$, student~$i$'s expected utility from submitting ROL $R$ is $\mathbf{u}_{i} \cdot \mathbf{L}_{R,i}$, where $\mathbf{L}_{R,i}$ is the vector of the student's perceived assignment probabilities at each school when submitting $R$. Thus, an applicant's choice over ROLs can be viewed as a portfolio choice over lotteries: among all possible reports $R$, the student selects the one that yields the highest expected utility.

This framework implies a two-step estimation strategy: (1)~for each priority profile~$\mathbf{t}_i$ and potential ROL $R$, estimate the applicant's beliefs about the probability of admission at each school; (2)~using the assignment probabilities estimated in the first step, recover the preference parameters that rationalize the reports submitted. Further details on the identification and estimation of assignment probabilities and preference parameters within this framework are provided in \Cref{subsubsec:portfolio_choice}.

\paragraph{Beliefs about assignment probabilities.}

\citet{Agarwal_Somaini(2018)ECTA} consider several alternative assumptions regarding how agents form beliefs about their assignment probabilities. In their baseline model, applicants are assumed to hold \emph{rational expectations}: they form accurate beliefs about their admission probabilities based on their own priorities, the distribution of priorities and preferences in the applicant population, and the reporting strategies of others. Under this assumption, uncertainty in admission probabilities stems from two sources: (i)~uncertainty about the reports submitted by other participants and (ii)~randomness introduced by lottery tie-breakers when priorities are coarse.

In school choice environments, rational expectations may appear restrictive, given empirical evidence that many agents behave unsophisticatedly in strategic mechanisms. For instance, \citet{Abdulkadiroglu_Pathak_Roth_Sonmez(2006)NBER, Abdulkadiroglu_Pathak_Roth_Sonmez(2006)WP} document that under Boston's former manipulable assignment mechanism, 13\% of applicants ranked two overdemanded schools as their top two choices while having random priority at the first, and 23\% of these applicants remained unassigned after the main admission round and were subsequently placed administratively in schools with unfilled seats. While listing one overdemanded school first can be a calculated risk under IA, listing two as the top choices is a clear mistake when a student ends up unassigned and does not withdraw from the public school system. This strategy effectively wastes the second slot, reducing the chances of securing a seat at the third or any subsequent choice.\footnote{We caution the reader, however, that this is only suggestive evidence, since applicants may not have accurate beliefs about which schools are eventually overdemanded.} 

To relax the rational expectations assumption, \citet{Agarwal_Somaini(2018)ECTA} explore alternative models of belief formation: (i)~In the \emph{adaptive expectations} model, agents form expectations based on their own priorities and past-year data on submitted reports; (ii)~In the \emph{coarse beliefs} model, they ignore their own priorities---either due to lack of information or misunderstanding of the mechanism---and rely only on aggregate information such as prior-year application numbers and school capacities; (iii)~Drawing on the framework of \citet{Pathak_Sonmez(2008)AER}, the authors also introduce a model of \emph{heterogeneous sophistication}, in which the applicant pool consists of a mix of sophisticated agents (who report optimally given correct beliefs) and naive agents (who report their preferences truthfully even when doing so is suboptimal). 

An alternative and more direct way to recover preferences in IA settings, without relying on specific assumptions about belief formation, is to elicit applicants' subjective expectations. \citet{Kapor_Neilson_Zimmerman(2020)AER} develop a framework that combines survey data on students' stated preferences and beliefs about admission probabilities with observed application and enrollment decisions. This hybrid revealed-preference and stated-belief approach extends the methods of \citet{Agarwal_Somaini(2018)ECTA} by replacing rational or adaptive expectations with elicited beliefs in the utility maximization framework. Importantly, the analysis assumes that choices are optimal given subjective beliefs, so deviations from predicted behavior are attributed to belief errors rather than optimization failures. Applied to high school choice in New Haven, where the assignment mechanism closely resembles IA, the authors find that applicants often make choices based on systematically mistaken beliefs, leading to substantial welfare losses.\footnote{While IA prioritizes how highly a student ranks a school, considering priority groups only afterward, the assignment mechanism used in New Haven reverses this order, giving precedence to a student's priority group before considering the school's position in her ROL. This feature makes the mechanism close in spirit to the Choice-Augmented Deferred Acceptance (CADA) mechanism described in \Cref{subsec:cardinal_welfare}, with the difference that applicants do not declare a single ``target'' school to be prioritized among students in the same priority group at that school. Instead, the school's position in the submitted ROL is used to prioritize applicants within the same priority group.} This approach underscores the value of combining administrative and survey data to recover preferences in school choice settings where agents face strategic incentives and belief errors can significantly distort observed behavior.

\paragraph{Other approaches.} 

Contemporaneously with \citet{Agarwal_Somaini(2018)ECTA}, several alternative methods have been developed to estimate preferences using data from strategic mechanisms \citep[e.g.,][]{Hwang(2015)EAI, He(2017)WP, Calsamiglia_Fu_Guell(2020)JPE, Bayraktar_Hwang(2024)WP}. These approaches are generally more tailored to specific institutional settings and thus require some adaptation before being applied elsewhere.

Using data from Barcelona, where school assignment follows an IA mechanism with a single tie-breaking lottery, \citet{Calsamiglia_Fu_Guell(2020)JPE} develop a method suited to contexts with limited uncertainty about admission cutoffs. This situation typically arises when the number of students is large relative to the number of schools, so that aggregate uncertainty about cutoffs effectively disappears. The model assumes two types of applicants: strategic agents, who hold rational expectations and maximize expected utility based on admission probabilities, and non-strategic ones, who submit truthful rankings. Given fixed cutoffs, the optimal report for a strategic agent is derived using backward induction. Instead of evaluating all possible ROLs, the method sequentially builds the list from the bottom up, selecting at each rank~$r$ the school that maximizes expected utility, conditional on rejection from higher-ranked schools and the continuation value of subsequent choices. Preferences and the type shares are estimated by testing the optimality of submitted ROLs for strategic agents, while assuming sincere reporting for non-strategic ones.

Like \citet{Agarwal_Somaini(2018)ECTA}, this approach relies on a complete model in which an applicant's report is uniquely determined by her beliefs about admission probabilities. By contrast, \citet{Hwang(2015)EAI}, \citet{He(2017)WP}, and \citet{Bayraktar_Hwang(2024)WP} propose partial identification strategies that relax these behavioral assumptions. 

\citet{Hwang(2015)EAI} and \citet{Bayraktar_Hwang(2024)WP} allow for belief errors under two weaker conditions: (i)~agents can correctly rank schools by competitiveness, i.e., they correctly anticipate which schools will be oversubscribed or undersubscribed in each round of the IA mechanism, and they know the relative ranking of admission cutoffs, though not their exact values; (ii)~they do not rank a more competitive school above a less competitive one unless they strictly prefer it. Violating the latter condition would reduce the chance of being assigned to a more preferred option, making such a strategy suboptimal. Using data from Seoul, the authors derive moment inequalities to partially identify preferences and construct confidence regions for the model's parameters. 

\citet{He(2017)WP} develops a related approach in the context of middle school admissions in Beijing. His method allows parents to make mistakes---such as playing ``safe'' strategies too often---by allowing for heterogeneous levels of sophistication, moving beyond the dichotomy between naive and fully strategic parents. The only rationality requirement is that parents do not play dominated strategies in equilibrium, for example, ranking an unacceptable school first. At the same time, the framework allows for multiple best responses, such as ranking or excluding a school with zero admission probability. These relaxed assumptions yield bounds on choice probabilities, which in turn set-identify preference parameters. Importantly, He shows that this bounded rationality framework substantially improves fit relative to models that assume fully strategic, best-responding agents. 

\subsection{Preference Estimation under the Deferred Acceptance Mechanism}
\label{subsec:preference_estimation_DA}

At first glance, estimating preferences under the Deferred Acceptance mechanism may appear more straightforward than under IA. Because the canonical version of DA is strategy-proof for applicants, submitted ROLs are often interpreted as truthful reflections of preferences. However, a growing body of research has cautioned against taking this literal interpretation in real-world applications. In practice, truth-telling may be undermined by factors such as list-length restrictions, application costs, or behavioral frictions. In response, recent work has developed empirical methods that estimate preferences in DA settings while explicitly allowing for deviations from truthful reporting.

\subsubsection{The Truth-Telling Approach}
\label{sec:TT_approach}

A central rationale for adopting strategy-proof mechanisms like DA is to eliminate the need for applicants to strategize over their rankings, thereby simplifying the decision-making process and helping to level the playing field between strategically sophisticated agents and those who are less informed or able to game the system \citep{Pathak_Sonmez(2008)AER}. This property has led many empirical studies using data from centralized school choice or college admissions systems operating under DA to estimate preferences under the assumption that applicants submit truthful rankings \citep[e.g.,][]{Hallsten(2010)AmJSociol, Kirkeboen(2012)WP, Abdulkadiroglu_Agarwal_Pathak(2017)AER, Laverde(2024)WP, DeGroote_Fabre_Luflade_Maurel(2025)WP}. However, the validity of this assumption warrants careful scrutiny, as the notion of truth-telling admits multiple interpretations that require clarification.

\paragraph{Strict versus weak truth-telling.} 

The theoretical justification for truthful preference reporting under DA rests on relatively stringent conditions. The strategy-proofness of DA is only guaranteed if agents can rank all schools at no cost. In this case, \emph{strict truth-telling} (STT), i.e., ranking \emph{all} schools in the order of true preferences, as defined by \citet{Fack_Grenet_He(2019)AER}, is a dominant strategy. However, because STT is only a \emph{weakly} dominant strategy, multiple equilibria can arise: an applicant might achieve the same outcome by submitting a non-truthful report, such as omitting schools where she anticipates no chance of admission. For STT to be the unique equilibrium, an additional requirement is that each student has a strictly positive probability of being admitted to every school. These conditions are more likely to be met when applicants face no restrictions on the number of schools they can rank and when there is substantial uncertainty about admission probabilities, such as in systems where priorities are determined by lotteries.

In practice, however, applicants rarely rank all available options, even when permitted to do so at no cost. This raises the question of how to interpret the omission of certain schools or programs from an applicant's ROL. One approach is to introduce an outside option and reformulate the STT assumption accordingly: students are assumed to rank all schools they find acceptable (i.e., those they prefer to the outside option) and omit those they deem unacceptable. This modified version of STT can be rationalized as an equilibrium outcome when applicants are fully rational, unconstrained in list length \citep[as in the Boston pre-kindergarten setting studied by][]{Laverde(2024)WP}, and face substantial uncertainty about admission outcomes.\footnote{This formulation, however, introduces a potential issue of multiple equilibria: if applicants can always decline an unacceptable assignment, they may be indifferent about whether to include such schools in their ROL.} 

A more commonly used alternative in the literature omits the outside option and assumes that applicants are \emph{weakly truth-telling} (WTT): they rank their most preferred schools truthfully, and any unranked school is assumed to be less preferred than those included in the ROL. For instance, if there are four schools, $\{a,b,c,d\}$, and student~$i$ submits the ROL $(c, a)$, WTT dictates that her true preferences are $c\succ a \succ b,d$. WTT can thus be viewed as a truncated version of STT. While it is widely adopted in empirical work, it lacks a solid theoretical foundation, as it assumes that the number of ranked options in any ROL is exogenous to applicants’ preferences. Under this assumption,  WTT can be used to identify preferences through standard discrete choice models, which we discuss in more detail in \Cref{subsec:estimation_TT}.


\paragraph{Deviations from truth-telling in DA.} 

Despite its intuitive appeal, the WTT assumption in DA environments has been challenged in settings where applicants face list-length restrictions, incur application costs, have limited uncertainty about their admission chances, or are inattentive to mistakes that carry minor payoff consequences.

WTT is most clearly undermined in contexts where DA imposes a cap on the number of schools applicants can rank. This is a common feature in many centralized school choice and college admission systems worldwide. For example, in Paris, students applying to public high schools can rank up to 8 schools, even though the number of available options in a district typically ranges from 10 to~17. Similar list-length constraints exist in school choice systems implemented in Boston, Chicago, New York City, Finland, Ghana, Singapore, and Turkey, as well as in higher education admissions in countries like Australia, Chile, France, Norway, Spain, Tunisia, and Turkey 
(see Table~1 in \citealp{Fack_Grenet_He(2019)AER} and Figure~6 in \citealp{Neilson(2024)WP}).

When such constraints are imposed, DA ceases to be strategy-proof: students risk remaining unassigned if the mechanism exhausts the options in their submitted ROL \citep{Haeringer_Klijn(2009)JET}. In response, constrained applicants may rationally choose to ``skip'' schools with low admission probabilities in favor of safer, though less preferred, options. For instance, a student who can rank only three out of four acceptable schools, with preferences $a \succ b \succ c \succ d$, may submit the ROL $(a, c, d)$ if she anticipates a low admission probability at $b$, in order to reduce the risk of remaining unassigned in the event that $d$ is her only feasible option. \citet{Calsamiglia_Haeringer_Klijn(2010)AER} provide experimental evidence of this behavior, showing that constrained DA leads to lower rates of WTT compliance than its unconstrained counterpart.

Beyond the case of list-length restrictions, WTT can also fail when students face low uncertainty about their admission chances. This is particularly relevant in environments where schools or colleges rank applicants using a priority index, such as a test score, which is known to students when they submit their ROL. If a student expects her score to be too low for a selective school (based on past admission data, for example), she may rationally choose to ``skip the impossible'' by omitting it from her ROL, despite preferring it to some of the schools she does list. Even in the absence of explicit list-length limits, the cognitive burden of ranking many schools or the desire to avoid disappointment from likely rejections \citep{Meisner_Wangenheim(2023)JET} may discourage applicants from listing all acceptable options, thereby increasing the likelihood of WTT violations. Consistent with this pattern, \citet{Chen_Pereyra(2019)GEB} combine survey and administrative data from Mexico City's high school assignment system and find that 20\% of students who declared a selective school as their top choice did not list it first on their actual application. Since fewer than 3\% of these WTT-violating students listed the maximum 20 schools allowed, list-length constraints alone cannot explain these omissions. The authors provide suggestive evidence that applicants excluded preferred options because they believed their chances of admission were close to zero.

In light of these findings, relying on the fact that most applicants do not list the maximum number of allowed schools is not a strong argument in favor of the WTT assumption. If listing additional options is perceived as costly, students may still submit ``short lists,'' omitting programs they believe have low admission probabilities, even when those programs are preferred to their outside option or to others included in their ROL. Survey evidence supports this behavior. In Chile, for example, \citet{Larroucau_Rios(2020)WP} find that although 90\% of college applicants list fewer than the maximum number of allowed programs, only 43\% report ranking their true top choice. Similarly, \citet{Chrisander_Bjerre-Nielsen(2023)WP} show that in Denmark, 20\% of surveyed college applicants
agree with the statement that they would change their ROL if they could be admitted to
any program, with 60\% of these respondents omitting their most preferred option, even though only 2\% submit the maximum of eight choices.

Deviations from truth-telling have also been documented in unconstrained DA settings, which are strategy-proof. Evidence comes from both lab experiments \citep{Chen_Sonmez(2006)JET, Li(2017)AER} and high-stakes matching environments such as medical residency and university admissions \citep{Larroucau_Rios(2020)WP, Hassidim_Romm_Shorrer(2021)MNSC, Arteaga_et_al(2022)QJE, Artemov_Che_He(2023)JPE:Micro, Shorrer_Sovago(2023)JPE:Micro}. For instance, in the Israeli psychology graduate match, \citet{Hassidim_Romm_Shorrer(2021)MNSC} report that 19\% of applicants submitted non-truthful ROLs by either failing to list a scholarship position for a program or by ranking a non-scholarship position above its scholarship counterpart, contrary to their interests. Similar patterns are observed in Australia \citep{Artemov_Che_He(2023)JPE:Micro}, Hungary \citep{Shorrer_Sovago(2023)JPE:Micro}, and in the U.S.\ National Resident Matching Program, where \citet{Rees-Jones(2018)GEB} finds that 17\% of applicants acknowledged misrepresenting their preferences.

These widespread deviations from truthful reporting have motivated two strands of literature that seek to estimate preferences without relying on the WTT assumption: the optimal portfolio choice approach and the stability-based approach, which we discuss in turn.

\subsubsection{The Optimal Portfolio Choice Approach}
\label{subsubsec:methods_portfolio}

A first approach to addressing deviations from truthful reporting in DA settings is the framework proposed by \citet{Agarwal_Somaini(2018)ECTA}, presented in the previous section. This framework interprets each applicant's ROL as the outcome of an optimal portfolio choice over admission lotteries, conditional on preferences and beliefs about admission chances. Although originally applied to the IA mechanism, it extends to any mechanism that admits a cutoff representation, which includes DA.

\paragraph{Implementing the portfolio choice approach in DA environments.}

The optimal portfolio choice approach has been applied in DA settings where truth-telling may appear implausible because of list-length restrictions or application costs \citep[e.g.,][]{Luflade(2019)WP, Larroucau_Rios(2020)WP, Larroucau_Rios(2023)WP, Idoux(2023)WP, Vrioni(2023)WP, Agte_et_al(2024)NBER, Ajayi_Sidibe(2024)WP, Lee_Son(2024)WP, Wang_Wang_Ye(2025)NBER}. These settings are typically characterized by a large number of school options and relatively limited uncertainty about admission probabilities. For example, in New York City middle school admissions, applicants may rank up to 12 choices from among 450 programs \citep{Idoux(2023)WP}. In Chilean college admissions, students can apply to at most 10 out of 1,400 programs \citep{Larroucau_Rios(2020)WP}. In both cases, only a minority of applicants exhaust their lists (8\% in NYC, 10\% in Chile). However, these ``short lists'' cannot be readily interpreted as evidence of truthful reporting: constructing a ROL entails significant information acquisition costs, which may lead applicants to skip programs they believe are unattainable or to stop adding schools once they feel assured of admission to one of their ranked options. Such behavior weakens the plausibility of the WTT assumption.

When applying the portfolio choice approach to DA settings, researchers have developed various strategies for embedding application costs within the \citet{Agarwal_Somaini(2018)ECTA} framework. \citet{Larroucau_Rios(2023)WP} do not model application costs explicitly but assume that students include a program in their ROL only if it strictly increases their expected utility. This implies that students exclude programs with zero admission probabilities and do not rank less preferred programs below those with guaranteed admission. \citet{Idoux(2023)WP} specifies a linear application cost function, $C_{i}(R) = c_{i}|R|$, where $|R|$ denotes the number of programs listed on the applicant's ROL and $c_i$---the cost of adding a program---is applicant-specific and follows a truncated normal distribution with lower bound $c > 0$. \citet{Ajayi_Sidibe(2024)WP} introduce search frictions: applicants have imperfect information about school characteristics and incur a search cost~$c$ to learn about a school. In a related approach, \citet{Lee_Son(2024)WP} model applicants as considering only a subset of schools, with consideration sets determined by a latent variable whose distribution depends on observables, some of which are excluded from the utility function.

A common computational difficulty across these implementations stems from the fact that, regardless of the specific modeling assumptions, applying the portfolio choice approach to large-scale settings with many alternatives leads to a curse of dimensionality. In what follows, we discuss methods that have been developed to address this issue.

\paragraph{The curse of dimensionality: solutions.} 

A major challenge in applying the portfolio choice approach is that the number of possible ROLs increases exponentially with the number of available options. For instance, with just 10 options, the total number of possible ROLs approaches 9.9 million. Except in special cases, no polynomial-time algorithms or even qualitative characterizations of optimal behavior are available, making the problem computationally intractable. To address this, the literature has developed a range of solutions that either target specific admission processes or introduce behavioral assumptions about how applicants form their ROLs. 

One influential approach originates with \citet{Chade_Smith(2006)ECTA}, who show that when applicants view admission probabilities as independent across programs and application costs depend only on the number of applications submitted, the optimal portfolio problem becomes downward recursive. Under these assumptions, it can be solved in polynomial time using the Marginal Improvement Algorithm (MIA), a greedy algorithm that constructs the optimal ROL by sequentially adding the program offering the highest marginal gain in expected utility, given the current portfolio. 

\citet{Luflade(2019)WP} implements this algorithm to construct optimal ROLs in the context of Tunisian college admissions, where applicants may rank up to ten programs among more than 600 alternatives. Applicants are assumed to form expectations about admission probabilities based on their own (known) priority score and programs' past cutoffs (i.e., the adaptive expectations model), with cutoffs assumed to be independently normally distributed. In a similar vein, \citet{Larroucau_Rios(2020)WP} and \citet{Larroucau_Rios(2023)WP} apply the MIA to Chilean college admissions, where applicants observe their priority scores before submitting ROLs and uncertainty arises from stochastic cutoffs. They assume that students form rational expectations about admission probabilities and take the distributions over cutoffs to be independent across programs. \citet{Larroucau_Rios(2020)WP} show that whenever beliefs on admission probabilities can be estimated in a first stage and assumed to be independent across programs, then checking optimality requires comparing a ROL to a restricted set of ``one-shot swaps,'' which considerably reduces the dimensionality of the problem.\footnote{A \emph{one-shot swap} from a ROL $R$ is a ROL $R^{\prime}$ that differs from $R$ by a single school. For example, if there are three schools $\{a,b,c\}$, then the ROL $R=(a,b)$ has four one-shot swaps: $(a,c)$, $(b,c)$, $(c,a)$, and $(c,b)$, whereas the number of possible reports is 15.} An alternative approach is proposed by \citet{Hernandez-Chanto(2021)WP} for DA environments where applicants are ranked according to a known common priority score and face list-length restrictions. Under the assumptions that preferences are independent of priority scores and that applicants' behavior conforms to a set of rationality axioms, he develops a concatenation algorithm to sequentially infer students' underlying ordinal preferences from their observed, constrained reports.\footnote{The algorithm exploits variation across ``tiers'' of applicants, defined by how a student's priority score compares to the maximum historical cutoff in each program (e.g., tier~1 students can access all programs, tier~2 all but the most selective, and so on). By comparing reports across tiers, the algorithm constructs a set of recoverable ordinal preferences that are consistent with the specified rationality axioms.} These recovered  preferences are then used to condition the estimation of cardinal utilities in a second step, substantially reducing the dimensionality of the problem, as the optimality of each observed report only needs to be checked against the subset of ROLs that are consistent with the inferred ordinal preferences.
 
These methods, however, break down when admission probabilities are correlated across programs and applicants recognize these correlations in forming their expectations. This typically arises when programs share a lottery or rely on priority scores based on similar criteria that applicants observe only imperfectly. In such settings, rejection from one program conveys information about the likelihood of admission at others, violating the independence assumption.

\citet{Ali_Shorrer(2025)AER} propose a solution for environments where admissions are based on a ``common score,'' such as a single exam score or tie-breaking lottery, and where students must choose their portfolio before learning their own score.\footnote{ 
While admission cutoffs are assumed to be fixed and known to applicants in their core model, the results extend to cases where applicants know the relative ranking of cutoffs but not their precise values.} They show that although Chade and Smith's MIA fails in such environments, the optimal portfolio problem remains recursively solvable, and they develop an algorithm that finds the optimal portfolio in polynomial time, drawing on a recursive approach similar to that introduced by \citet{Calsamiglia_Fu_Guell(2020)JPE} for the IA mechanism (see \Cref{{subsec:preference_estimation_IA}}). \citet{Ajayi_Sidibe(2024)WP} apply this framework to Ghana's centralized high school admissions system, where students submit ROLs before sitting for a common entrance exam and assignments are made using DA.

Outside of these special cases, general polynomial-time solutions are not available. Moreover, even in settings where the algorithms of \citet{Chade_Smith(2006)ECTA} and \citet{Ali_Shorrer(2025)AER} are applicable, solving the portfolio problem remains computationally burdensome in large markets. This has led researchers to develop models of bounded rationality to simplify the decision problem. 

Using data from New York City middle school admissions, \citet{Idoux(2023)WP} develops a model in which applicants hold rational expectations over their admission probabilities but face application costs and do not optimize over the full set of ROLs. Instead, they are assumed to rely on a sequential heuristic that ignores the full impact of each choice on continuation values. This heuristic reduces the dimensionality of the problem by introducing bounds on the indirect utilities associated with schools. Similarly, \citet{Wang_Wang_Ye(2025)NBER} model applicants' construction of their ROL as a cognitively limited process, drawing on the direct cognition framework of \citet{Gabaix_et_al(2006)AER}, in which ``at each decision point, agents act as if their next search of operations were their last opportunity for search.'' In their application, applicants construct their ROL step by step, myopically choosing at each stage the program offering the highest expected utility at that point, without considering how this choice affects the value of subsequent options. This behavior tends to produce overly conservative preference submissions. Taking a different perspective, \citet{Ajayi_Sidibe(2024)WP} and \citet{Lee_Son(2024)WP} explore models of sequential search under imperfect information about program characteristics or admission probabilities: applicants gradually expand their consideration set and stop searching when the expected value of additional information falls below its cost. While these various approaches greatly enhance tractability and facilitate estimation, they rely on context-specific behavioral assumptions that may not generalize across settings. 

\subsubsection{The Stability-Based Approach}

An alternative to the truth-telling and optimal portfolio choice approaches is the stability-based approach, which offers a different framework for estimating preferences in DA settings. Rather than modeling agents' beliefs and application behavior, this approach leverages a key equilibrium property of the matching outcome under DA: stability, or the absence of justified envy.

In the context of school choice, stability implies that each student is matched to her most preferred \emph{feasible} school, where feasibility is determined by whether the student's priority exceeds the school's ex-post cutoff. For example, consider a district with four schools, $\{a,b,c,d\}$, all of which rank students based on a common priority index, such as an exam score. Suppose the cutoffs are $P_{a}=0.2$, $P_{b}=0.4$, $P_{c}=0.6$, and $P_{d}=0.8$. A student with a score of 0.7 is thus eligible for admission at schools $a$, $b$, and $c$. If she is observed to be assigned to school~$b$, the stability assumption implies that she prefers $b$ over both $a$ and $c$, i.e., $b$ is her most preferred feasible option.

Although stability is guaranteed when students truthfully rank all schools, it may still hold even when they do not. This makes it a weaker, and arguably more realistic, behavioral assumption in many settings. Importantly, because it applies to the matching outcome rather than to the ranking behavior itself, the stability assumption enables identification of student preferences from observed assignments and feasible choice sets, without relying on submitted ROLs.

\paragraph{Stability in strict-priority DA settings.} 

Early applications of the stability-based approach to preference estimation in centralized admissions include \citet{Bordon_Fu(2015)REStud} for Chilean college admissions, \citet{Burgess_et_al(2015)EJ} for school choice in England, and \citet{Akyol_Krishna(2017)EER} for Turkish high school admissions. Its theoretical foundations were subsequently formalized by \citet{Fack_Grenet_He(2019)AER} in the context of ``strict-priority'' DA settings. These are environments in which schools or colleges rank students based on a priority index, such as a test score, that is known to applicants at the time of ROL submission. In such settings, Fack et al.\ argue that stability is a more plausible and robust assumption than WTT, particularly when list-length limits or application costs are present.

The key insight from \citet{Fack_Grenet_He(2019)AER} is that in such markets, stability is more likely to be satisfied than truth-telling as the economy grows large. Building on theoretical results in \citet{Azevedo_Leshno(2016)JPE}, they show that a unique stable matching exists in the continuum economy, and that in a sequence of finite random economies with Bayes--Nash equilibrium strategies, the resulting matchings are \emph{asymptotically stable}---meaning the fraction of students assigned to their most preferred feasible schools converges to one almost surely as the market size increases (see Chapter~2 for details). The intuition behind this result is that aggregate uncertainty about cutoffs diminishes with market size. As uncertainty declines, students can more accurately identify which schools are feasible given their priority scores, and they tend to secure stable assignments by ranking and being matched to their most preferred feasible option. At the same time, when list-length restrictions or application costs are present, reduced uncertainty may increase the incentive to omit low-probability options, leading to WTT violations and weakening its empirical relevance.

A fundamental distinction between truth-telling and stability-based approaches lies in their underlying choice structure. Under the truth-telling assumption, students choose from a universal set of schools, leading to a standard discrete choice model. In contrast, the stability-based approach implies a model with \emph{personalized choice sets}, where each student is assumed to ``choose'' from the subset of schools that are feasible given her priority score and the school cutoffs. A related implication is that estimation under stability relies on observed assignments, while truth-telling models use submitted ROLs. The trade-off, then, is between an approach that accommodates non-truthful behavior but uses only final matches and feasible sets, and one that assumes truthful reporting in order to exploit the richer data on rankings.

Growing skepticism about the plausibility of truth-telling in strict-priority DA settings has led empirical work on centralized school choice and college admissions to rely increasingly on the stability assumption for preference estimation. Applications include school choice in England \citep{Burgess_et_al(2015)EJ}, Ghana \citep{Ajayi(2024)JHR}, Mexico City \citep{Bobba_Frisancho_Pariguana(2023)WP,Ngo_Dustan(2024)AEJ:App, Pariguana_Ortega-Hesles(2025)WP}, New York City \citep{Hahm_Park(2024)WP}, Paris \citep{Fack_Grenet_He(2019)AER}, Sweden \citep{Anderson_et_al(2024)WP}, and Turkey \citep{Akyol_Krishna(2017)EER}; and college admissions in Brazil \citep{Barahona_Dobbin_Otero(2023)WP}, Chile \citep{Bordon_Fu(2015)REStud, Bucarey(2018)WP, Kapor_Karnani_Neilson(2024)JPE}, Denmark \citep{Chrisander_Bjerre-Nielsen(2023)WP, Gandil(2025)WP}, Germany \citep{Grenet_He_Kubler(2022)JPE}, and Turkey \citep{Arslan(2021)EmpEcon, Akyol_Krishna_Lychagin(2024)NBER}. Beyond student assignment, the stability-based approach has also been applied to preference estimation in other education-related matching markets, including childcare allocation in Japan \citep{Kamada_Kojima_Matsushita(2025)WP} and centralized teacher assignment in France \citep{Combe_Tercieux_Terrier(2022)ReStud, Combe_et_al(2025)WP}.

\paragraph{Stability: robustness to mistakes.}

\citet{Fack_Grenet_He(2019)AER} provide a rationale for the stability-based approach in strict-priority settings where agents are fully rational. In their framework, deviations from truth-telling arise when the probability of admission to a school is too low to justify the application cost of including it in a student's ROL. However, this framework does not account for broader forms of non-truthful behavior, such as those documented in the literature on strategic mistakes in strategy-proof environments (see~\Cref{sec:TT_approach}). These mistakes typically involve the use of weakly dominated strategies, even when they result in no actual payoff losses. Examples include ``skipping'' schools with positive admission probabilities or ``flipping'' the order of schools, e.g., ranking school~$a$ above school~$b$ despite preferring $b$ to $a$.

A central empirical finding in this literature is that while mistakes are common, most have minimal payoff consequences in the sense that they rarely alter matching outcomes. This is because misranked or omitted options tend to have low admission probabilities. For example, \citet{Shorrer_Sovago(2023)JPE:Micro} show that although 17\% of college applicants in Hungary ranked a program without a scholarship above an identical program with a scholarship (which is an identifiable mistake), only 4\% of these applicants would have received a different assignment had their mistake been unilaterally corrected. 

Building on this empirical pattern, \citet{Artemov_Che_He(2023)JPE:Micro} offer a broader theoretical foundation for the stability-based approach by allowing for mistakes. They study a strategy-proof DA environment with strict, known priorities and propose a solution concept called \emph{robust equilibrium}, which relaxes Bayesian Nash equilibrium by allowing agents to make mistakes as long as the resulting payoff losses are arbitrarily small. In this framework, applicants may adopt dominated strategies, as long as the resulting payoff loss is negligible in a large economy. While universal truth-telling clearly constitutes a robust equilibrium, it is not the only one. In fact, all but a vanishing fraction of applicants may submit untruthful ROLs, such as omitting their most preferred school or misordering schools in their ROLs, as long as doing so does not affect their admission outcomes appreciably.

Despite these potentially widespread deviations from truth-telling, \citet{Artemov_Che_He(2023)JPE:Micro} show that the key predictions of stable matching under DA remain valid. Under mild regularity conditions, all robust equilibria yield a virtually unique matching outcome in a sufficiently large economy. This outcome is asymptotically stable, meaning that almost all students are matched to their most preferred feasible school, and it converges to the outcome that would arise under truth-telling. In other words, even in the presence of mistakes, the matching outcome closely approximates that generated by fully rational applicants.

These results carry important implications for preference estimation in DA settings. On the one hand, they cast doubt on the validity of approaches that rely on truth-telling as an identifying assumption, showing that such methods are vulnerable to seemingly minor strategic errors. On the other hand, they strengthen the case for the stability-based approach, which relies on the observed \emph{outcome} of DA rather than on assumptions about applicant \emph{behavior}, and is more robust to small, payoff-irrelevant mistakes.

Although Artemov et al.\ caution against relying on truth-telling for identification, they do support its use for counterfactual simulations, provided the underlying preferences have been consistently estimated. Indeed, the asymptotic stability of robust equilibria justifies an approach in which preferences are estimated using the stability-based method, and truthful reporting is assumed only at the simulation stage. Despite the presence of strategic mistakes, the resulting counterfactual assignment is well approximated by the outcome that would arise if all applicants reported truthfully. By contrast, using observed ROLs directly across regimes is not theoretically justified. If a policy reform makes previously unattainable schools feasible, applicants' ROLs would likely change in response. Artemov et al.\ illustrate this point using Monte Carlo simulations: assuming fixed ROLs across regimes leads to a misprediction rate of 40\%, whereas the stability-based approach with counterfactual truth-telling yields an error rate below 5\%.

\paragraph{Leveraging uncertainties to infer preferences.}

While \citet{Artemov_Che_He(2023)JPE:Micro} focus on strict-priority settings, \citet{Che_Hahm_He(2023)WP} extend the stability-based approach to strategy-proof environments where applicants face uncertainty about their priorities, and show how such uncertainty can be leveraged to infer students' preferences. These settings include cases where lotteries are used to break ties among applicants, or where priority scores are unknown or only imperfectly known at the time applicants submit their ROLs.\footnote{Examples include admissions systems in which applicants are ranked based on an entrance exam taken after submitting their ROL, such as New York City's Specialized High Schools \citep{Corcoran_Levin(2011)chapter} and high school choice in Ghana \citep{Ajayi_Sidibe(2024)WP}, or where exam results are disclosed only after ROL submission, as in Chinese college admissions \citep{Chen_Kesten(2017)JPE}.}

Unlike in settings with known priorities, such as those studied by \citet{Fack_Grenet_He(2019)AER} and \citet{Artemov_Che_He(2023)JPE:Micro}, priority uncertainty implies that admission probabilities remain uncertain even in large markets. Without knowing which schools are feasible, applicants must be more cautious to avoid costly mistakes. Yet uncertainty does not make all forms of misranking consequential. For instance, consider three schools, $\{a,b,c\}$, that admit students based on a common tie-breaking lottery. Suppose the cutoffs are always ordered such that $P_a > P_b > P_c$, where $P_j$ denotes the cutoff for school~$j$. A student with preferences $b \succ a \succ c$ would incur no payoff loss by submitting $(b,c)$ or $(b,c,a)$ instead of the fully truthful ranking $(b,a,c)$, since $a$ is irrelevant regardless of her lottery draw: whenever $a$ is feasible, $b$ is also feasible---and preferred. Such ``mistakes'' are payoff-irrelevant. Thus, truth-telling cannot be taken for granted even in uncertain environments, as agents may still make harmless deviations.

At the same time, uncertainty renders \emph{some} mistakes costly, thereby encouraging more accurate preference revelation. In the same example, reversing the order of $b$ and $c$ would be a payoff-relevant mistake, since there is a positive probability of encountering a lottery draw where either all three schools are feasible or only $b$ and $c$ are.\footnote{The first case occurs when the student's lottery number is above $P_a$, making all schools feasible; the second arises when the number lies between $P_b$ and $P_a$, rendering only $b$ and $c$ feasible.} In the latter case---when the draw falls between $P_b$ and $P_a$---ranking $c$ above $b$ would result in an avoidable welfare loss. To avoid such losses, the student must carefully rank $b$ ahead of $c$. Che et al.\ show how this logic can be used to infer preferences from submitted ROLs. In the example, one can conclude that $b \succ a$ and $b \succ c$ by carefully examining the structure of uncertainty the student faces. 

This approach can reveal more information than what is obtainable by invoking stability in environments without priority uncertainty, such as in \citet{Fack_Grenet_He(2019)AER}. For example, suppose that schools rank applicants based on GPA, and a student ``knows'' her GPA is above the cutoff for school~$c$ but below that of $b$. In this case, $c$ is her only feasible option. Whether she submits $(b, c)$ or $(b, c, a)$, she will be assigned to $c$, and her ROL provides no information about preferences under the stability assumption alone. By contrast, the methods developed by \citet{Che_Hahm_He(2023)WP} exploit the structure of uncertainty to recover as much information as possible on participants' preferences in DA environments. These methods are discussed in more detail in \Cref{subsec:estimation_stability}. 

\subsubsection{Undominated Strategies}
\label{subsec:incomplete_models}

A drawback of the stability-based approach is that it relies only on comparisons between ex-post feasible options, without exploiting the full content of the submitted ROL. This leads in a loss of information when estimating preferences. A natural question, then, is whether additional identifying restrictions can be derived from ROL data in DA settings. 

As noted above, \citet{Che_Hahm_He(2023)WP} show that in strategy-proof DA environments with priority uncertainty, ROLs can be used to infer preferences consistent with robust equilibrium play. However, this approach is not readily applicable to DA settings where strategy-proofness fails, such as when students face list-length constraints or application costs. 

In such cases, \citet{Fack_Grenet_He(2019)AER} show that researchers can still extract useful identifying restrictions if they are willing to assume that applicants do not play weakly dominated strategies. Indeed, \citet[][Proposition 4.2]{Haeringer_Klijn(2009)JET} demonstrate that under DA, even when list-length constraints or application costs are present, submitting a ROL in which schools are not ranked in true preference order is weakly
dominated by submitting a ROL that truthfully ranks the same schools. For example, if a student's preferences are $a \succ b \succ c \succ d$, then submitting $(b,a,c)$ is weakly dominated by $(a,b,c)$, since ranking $b$ above $a$ can only reduce the chance of being admitted to the more preferred school without improving the chances at $b$. 

The undominated-strategies assumption provides over-identifying restrictions that can complement those derived from stability.\footnote{Note that the undominated-strategies assumption does not add new identifying power in approaches based on truth-telling or optimal portfolio choice in DA settings, as these frameworks already rule out the submission of ROLs that are not true partial orders of applicants' preferences.} The benefits of combining the two assumptions can be illustrated with a stylized example. Consider a DA setting with strict priorities, where students may rank up to three schools out of four. Under stability, preferences over $a$ and $b$ are primarily estimated using students who have both schools in their feasible sets. But if all students rank both schools, the undominated-strategies assumption allows the use of the full sample of students, including those for whom $a$ or $b$ was not feasible ex-post. \Cref{subsec:estimation_undominated_strategies} details how the moment inequalities derived from undominated strategies can be combined with the moment equalities derived from stability to sharpen identification. 

Undominated strategies can also be used on their own to recover preference information in non-strategy-proof DA environments, akin to the approaches developed by \citet{Hwang(2015)EAI}, \citet{He(2017)WP}, and \citet{Bayraktar_Hwang(2024)WP} in IA settings. For example, consider a DA market where assignments are based purely on lotteries and applicants face severe list-length limits. In such contexts, neither truth-telling nor stability is likely to hold, and researchers may be reluctant to impose the behavioral restrictions embedded in optimal portfolio models. Yet the undominated-strategies assumption yields inequality restrictions that support partial information of preferences. However, as the econometric structure is incomplete \citep{Tamer(2003)Restud, Tamer(2010)AnnuRev}, the assumption does not pin down a unique ROL given utilities~$\mathbf{u}_i$ and priorities~$\mathbf{t}_i$. Instead, it generates a set of inequalities consistent with a range of preference parameters. \Cref{subsec:estimation_undominated_strategies} explains how these inequalities can be used to achieve partial identification of student preferences. Still, the power of this approach may be limited when submitted ROLs are short relative to the total number of available options, in which case the identified bounds can be too wide to be informative. 

Finally, it is important to note that undominated strategies do not nest the stability assumption. While they involve weaker behavioral assumptions than truth-telling or optimal portfolio choice, they are stronger than stability in some dimensions. In particular, they rule out certain types of payoff-irrelevant mistakes---such as arbitrarily ranking schools with negligible admission probabilities---that stability-based methods allow. The undominated strategies assumption is therefore best suited to settings without ``safety'' or ``impossible'' schools, where arbitrary rankings are less likely.

\subsection{Identification and Estimation of Preferences}
\label{subsec:identification_estimation}

This section describes how the approaches discussed in the previous sections can be implemented to estimate student preferences in both strategic and non-strategic assignment mechanisms. We start by introducing the random utility framework for representing preferences and the conditions required for identification.

\subsubsection{Specifying Student Preferences\label{sec:specification_preferences}}

\paragraph{Random utility model.}

Let $i\in \mathcal{I}$ index applicants and $j\in\mathcal{J}$ index schools or programs. A general specification of applicant~$i$'s indirect utility from being assigned to school~$j$ is given by:
\begin{align}
\label{eq:rum}
u_{ij} = V(\z_{ij},\bm{\xi}_{j},\bm{\epsilon}_{i})\text{,}
\end{align}
where $\z_{ij}$ is a vector of observable characteristics that may vary by student, school, or both; $\bm{\xi}_{j}$ captures school-level attributes observed by students but unobserved by the econometrician; and $\bm{\epsilon}_{i}$ denotes student-specific unobserved heterogeneity in preferences.

In school choice settings, $\z_{ij}$ typically includes student-level characteristics (e.g., academic achievement, socioeconomic background), school-level attributes (e.g., average test score), and student--school-specific variables (e.g., commuting distance). The vector $\bm{\xi}_{j}$ reflects school-level characteristics that are observed by students but unobserved by the econometrician, such as disciplinary climate, teacher quality, or school culture. The vector~$\bm{\epsilon}_{i}$ captures individual heterogeneity in how parents or students value both observed and unobserved school attributes.

Although fairly general, the specification in \cref{eq:rum} embeds two simplifying assumptions. First, utility is assumed to depend only on the applicant's own assignment, ruling out preferences over peers at the individual level. However, preferences for school-wide peer characteristics (e.g., racial or socioeconomic mix) can be accommodated by including lagged aggregates in $\z_{ij}$. Second, the model assumes that preferences are fixed and known to students ex ante. Yet empirical evidence suggests that this assumption is unlikely to hold in contexts where participants face a large set of alternatives and must therefore learn their preferences through a costly discovery process (e.g., \citealt{Narita(2018)WP}, \citealt{Grenet_He_Kubler(2022)JPE}). To address this limitation, an increasing number of studies adopt models that allow applicants' preferences to evolve over time through dynamic learning \citep[e.g.,][]{Bordon_Fu(2015)REStud, Narita(2018)WP, Grenet_He_Kubler(2022)JPE, Larroucau_Rios(2023)WP, Agte_et_al(2024)NBER, Hahm_Park(2024)WP, Kapor(2024)NBER, Arcidiacono_et_al(2025)JPE, DeGroote_Fabre_Luflade_Maurel(2025)WP}.

\paragraph{Non-parametric identification of preferences.}

\citet{Agarwal_Somaini(2018)ECTA} provide general identification results for the random utility model described in \cref{eq:rum}. The goal is to identify and estimate the joint distribution of the utility vector $\textbf{u}_i = (u_{i1}, \dots, u_{iJ})$ conditional on observed characteristics $\z_{i}=(\z_{i1},\dots,\z_{iJ})$ and school-specific unobservables $\bm{\xi}=(\bm{\xi}_{1},\dots,\bm{\xi}_{J})$, denoted by the conditional density $f_{U}(u_{i1},\dots,u_{iJ}\mid\z_{i},\bm{\xi})$.

The empirical approaches outlined in the previous sections yield revealed preference relations, derived either from submitted ROLs (as in the portfolio choice and truth-telling approaches) or from observed assignments (as in the stability-based approach). Under complete models of behavior, a given utility vector $\mathbf{u}_i$ maps to a unique ROL or assignment, allowing the likelihood of the observed outcome to be written as a function of the underlying distribution of preferences. For instance, under the portfolio choice approach, the probability that student~$i$ submits her observed ROL $R_{i}$ is:
\begin{align*}
\Prob(i \text{ submits } R_{i} \mid \z_{i},\bm{\xi}) = \Prob(R_{i} = \argmax_{R\in \mathcal{R}_i} \mathbf{u}_i\cdot \mathbf{L}_{R,i}\mid\z_{i},\bm{\xi};f_{U})\text{,}
\end{align*}
where $\mathbf{L}_{R,i}$ denotes student~$i$'s vector of (perceived) assignment probabilities under ROL~$R$ and $\mathcal{R}_i$ is the set of possible ROLs that student $i$ can submit. Submitting $R_{i}$ reveals that her utility vector~$\mathbf{u}_{i}$ lies within the convex set $C_{R_i}$:
\begin{align}
\label{eq:convex_subset_b}
C_{R_i} = \left\{\mathbf{u}_{i} \in \mathbb{R}^{J}: \mathbf{u}_{i}\cdot (\mathbf{L}_{R_{i},i} - \mathbf{L}_{R,i}) \geq 0 \text{ for all } \mathbf{L}_{R,i} \in \mathcal{L}_i\right\}\text{,}
\end{align}
where $\mathcal{L}_i$ denotes the set of all assignment lotteries faced by $i$. 
The subset $C_{R_i}$ is defined by linear inequalities on $\mathbf{u}_{i}$ and characterizes the set of utilities for which $R_i$ is optimal given the applicant's beliefs. Analogous restrictions on $\mathbf{u}_i$ arise under the truth-telling and stability approaches, as discussed in \Cref{subsec:estimation_TT,subsec:estimation_stability}.

\citet{Agarwal_Somaini(2018)ECTA} analyze the conditions required for non-parametric identification of the joint distribution of preferences $f_U$ in mechanisms that admit a cutoff representation (RSP+C). They show that point identification can be achieved by ``tracing out'' the distribution of utilities using two types of exogenous variation: (i)~variation in the choice environment (e.g., different priorities or cutoffs faced by otherwise similar applicants), and (ii)~variation generated by a preference shifter, i.e., a student--school-specific covariate $z_{ij}^{0}$ that enters utility additively and is orthogonal to unobserved heterogeneity.

The first source of variation applies, for example, when components of applicant priority (such as sibling or walk-zone priority) are plausibly exogenous to preferences, or when the data span multiple years with fixed school sets but varying capacities or cutoffs.\footnote{See \citet{Carvalho_Magnac_Xiong(2019)QE} for a discussion of how exogenous variation in the choice environment can be used to non-parametrically identify preferences in a college admissions context.} However, Agarwal et al.\ note that such variation is often limited in practice, particularly in settings with few priority types or few cohorts, making full non-parametric identification difficult.

Alternatively, preferences can be non-parametrically identified if the covariates $\z_{ij}$ include a ``special regressor,'' as defined in the econometrics literature \citep[see][for a review]{Lewbel(2014)chapter}. In the context of school choice, such a variable must vary at the student--school level, enter additively in the utility function, and be orthogonal to unobserved preference heterogeneity ($\bm{\epsilon}_{i}$). A commonly used special regressor in empirical applications is the distance between a student and a school. When this assumption is satisfied, \cref{eq:rum} can be rewritten as:
\begin{align}
\label{eq:rum_z}
u_{ij} = V(\widetilde{\z}_{ij},\bm{\xi}_{j},\bm{\epsilon}_{i}) - d_{ij}\text{,}
\end{align}
where $d_{ij}$ is the distance from student~$i$ to school~$j$, and $\widetilde{\z}_{i,j} := \z_{i,j}\backslash d_{ij}$. The critical assumption is that distance is orthogonal to unobserved preference heterogeneity. This may be violated if families sort into neighborhoods based on unobserved school preferences \citep[see][]{Park_Hahm(2023)WP}. The credibility of this assumption is thus context-specific and becomes more plausible when the data include rich controls to adjust for potential confounders.

\citet{Agarwal_Somaini(2018)ECTA} show that when the revealed preference sets (like $C_{R_i}$) are defined by systems of linear inequalities, variation in $d_{ij}$ non-parametrically identifies $f_U$.\footnote{See \citet[][Section~3.2]{Agarwal_Somaini(2020)AnnuRev}, for a graphical illustration of the identification argument.}

Note that preference distributions in random utility models are only identified up to an affine transformation, requiring both location and scale normalization. Location is typically normalized by setting the utility of a reference alternative---often the outside option (denoted school~0)---to zero, i.e.\ $u_{i0}=0$, anchoring all utility comparisons to that baseline. If the model does not include an outside option, the mean utility of an arbitrary school may be normalized to zero instead. For scale, identification is generally achieved by normalizing the variance of $\bm{\epsilon}_i$ or by fixing the coefficient of one observed variable. In \cref{eq:rum_z}, the scale normalization is ensured by setting the coefficient on distance to $-1$, allowing preferences to be interpreted in terms of students' willingness to travel.\footnote{While this willingness-to-travel representation aids interpretation, it does not allow for interpersonal utility comparisons. Because distance is not a transferable unit of measurement, utilitarian welfare metrics based on the Kaldor--Hicks criterion (such as total or average utility) cannot be justified. Still, ordinal comparisons (e.g., the share of students who prefer one assignment over another) remain valid. See \citet{Agarwal_Somaini(2020)AnnuRev} for further discussion.}

\paragraph{Common parametrizations.}

Although the identification results in \citet{Agarwal_Somaini(2018)ECTA} do not rely on parametric assumptions about the random utility function~$V(\cdot)$, such assumptions are typically necessary in empirical applications to ensure computational tractability.
\medskip

\noindent\emph{Multinomial logit.} A widely used parametrization in the literature \citep[e.g.,][]{Burgess_et_al(2015)EJ, Fack_Grenet_He(2019)AER, Nguyen(2021)WP, Oosterbeek_Sovago_van_der_Klauw(2021)JPubE, Grenet_He_Kubler(2022)JPE, Ngo_Dustan(2024)AEJ:App} is to specify preferences using a multinomial logit model of the form:
\begin{align}
\label{eq:logit}
u_{ij} = V(\mathbf{x}_j,\mathbf{w}_{ij},\xi_{j},\bm{\alpha},\bm{\beta})
= \underbrace{\mathbf{x}^{\prime}_{j}\bm{\alpha} + \xi_{j}}_{\textstyle \delta_j\mathstrut} + \mathbf{w}^{\prime}_{ij}\bm{\beta} + \epsilon_{ij}\text{,}
\end{align}
where $u_{i0} = \epsilon_{i0}$ and $\epsilon_{ij}$ follows an extreme-value type I (EVT1) or Gumbel distribution with location parameter~0 and scale parameter~1.\footnote{When the model includes distance to school, the scale normalization can instead be achieved by fixing the coefficient on distance to $-1$, as in \cref{eq:rum_z}. In this case, the model parameters are interpreted in terms of willingness-to-travel, and the scale of the EVT1 distribution becomes a parameter $\sigma$ to be estimated.} 

In this formulation, $\mathbf{x}_{j}$ includes observed school-specific characteristics, while $\mathbf{w}_{ij}$ captures observed characteristics that vary across both students and schools, potentially through interactions between student and school attributes. Unobserved school-specific characteristics are summarized by a single dimension, $\xi_{j}$, which reflects the overall influence of unobserved attributes that make a school more or less desirable to all applicants. Observed and unobserved school characteristics are often grouped into a single school fixed effect $\delta_j := \mathbf{x}'_j \bm{\alpha} + \xi_{j}$, which flexibly captures all systematic variation in average school desirability.\footnote{Without such fixed effects, $\xi_j$ would be absorbed into the error term, potentially biasing estimates of $\bm{\alpha}$ and $\bm{\beta}$ when $\mathbf{x}_j$ and $\mathbf{w}_{ij}$ are correlated with $\xi_{j}$.}

A major advantage of the logit model is computational: as discussed in \Cref{subsec:estimation_TT,subsec:estimation_stability}, it yields closed-form expressions for the likelihood of observed ROLs or assignments under both the truth-telling and stability-based approaches, enabling estimation via maximum likelihood. However, a key limitation is that it assumes fixed coefficients $\bm{\alpha}$ and $\bm{\beta}$ across individuals, preventing unobserved heterogeneity in how students value observed school characteristics. This leads to restrictive substitution patterns due to the independence of irrelevant alternatives (IIA) property.\footnote{For example, IIA implies that improving a positively valued attribute of one school reduces the choice probabilities of all other schools proportionally---a substitution pattern known as ``proportionate shifting'' \citep[see][]{Train(2009)book}.}

To introduce greater flexibility while preserving tractability, a growing number of studies \citep[e.g.,][]{Abdulkadiroglu_Pathak_Schellenberg_Walters(2020)AER, Barahona_Dobbin_Otero(2023)WP, Campos_Kearns(2024)QJE, Corradini(2024)WP, Laverde(2024)WP, Campos_Munoz_Bucarey_Contreras(2025)WP} estimate \cref{eq:logit} separately across subsamples defined by combinations of observable student characteristics such as gender, race/ethnicity, or academic performance. This stratified specification allows model parameters to vary freely across subgroups, thereby capturing observed heterogeneity in preferences.
\medskip

\noindent\emph{Multinomial mixed logit.} To allow for unobserved heterogeneity in how students value observed characteristics, the multinomial logit model can be extended by introducing applicant-specific random coefficients $\bm{\beta}_{i}$ on $\mathbf{w}_{ij}$, resulting in a mixed logit specification:
\begin{align}
\label{eq:mixed_logit}
u_{ij} = \delta_{j} + \mathbf{w}^{\prime}_{ij}\bm{\beta_{i}} + \epsilon_{ij}\text{,}
\end{align}
where $\bm{\beta_{i}}$ is a vector of individual-specific taste coefficients drawn from a multivariate normal distribution, i.e.\ \mbox{$\bm{\beta_{i}} \sim \mathcal{N}(\bm{\beta},\bm{\Sigma}_{\beta})$}, with $\bm{\beta}$ and $\bm{\Sigma}_{\beta}$ as parameters to be estimated, and $\epsilon_{ij}$ follows an EVT1 distribution.

This model introduces flexible substitution patterns and richer heterogeneity in preferences, but at the cost of losing closed-form expressions for choice probabilities. As a result, estimation requires multi-dimensional numerical integration. Common simulation-based methods include maximum simulated likelihood, the method of simulated moments, and simulated scores \citep[see, e.g.,][]{Hastings_Kane_Staiger(2009)WP, Luflade(2019)WP, Larroucau_Rios(2023)WP}. These methods can be computationally demanding, especially as the dimension of $\bm{\beta_{i}}$ increases.\footnote{This is particularly true for the method of maximum simulated likelihood, which may exhibit substantial asymptotic bias if the number of simulation draws per observation does not grow fast enough relative to the sample size \citep{Lee(1995)ET}.}
\medskip

\noindent\emph{Multinomial probit.} To circumvent the computational burden associated with evaluating integrals in mixed logit models, an increasingly popular alternative is to specify student preferences using a multinomial probit model estimated via Bayesian methods \citep[e.g.,][]{Abdulkadiroglu_Agarwal_Pathak(2017)AER, Agarwal_Somaini(2018)ECTA, Kapor_Neilson_Zimmerman(2020)AER, Larroucau_Rios(2020)WP, Che_Hahm_He(2023)WP, Agte_et_al(2024)NBER, Campos_Kearns(2024)QJE, Kapor_Karnani_Neilson(2024)JPE}. In this setup, both the random coefficients and the idiosyncratic preferences in \cref{eq:mixed_logit} are assumed to follow normal distributions, i.e., $\bm{\beta_{i}} \sim \mathcal{N}(\bm{\beta},\bm{\Sigma}_{\beta})$ and $\epsilon_{ij} \sim \mathcal{N}(0,\sigma_{\epsilon}^{2})$.

Probit models offer two main advantages. First, under both the truth-telling and stability assumptions, Bayesian methods provide a tractable framework for estimating random-coefficient models without requiring the simulation of choice probabilities, a feature that facilitates preference estimation in large-scale markets such as the Chilean college admissions system analyzed by \cite{Kapor_Karnani_Neilson(2024)JPE}. Second, the probit specification is particularly well-suited for the portfolio choice approach in strategic settings, where even logit models without random coefficients yield likelihood functions that lack closed-form expressions. It is also recommended for stability-based inference in DA environments with priority uncertainty \citep{Che_Hahm_He(2023)WP}, where the likelihood is typically non-analytic (see \Cref{{subsec:estimation_stability}}).

The next section illustrates how Bayesian methods can be implemented to estimate probit-based models under the portfolio choice approach.

\subsubsection{Estimation under the Portfolio Choice Approach}
\label{subsubsec:portfolio_choice}

In the portfolio choice approach, each applicant is assumed to submit the ROL~$R_{i}$ that maximizes her expected utility, given her beliefs about the assignment probabilities associated with each possible report:
\begin{align}
\label{eq:portfolio_prob}
\Prob(R_{i} = \argmax_{R\in \mathcal{R}} \mathbf{u}_i\cdot \mathbf{L}_{R,i} \mid \z_{i}; \bm{\theta})\text{,}
\end{align}
where $\bm{\theta}$ denotes the preference parameters to be estimated, and $\mathbf{L}_{R,i}$ is the vector of student~$i$'s perceived assignment probabilities across schools when she submits ROL~$R$.

\citet{Agarwal_Somaini(2018)ECTA} propose a two-step estimator for $\bm{\theta}$. In the first step, assignment probabilities $\mathbf{L}_{R}$ are estimated for each student--ROL pair. In the second step, the preference parameters $\bm{\theta}$ are estimated, treating the first-stage estimates $\hat{\mathbf{L}}_{R}$ as known.\footnote{\citet{Agarwal_Somaini(2018)ECTA} note that while joint estimation of $\mathbf{L}_{R}$ and $\bm{\theta}$ is theoretically feasible, the two-step approach is more computationally tractable, though it may entail some loss in efficiency.}

\paragraph{Step~1: Estimating assignment probabilities.} 

Under the rational expectations assumption, applicants have correct beliefs about the probability of assignment to each school, conditional on their submitted ROL. This implies that each component $L^{(j)}_{R,i}$ of the assignment probability vector $\mathbf{L}_{R,i}$ corresponds to the probability that student~$i$ is matched to school~$j$, given she submits ROL~$R$. 

In mechanisms that admit a cutoff representation (RSP+C), these probabilities depend on the student's effective priority at each school and the school's admission cutoff. The student's effective priority can be summarized by an ex-post eligibility score, denoted by $e_{ij}$, which aggregates the relevant inputs of the assignment mechanism. Under DA, $e_{ij}$ can be constructed as a lexicographic function of the student's priority~$t_{ij}$ and, in case of coarse priorities, a random tie-breaker $\tau_{ij}$. Under IA, the eligibility score $e_{ij}$ also depends on the position of school~$j$ in the student's ROL, making it endogenous to the submitted ranking. The admission cutoff at school~$j$, denoted by $p_j$, is defined as the lowest eligibility score among students assigned to the school when its capacity is exhausted, and as the lowest possible score otherwise. 

Given these definitions, the probability that student~$i$ is matched to $j$ conditional on submitting ROL~$R$ is given by: 
\begin{align}
\label{eq:l_ij}
L^{(j)}_{R,i} = \Prob\left( e_{ij} \geq p_{j} \text{ and } e_{ij^{\prime}}<p_{j'} \text{ if } j^{\prime} \text{ is ranked above } j \right)\text{,}
\end{align}
that is, the probability that student~$i$'s eligibility score exceeds the cutoff at school~$j$ and falls below the cutoffs at all schools she ranks higher. 

From the applicant's perspective, assignment uncertainty stems from two sources: (i)~the ROLs submitted by other applicants, and (ii)~random tie-breakers, when priorities are coarse. To approximate both, \citet{Agarwal_Somaini(2018)ECTA} propose a resampling procedure. This involves repeatedly drawing (with replacement) samples of students, their submitted ROLs, and lottery draws when relevant. The assignment mechanism is then simulated on each resampled dataset to compute market-clearing cutoffs. The empirical distribution of these cutoffs is used to estimate $\hat{\mathbf{L}}_{R}$ for each student and ROL, by averaging over simulations---across both cutoff realizations and, in the presence of lotteries, multiple tie-breaker draws. \citet{Agarwal_Somaini(2018)ECTA} show that this estimator is consistent and asymptotically normal.

This procedure can be extended to accommodate alternative belief formation assumptions, such as the adaptive or coarse expectations models discussed in \citet{Agarwal_Somaini(2018)ECTA}. It is also compatible with models allowing for heterogeneous degrees of sophistication in application behavior \citep[e.g.,][]{He(2017)WP, Agarwal_Somaini(2018)ECTA, Calsamiglia_Fu_Guell(2020)JPE}. Alternatively, perceived assignment probabilities $\hat{\mathbf{L}}_{R}$ can be constructed using data on applicants' stated beliefs, when available from surveys \citep{Kapor_Neilson_Zimmerman(2020)AER}.

\paragraph{Step~2: Estimating preference parameters.}

Given the estimated assignment probabilities $\hat{\mathbf{L}}_{R}$ from Step~1, the second step estimates preference parameters using \cref{eq:portfolio_prob} and~\cref{eq:l_ij}. The corresponding log-likelihood function is:
\begin{align}
\label{eq:ll_portfolio}
\ln \mathcal{L}_{\text{Portfolio}}(\bm{\theta}\mid \z) =
\sum_{i=1}^{I}\ln \Prob(R_{i} = \argmax_{R\in \mathcal{R}} \mathbf{u}_i\cdot \hat{\mathbf{L}}_{R,i}\mid\z_{i};\bm{\theta})\text{.}
\end{align}

Because this likelihood lacks a closed-form expression, estimation of $\bm{\theta}$ typically relies on simulation-based methods. Common approaches include maximum simulated likelihood \citep[e.g.,][]{Calsamiglia_Fu_Guell(2020)JPE, Vrioni(2023)WP} and the method of simulated moments \citep[e.g.,][]{Ajayi_Sidibe(2024)WP}. An alternative, increasingly used in recent work, is to adopt a probit specification and estimate $\bm{\theta}$ via Bayesian methods \citep[e.g.,][]{Agarwal_Somaini(2018)ECTA, Kapor_Neilson_Zimmerman(2020)AER, Larroucau_Rios(2020)WP, Hernandez-Chanto(2021)WP, Idoux(2023)WP, Agte_et_al(2024)NBER, Campos_Kearns(2024)QJE, Xu_Hammond(2024)EcIn}.

As outlined in \Cref{sec:specification_preferences}, the probit model assumes that both the idiosyncratic taste parameters $\bm{\beta}_{i}$ and the error terms $\epsilon_{ij}$ follow normal distributions: $\bm{\beta}_{i} \sim \mathcal{N}(\bm{\beta}, \bm{\Sigma}_{\beta})$ and $\epsilon_{ij} \sim \mathcal{N}(0,\sigma_{\epsilon}^{2})$. \citet{Abdulkadiroglu_Agarwal_Pathak(2017)AER} and \citet{Agarwal_Somaini(2018)ECTA} adapt the Gibbs sampling procedure of \citet{McCulloch_Rossi(1994)JoE} and \citet{Rossi_McCulloch_Allenby(1996)MKSC} to estimate such models in the school choice context. The algorithm specifies conjugate priors for $\bm{\theta} := (\bm{\delta}, \bm{\beta}, \bm{\Sigma}_{\beta}, \sigma^{2}_{\epsilon})$ and proceeds iteratively, starting from initial values $\bm{\theta}_0$ and a set of utilities consistent with the observed ROLs. Each iteration consists of two steps: (i)~drawing each student's latent utility vector $\mathbf{u}_{i}$ from its posterior distribution, conditional on the parameters and the revealed preference inequalities from \cref{eq:convex_subset_b}; and (ii)~drawing the parameters $\bm{\theta}$ from their posterior distribution, conditional on the sampled utilities. After discarding an initial set of ``burn-in'' draws, posterior means (as point estimates) and standard deviations (as standard errors) of the parameters can be obtained by sampling from their posterior distributions.\footnote{Priors for $\bm{\delta}$ and $\bm{\beta}$ are typically specified as Gaussian, while those for $\bm{\Sigma}_{\beta}$ and $\sigma_{\epsilon}^{2}$ are generally chosen to follow inverse-Wishart distributions.} By the Bernstein-von Mises theorem \citep[][Section~10.2]{vanDerVaart(2000)book}, these posterior means converge asymptotically to the same distribution as maximum likelihood estimators.\footnote{Since the assignment probabilities $\mathbf{L}_{R}$ are estimated in a first step, standard errors must account for the resulting estimation error. \citet{Agarwal_Somaini(2018)ECTA} address this using a bootstrap procedure described in Appendix~E.2 of their paper.}

The key advantage of the Gibbs sampler is that it sidesteps the need to compute the choice probabilities in \cref{eq:ll_portfolio} directly. Instead, it generates samples that approximate the posterior distribution of the parameters of interest. In school choice applications, the crucial insight behind step~(i) is that conditional draws of $u_{ij}$---given current values for $\mathbf{u}_{i,-j}$ and $\bm{\theta}$---are obtained from a (possibly two-sided) truncated normal distribution. This follows from the joint normality assumption and the fact that the revealed preference conditions impose the constraint $\mathbf{u}_{i} \cdot (\mathbf{L}_{R_i, i} - \mathbf{L}_{R,i}) \geq 0$ for all $R \in \mathcal{R}_{i}$, which bounds the support of each $u_{ij}$ conditional on the others.

The main computational challenge in implementing the portfolio choice approach lies in the exponential growth of the set of ROLs that must be considered for each student as the number of options increases. Whether using simulation-based or Bayesian methods, the utility vector of each student must be evaluated against all possible alternative ROLs, a task that becomes computationally prohibitive in large markets. While this challenge is manageable in small-scale settings---such as those studied by \citet{Agarwal_Somaini(2018)ECTA} and \citet{Kapor_Neilson_Zimmerman(2020)AER}---applications to more complex environments often require simplifying assumptions or approximations. Several strategies have been proposed to address this dimensionality issue (see \Cref{subsubsec:methods_portfolio}). These include assuming that applicants know their priorities and perceive admission probabilities as independent across schools or programs \citep{Luflade(2019)WP, Larroucau_Rios(2020)WP, Larroucau_Rios(2023)WP, Vrioni(2023)WP}; that applicants face limited uncertainty over admission cutoffs when priorities are based on a ``common score,'' such as a single tie-breaking lottery \citep{Calsamiglia_Fu_Guell(2020)JPE} or exam score \citep{Ajayi_Sidibe(2024)WP}; or that applicants rely on decision heuristics to simplify the application process due to bounded rationality \citep{Idoux(2023)WP, Lee_Son(2024)WP, Wang_Wang_Ye(2025)NBER}. The validity of these modeling choices should be carefully assessed in light of the institutional and informational context of each study.

\subsubsection{Estimation under Truth-telling}
\label{subsec:estimation_TT}

Empirical studies in DA environments have commonly relied on the weak truth-telling assumption to recover student preferences. WTT can be viewed as a truncated version of STT, and rests on two assumptions: (a)~the length of each submitted ROL is exogenous to student preferences, and (b)~students rank their most preferred schools truthfully while omitting less-preferred options beyond a certain point.

As an illustration, suppose a student faces four schools, $\{a,b,c,d\}$, and submits the ROL $R_{i}=(b, a)$. Under WTT, the student is assumed to have preferences $b \succ a \succ c,d$, while the choice to rank only two schools is treated as independent of her preferences over the full set of options.

To fix ideas, suppose that student preferences can be represented using a multinomial logit model without an outside option, specified as:
\begin{align*}
u_{ij} = V_{ij} + \epsilon_{ij} = V(\z_{ij}, \bm{\theta}) + \epsilon_{ij}\text{,}
\end{align*}
where $V(\cdot, \cdot)$ is a known function of $\z_{ij}$, a vector of observable student--school characteristics, and $\bm{\theta}$, a vector of parameters to be estimated. The idiosyncratic shocks $\epsilon_{ij}$ are assumed to be i.i.d.\ across students and schools, and to follow an EVT1 distribution with location parameter 0 and scale parameter 1.

Let $\sigma_i(\mathbf{u}_i, \mathbf{t}_i)$ denote student~$i$'s strategy, mapping her utility vector $\mathbf{u}_i$ and priority type $\mathbf{t}_i$ into a submitted ROL. The observed ROL is written as $R_i = (r^1, \dots, r^{|R_i|})$, where $r^k$ denotes the school ranked in position $k$, and $|R_i|$ is the number of schools ranked.

Under the WTT assumption, the probability that student~$i$ submits ROL~$R_i$ can be decomposed as:
\begin{align*}
\Prob\left(\text{$i$ submits $R_i$} \bigm|\z_{i};\bm{\theta} \right) =
\Prob\left(\sigma_{i} = R_{i} \bigm|\z_{i};\bm{\theta}; |\sigma_i| = |R_i| \right)
\times \Prob\left(|\sigma_i|=|R_i| \bigm|\z_{i};\bm{\theta}\right)\text{,}
\end{align*}
where the second term---the probability that the student submits a list of length $|R_i|$---is treated as exogenous and unrelated to the student's preferences $u_{ij}$ for all~$j$. Estimation therefore focuses on the first term. When the preferences shocks $\epsilon_{ij}$ follow an EVT1 distribution, this conditional probability admits a closed-form expression:
\begin{align*}
\Prob\left(\sigma_{i} = R_i \bigm| \z_{i}; \bm{\theta}; |\sigma_i|=|R_i| \right)
&= \Prob\left(u_{i,r^{1}} > \cdots > u_{i,r^{|R_{i}|}} > u_{i,j'} \; \forall j' \in \mathcal{J}\backslash R_{i} \bigm|\z_{i};\bm{\theta};|\sigma_i|=|R_i|\right) \\
&= \prod_{j \in R_{i}} \frac{\text{exp}(V_{ij})}{\sum_{j' \nsucc_{R_{i}} j }\text{exp}(V_{ij'})}\text{,}
\end{align*}
where $j' \nsucc_{R_{i}} j$ denotes schools that are not ranked above $j$ in $R_i$, including $j$ itself and all unranked schools. This formulation corresponds to the rank-ordered (or ``exploded'') logit, which can be interpreted as a sequence of conditional logits: the first for the top-ranked school being preferred to all others, the second for the next-ranked school being preferred to all remaining options, and so on.\footnote{In the earlier example with four schools, the WTT assumption implies that the probability that student~$i$ submits $R_i=(b,a)$ is:
\begin{align*}
\Prob\left(\sigma_{i} = (b,a) \bigm| \z_{i};\bm{\theta};|\sigma_i|=2\right)
= \frac{\text{exp}(V_{ib})}{\sum_{j=a,b,c,d}\exp(V_{ij})}
\frac{\text{exp}(V_{ia})}{\sum_{j=a,c,d}\exp(V_{ij})}\text{,}
\end{align*}
i.e., the product of two terms: the probability that the student prefers school~$b$ to all others, and the probability that she prefers school~$a$ to the remaining unranked schools ($c$ and $d$).}

With a location normalization (e.g., setting $V_{i,1}=0$ for all $i$), the model can be estimated via maximum likelihood using the following log-likelihood function:\footnote{The rank-ordered logit model can be implemented using standard software tools such as the \texttt{Rologit} command in Stata or the \texttt{ROlogit} package in \textsf{R}.}
\begin{align*}
\ln \mathcal{L}_{\text{WTT}}\left(\bm{\theta} \; \big| \; \z, \{|R_{i}|\}_{i=1,\dots,I}\right) =
\sum_{i=1}^{I} \sum_{j \in R_{i}} V_{ij} - \sum_{i=1}^{I}\sum_{j \in R_{i}}\ln \left(\sum_{j^{\prime} \nsucc j} \text{exp}(V_{ij^{\prime}})\right)\text{,}
\end{align*}
where the first term sums the deterministic utilities of ranked schools, and the second accounts for the normalizing denominators in the rank-ordered logit. The resulting estimator is denoted by $\hat{\bm{\theta}}_{\text{WTT}}$.

The truth-telling framework can be extended to accommodate unobserved tastes for school characteristics by introducing random coefficients, using either a mixed logit or a probit specification as in \cref{eq:mixed_logit}. Although the likelihood function no longer admits a closed-form expression under these more flexible models, preference parameters can still be estimated using simulation-based methods \citep[e.g.,][]{Hastings_Kane_Staiger(2009)WP} or Bayesian approaches such as Gibbs sampling \citep[e.g.,][]{Abdulkadiroglu_Agarwal_Pathak(2017)AER, Pathak_Shi(2021)JoE}. Compared to the portfolio choice approach, the computational burden of estimating random-coefficient models under WTT is substantially lower. This is because estimation under WTT relies solely on the internal rank-ordering within each submitted list and assumes that unranked schools are strictly less preferred than the lowest-ranked option, thereby avoiding the need to evaluate expected utility across all possible ROLs.

\subsubsection{Estimation under Stability}
\label{subsec:estimation_stability}

Under the assumption that the matching outcome is stable, each student is matched to her most preferred school among those that are feasible ex-post. This section outlines how the stability-based approach can be used to estimate preferences in DA environments. We begin with settings in which students are ranked by strict, known priority scores, and then extend the analysis to cases where applicants face uncertainty about their priorities.

\paragraph{Strict priorities.} 

In settings with strict priorities, stable matching can be formulated as the outcome of a discrete choice model with personalized choice sets.

Let $\mu$ denote the matching and $\mathbf{P}(\mu)$ the associated vector of admission cutoffs, which are random variables determined by the unobserved utility shocks $\bm{\epsilon}$. A school~$j$ is ex-post feasible to student~$i$ if $i$'s priority score $t_{i,j}$ at that school meets or exceeds the cutoff, i.e., $t_{i,j} \geq P_{j}$. Denote by $\mathcal{S}(\mathbf{t}_i,\mathbf{P})$ the set of schools that are feasible for student~$i$ under the realized cutoffs. 

For example, consider a district with four schools, $\{a,b,c,d\}$, admission cutoffs $P_a = 0$, $P_b = 0.4$, $P_c = 0.6$, and $P_d = 0.8$, and a student with priority scores $t_{i,a} = 0.5$, $t_{i,b} = 0.2$, $t_{i,c} = 0.8$, and $t_{i,d} = 0.7$. The student's feasible set is then $\mathcal{S}(\mathbf{t}_i,\mathbf{P})=\{a,c\}$.

Under the stability assumption, the observed match $j^* = \mu(i)$ corresponds to the utility-maximizing option within the feasible set. That is, the probability that student~$i$ is matched to school~$j^{*}$ is given by:
\begin{align*}
\Prob\left(j^{*}=\mu(i) =
\argmax_{j \in \mathcal{S}(\mathbf{t}_{i}, \mathbf{P})} u_{ij} \bigm|\z_{i}, \mathbf{t}_{i}, \mathcal{S}(\mathbf{t}_i,\mathbf{P});\bm{\theta}\right)\text{.}
\end{align*}

This characterization transforms the assignment problem into a discrete choice model, where each student ``chooses'' a school from her personalized feasible set. Standard methods from the discrete choice literature can then be applied to estimate preferences.

The identification of preferences in this framework relies on two exogeneity assumptions:
\begin{itemize}
\item EXO~1 (\emph{Exogeneity of priority scores}). $\mathbf{t}_{i} \perp \bm{\epsilon}_{i} \mid \z_i$: Conditional on observed covariates, a student's priority scores are independent of her unobserved preference shocks.
\item EXO~2 (\emph{Exogeneity of feasible set}). $\mathcal{S}(\mathbf{t}_{i},\mathbf{P})\perp\bm{\epsilon}_{i} \mid \z_i$: Conditional on $\z_i$, a student's set of feasible schools is independent of her unobserved preferences.
\end{itemize}

Assumption EXO1 is necessary because preferences are only identified over the schools that fall within a student's personalized choice set. In particular, the preferences of students with limited access---due to low priority scores---must be inferred from those of students with broader access, i.e., higher priority scores. This assumption may fail if, conditional on observables, priority scores are correlated with unobserved determinants of preference such as ability. For instance, when priorities are based on test scores, the stability assumption does not reveal information about low-scoring students' preferences for popular schools because such schools are generally infeasible to them. This may lead to a failure to identify how test scores determine student preferences. This issue can be mitigated if alternative measures of ability are available, as in \citet{Fack_Grenet_He(2019)AER}'s study of high school choice in Paris. Moreover, if priority scores have full support within each observable ability group, e.g., because they combine multiple priority criteria, then even students with relatively weak academic records may occasionally have access to the full set of schools, restoring non-parametric identification.

Assumption EXO2 requires that the personalized choice set $\mathcal{S}(t_i, \mathbf{P})$ be exogenous to a student's unobserved preferences. Importantly, this does not imply that the cutoffs $\mathbf{P}$ are independent of preference shocks, but only that each student's feasible set is conditionally independent of her own idiosyncratic tastes. This condition may be violated in finite markets, for example, if a student can affect the cutoff of a school by choosing whether or not to apply, thereby changing which schools are feasible to her. However, such endogeneity becomes increasingly implausible in large markets, where the influence of any single student's application decision on admission cutoffs is negligible.

When student preferences follow a logit specification, and under assumptions EXO1 and EXO2, the probability that student~$i$ is matched with her assigned school~$j^{*}$ can be written as:
\begin{align*}
\Prob\left(j^{*}=\mu(i) = \argmax_{j' \in \mathcal{S}(\mathbf{t}_{i},\mathbf{P})} u_{ij'} \bigm|\z_{i};\bm{\theta}\right) =
\frac{\exp(V_{ij^{*}})}{\sum_{j' \in \mathcal{S}(\mathbf{t}_{i},\mathbf{P})}\exp(V_{ij'})}\text{.}
\end{align*}
This is the standard conditional logit expression, with the key distinction that each student's choice set is individualized, restricted to schools that are feasible given her priority scores.\footnote{In the earlier example with four schools, student~$i$ is matched with school~$c$ and her feasible set is $\{a, c\}$. Then, under stability, $c$ must be her most preferred school in that set, and the match probability is:
\begin{align*}
\Prob\left(\mu(i) = c\bigm|\z_{i};\bm{\theta}\right) =
\frac{\exp(V_{ic})}{\exp(V_{ia}) + \exp(V_{ic})}\text{.}
\end{align*}}

This leads to the following log-likelihood function, which can be estimated by maximum likelihood:
\begin{align}
\label{eq:ll_st}
\ln \mathcal{L}_{\text{Stability}}\left(\bm{\theta} \; \big| \; \z\right) =
\sum_{i=1}^{J} V_{i,j^{*}} - \sum_{i=1}^{I}\ln\left(\sum_{j' \in \mathcal{S}(\mathbf{t}_{i},\mathbf{P})} \exp(V_{ij'}) \right)\text{,}
\end{align}
with the resulting estimator denoted by $\hat{\bm{\theta}}_{\text{Stability}}$. 

As with WTT, the stability-based approach can be extended to accommodate random coefficients via mixed logit or probit models. However, it relies on more limited information: while WTT leverages the full ranking of submitted preferences, the stability assumption restricts attention to the matching outcome only. As a result, the stability-based framework may entail a loss of information, particularly regarding substitution patterns between schools. This limitation becomes more pronounced when estimating flexible models that allow for rich heterogeneity in preferences over school attributes.\footnote{This limitation of the stability-based approach has implications not only for preference estimation and welfare analysis but also for causal inference. Preference estimates often feed into control-function approaches that correct for selection when estimating treatment effects by capturing idiosyncratic preference heterogeneity \citep[e.g.,][]{Abdulkadiroglu_Pathak_Schellenberg_Walters(2020)AER, Barahona_Dobbin_Otero(2023)WP, Campos_Munoz_Bucarey_Contreras(2025)WP, Gandil(2025)WP}. If stability-based estimates systematically underrepresent certain types of heterogeneity (e.g., among low-achieving students), this shortcoming may carry over into the causal analysis. By contrast, WTT-based estimates exploit the full information in applicants' ROLs and, although potentially biased in demand estimation, may yield less biased causal forecasts. Developing a deeper understanding of these trade-offs and their implications for causal inference represents an important avenue for future research.}

\paragraph{Priority uncertainty.} 

The stability-based approach just described does not directly extend to DA environments where students face uncertainty about their priorities, such as when assignments depend on lottery-based tie-breakers.

In such cases, one might consider mechanically applying the stability condition by defining each student's feasible schools based on ex-post cutoffs derived from an arbitrary lottery draw. However, this method would be problematic for two reasons. First, it lacks a strong theoretical foundation, especially in light of results from \citet{Che_Hahm_He(2023)WP} on the asymptotic stability of robust equilibria in DA markets with payoff-insignificant mistakes. Specifically, the ex-post cutoffs from a single lottery realization could represent a low-probability event from the applicant's perspective. Hence, there is no guarantee that the preferences inferred by this method reflect the applicant's true preferences. Second, this approach requires observing the realized lottery numbers, which are often unavailable in practice. 
In their absence, researchers might resort to simulating lotteries, but this raises concerns about sensitivity to the simulated draw. Without a principled method for selecting among simulations, the inference becomes vulnerable to cherry-picking.

To address these challenges, \citet{Che_Hahm_He(2023)WP} develop a procedure known as the \emph{Transitive Extension of Preferences from Stability} (TEPS), which adapts the stability-based approach to environments with priority uncertainty. Rather than relying on a single cutoff realization, TEPS constructs a transitive preference relation that is consistent with all possible stable assignments under the observed ROL. The authors apply this method to school choice data from Staten Island in New York City.

The TEPS method proceeds in three steps to infer preference relations from each applicant's ROL. To illustrate, consider a setting with six schools $\{a,b,c,d,e,f\}$ that use tie-breaking lotteries (possibly in conjunction with coarse priorities), and suppose a student submits the ROL $(e,d,c,b)$.
\bigskip

\emph{Step~1: Simulating uncertainty and compiling choice data.} The first step simulates the uncertainty faced by applicants due to tie-breaking lotteries. This involves repeatedly executing the DA algorithm under different realizations of the lottery and recording two key pieces of information for each draw: (i)~the set of schools that are feasible (i.e., schools the applicant could have been admitted to given her lottery position), and (ii)~the school to which she is assigned.

For example, suppose the following four scenarios arise with positive probability:
\begin{itemize}
\item the feasible set is $\{d, e\}$ with probability 0.4, and the student is matched to $e$; 
\item the feasible set is $\{a, b\}$ with probability 0.3, and the student is matched to $b$; 
\item the feasible set is $\{a, b, c\}$ with probability 0.25, and the student is matched to $c$; 
\item the feasible set is $\{b, e\}$ with probability 0.05, and the student is matched to~$b$. 
\end{itemize}
These events are illustrated in Panel~(a) of \Cref{fig:TEPS}.

\begin{figure}[htbp]
\bigskip
\centering
\begin{subfigure}{0.33\textwidth}
\centering

\begin{tikzpicture}

\tikzset{every node/.style={anchor=base},
    mynode/.style={
    circle, draw, thick,
    minimum size=15pt,
    inner sep=0pt,
    text height=1.5ex,
    text depth=.25ex
    }
}

\draw[thick, rounded corners=20pt] (0,0) ellipse (0.5cm and 1cm);
\draw[thick, rounded corners=20pt] (1.2,0) ellipse (0.5cm and 1cm);
\draw[thick, rounded corners=20pt] (2.4,0) ellipse (0.5cm and 1cm);
\draw[thick, rounded corners=20pt] (3.6,0) ellipse (0.5cm and 1cm);

\node[mynode] (e) at (0,0.4) {$e$};
\node (d) at (0,-0.5) {$d$};

\node[mynode] (d) at (1.2,0.4) {$b$};
\node (e) at (1.2,-0.5) {$a$};

\node[mynode] (d) at (2.4,0.4) {$c$};
\node (a) at (2.2,-0.5) {$a$};
\node (b) at (2.6,-0.5) {$b$};

\node[mynode] (d) at (3.6,0.4) {$e$};
\node (e) at (3.6,-0.5) {$b$};

\end{tikzpicture}
    \caption{Step~1}
    \end{subfigure}%
    \begin{subfigure}{0.33\textwidth}
    \centering
\begin{tikzpicture}
\tikzset{
    labelstyle/.style={
    text height=1.5ex,
    text depth=.25ex
    }
}

\coordinate (D1) at (0,0);
\coordinate (B1) at (0.8,0);
\coordinate (B2) at (1.6,0);
\coordinate (A2) at (2.4,0);
\coordinate (A3) at (3.2,0);

\coordinate (E1) at (0.4,1);
\coordinate (C2) at (2,1);
\coordinate (B3) at (3.2,1);

\draw[thick] (D1) -- (E1) -- (B1);
\draw[thick] (B2) -- (C2) -- (A2);
\draw[thick] (A3) -- (B3);

\node[below, labelstyle] at (D1) {$d$};
\node[below, labelstyle] at (B1) {$b$};
\node[below, labelstyle] at (B2) {$b$};
\node[below, labelstyle] at (A2) {$a$};
\node[below, labelstyle] at (A3) {$a$};
\node[above, labelstyle] at (E1) {$e$};
\node[above, labelstyle] at (C2) {$c$};
\node[above, labelstyle] at (B3) {$b$};

\end{tikzpicture}

\caption{Step~2}
\end{subfigure}
\begin{subfigure}{0.33\textwidth}
\centering

\begin{tikzpicture}

\tikzset{
    labelstyle/.style={
    text height=1.5ex,
    text depth=.25ex
    }
}

\coordinate (A1) at (1,0);
\coordinate (D1) at (0,0.5);
\coordinate (B1) at (1,0.5);
\coordinate (E1) at (0.5,1);
\coordinate (A2) at (2,0);
\coordinate (B2) at (2,0.5);
\coordinate (C2) at (2,1);

\draw[thick] (D1) -- (E1) -- (B1) ;
\draw[thick] ([yshift=-0.7ex]A1) -- ([yshift=-1.5ex]B1);
\draw[thick] ([yshift=-0.7ex]A2) -- ([yshift=-1.5ex]B2);
\draw[thick] (B2) -- ([yshift=0.2ex]C2);

\node[circle, fill=white, inner sep=1pt] at (D1) {$d$};
\node[circle, fill=white, inner sep=1pt] at (B1) {$b$};
\node[below, labelstyle] at (A1) {$a$};
\node[above, labelstyle] at (E1) {$e$};
\node[below, labelstyle] at (A2) {$a$};
\node[circle, fill=white, inner sep=1pt]  at (B2) {$b$};
\node[above, labelstyle] at (C2) {$c$};

\end{tikzpicture}

\caption{Step~3}
\end{subfigure}
\caption{Transitive Extension of Preferences from Stability (TEPS): Example}
\label{fig:TEPS}
\begin{threeparttable}
\begin{tablenotes}
\item \emph{Source:} \citet{Che_Hahm_He(2023)WP}.
\end{tablenotes}
\end{threeparttable}
\end{figure}
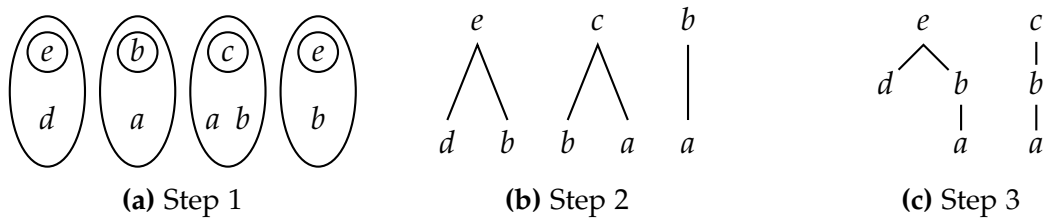

\medskip

\emph{Step~2: Inferring preference relations within each realization.} Assuming that students are matched to their most preferred school among those that are feasible in each lottery realization (as implied by robust equilibrium in large markets), one can recover a set of ordered preference relations from each assignment. In the example: assignment to $e$ when the feasible set is $\{d, e\}$ or $\{b, e\}$ implies $e \succ d$ and $e \succ b$; assignment to $c$ when the feasible set is $\{a, b, c\}$ implies $c \succ a$ and $c \succ b$; and assignment to $b$ when the feasible set is $\{a, b\}$ implies $b \succ a$. These inferred relations are visualized in Panel~(b) of \Cref{fig:TEPS} as directed trees, with each tree rooted at the school to which the student is assigned.

\medskip

\emph{Step 3: Extending preference relations by transitivity.} The final step aggregates all preference relations recovered in Step~2 by applying the transitivity axiom. In the example, this yields the extended set of preference relations $\{(b, a),\allowbreak (c, b),\allowbreak (c, a),\allowbreak (e, b),\allowbreak (e, d),\allowbreak (e, a) \}$, which are depicted as a transitive closure in Panel~(c) of \Cref{fig:TEPS}.

\bigskip

The TEPS procedure recovers all preference relations that are consistent with the stability condition across all possible realizations of uncertainty and the transitivity of preferences. Importantly, while these assumptions allow researchers to infer rich ordinal information, they do not uniquely determine the ROL a student submits. In fact, given other participants' strategies, a student may have multiple best-response ROLs. Despite this multiplicity, Che et al.\ show that TEPS does not suffer from incompleteness or incoherence in the sense of \citet{Tamer(2003)Restud}. The key reason is that this procedure relies solely on the distribution of possible assignment outcomes for each student, which is unique in a large market, regardless of which best-responding ROL is observed.

Once preference relations have been inferred, one can estimate a parametric utility function that best fits these preferences. Che et al.\ recommend a multinomial probit specification estimated via Gibbs sampling, as the sampler can flexibly draw cardinal utilities that satisfy the observed ordinal constraints.\footnote{\label{fn:non_analytic_LL}In this framework, constructing the likelihood function under a logit specification is generally infeasible, even with i.i.d.\ EVT1 errors, e.g., if the inferred preferences are $b \succ c$ and $c \succ d \succ e$. By contrast, these ordinal constraints can be directly translated into inequality bounds on the utilities, which are readily handled in a Gibbs sampler.}

\subsubsection{ Using Undominated Strategies}
\label{subsec:estimation_undominated_strategies}

The assumption that applicants do not play weakly dominated strategies does not generally imply a unique ROL for each student (see \Cref{subsec:preference_estimation_IA,subsec:incomplete_models}). Instead, it yields a set of inequality restrictions that can be transformed into moment inequalities. These can be used to partially identify preferences \citep[e.g.,][]{Hwang(2015)EAI, He(2017)WP, Bayraktar_Hwang(2024)WP} or serve as over-identifying restrictions in estimation \citep[e.g.,][]{He(2017)WP, Fack_Grenet_He(2019)AER}. 

\paragraph{Estimation with moment inequalities.}

Consider a DA environment under the assumption that students do not play weakly dominated strategies (see \Cref{subsec:incomplete_models}). In such settings, submitted ROLs provide partial information about students' preferences, which can be exploited to derive inequality restrictions. These restrictions, in turn, can be translated into conditional and unconditional moment inequalities for estimation.

To illustrate, consider two schools, $a$ and $b$. Because students do not necessarily rank all schools, only some ROLs reveal a preference between $a$ and $b$. Under the assumption of undominated strategies, the probability that student~$i$ ranks $a$ above $b$ in her ROL~$R_i$ provides a lower bound on the probability that she prefers $a$ to $b$:
\begin{align*}
\Prob(a \succ_{R_i} b \mid \z_i;\bm{\theta} )
&= \Prob\left(\{u_{i,a}>u_{i,b}\} \cap \{a,b \in R_{i}\} \mid \z_i;\bm{\theta} \right) \notag \\
& \leq \Prob\left(u_{i,a} > u_{i,b} \mid \z_i; \bm{\theta}\right)\text{.}
\end{align*}
Similarly, the probability that $b$ is ranked above $a$ implies an upper bound:
\begin{align*}
\Prob\left(u_{i,a} > u_{i,b} \mid \z_i; \bm{\theta}\right) \leq 1 - \Prob(b \succ_{R_i} a \mid \z_i; \bm{\theta})\text{.}
\end{align*}
Together, these bounds yield two conditional moment inequalities:
\begin{align*}
\Prob\left(u_{i,a} > u_{i,b} \mid \z_i; \bm{\theta}\right) - \mathbb{E}\left[\mathds{1}(a \succ_{R_i} b) \mid \z_i; \bm{\theta}\right] &\geq 0 \text{;} \\
1 - \mathbb{E}\left[\mathds{1}(b \succ_{R_i} a) \mid \z_i;\bm{\theta}\right] - \Prob\left(u_{i,a} > u_{i,b}\mid \z_i; \bm{\beta}\right) &\geq 0\text{.}
\end{align*}
Such inequalities can be derived for any pair or subset of schools and then interacted with functions of observable covariates $\z_i$ to construct a vector of $M_1$ unconditional moment inequalities, denoted by $(m_1, \dots, m_{M_1})$.

Estimation proceeds by testing whether these inequalities hold in the data. \citet{Andrews_Shi(2013)ECTA} propose a Cramér--von Mises-type statistic for this purpose, based on the modified methods of moments (or sum function):
\begin{align}
\label{eq:t_mi}
T_{\text{MI}}(\bm{\theta}) = \sum_{k=1}^{M_1}\left[ \frac{\overline{m}_{k}(\bm{\theta})}{\hat{\sigma}_{k}(\bm{\theta})} \right]^{2}_{-}\text{,}
\end{align}
where $\overline{m}_{k}(\bm{\theta})$ is the sample mean, $\hat{\sigma}_{k}(\bm{\theta})$ is the sample standard deviation of the $k$th moment, and $[x]_{-} := \min\{0,x\}$ ensures that only violations of the inequality contribute to the test statistic. \citet{Bugni_Canay_Shi(2017)QE} develop methods to construct \emph{marginal confidence intervals} using this statistic. For each parameter component $\theta_k$, they provide a test for the null hypothesis $H_0: \theta_k = \theta_0$ for each candidate value $\theta_0 \in \mathbb{R}$. The confidence interval for $\theta_k$'s true value is then the convex hull of all $\theta_0$ values for which the null is not rejected.

\paragraph{Combining moment equalities and moment inequalities.} 

Approaches based solely on moment inequalities in school choice settings often yield wide confidence intervals of the parameters of interest, constrained by the limitations of current econometric tools. As a result, such inference may be uninformative in practice. A more effective strategy is to treat inequality restrictions as over-identifying information within a complete model framework \citep{Moon_Schorfheide(2009)JoE}. For example, one can combine the restrictions implied by undominated strategies (expressed as moment inequalities) with the moment equalities implied by the stability assumption. Specifically, the likelihood function in \cref{eq:ll_st} can be rewritten as a system of moment equalities by equating predicted and observed school-level matching frequencies:
\begin{align*}
\sum_{i=1}^{I}\Prob\left(j = \argmax_{j'\in\mathcal{S}(\mathbf{t}_i,\mathbf{P})} u_{ij'} \mid \z_i;\bm{\theta}\right) -
\mathbb{E}\left(\sum_{i=1}^{I} \mathds{1}(\mu(i)=j)\right) = 0, \quad \forall j \in \mathcal{J}\text{,}
\end{align*}
where $\mathds{1}(\mu(i)=j)$ is an indicator equal to one if student~$i$ is matched to school~$j$. These moment equalities can be further interacted with covariates $\z_i$ to construct $M_2$ unconditional moment equalities $(m_{M_1+1}, \dots, m_{M1+M2})$.

To jointly incorporate the full set of $M_1$ moment inequalities and $M_2$ moment equalities, the test statistic in \cref{eq:t_mi} can be extended to:
\begin{align}
\label{eq:mei}
T_{\text{MEI}}(\bm{\theta}) = \sum_{k=1}^{M_1}\left[ \frac{\overline{m}_{k}(\bm{\theta})}{\hat{\sigma}_{k}(\bm{\theta})} \right]^{2}_{-} +
\sum_{k=M_{1}+1}^{M_1+M_2}\left[ \frac{\overline{m}_{k}(\bm{\theta})}{\hat{\sigma}_{k}(\bm{\theta})} \right]^{2}_{-}\text{.}
\end{align}
Here, the first term penalizes violations of inequality constraints (as before), while the second term enforces the equality conditions. The estimator $\hat{\bm{\theta}}_{\text{MEI}}$ minimizes this objective function, and marginal confidence intervals can be constructed using the procedure developed by \citet{Bugni_Canay_Shi(2017)QE}.

\subsubsection{Selecting Among Alternative Preference Inference Hypotheses}

Each of the preference estimation approaches reviewed above relies on distinct identifying assumptions grounded in revealed preference restrictions. A natural question is whether empirical tests can distinguish between these alternatives. The answer depends critically on the richness of the data and the degree of structure one is willing to impose on preferences and behavior.

\paragraph{Overcoming the limitations of administrative data.}

Inference based solely on administrative data from centralized school choice or college admission systems faces fundamental limitations. Absent behavioral or parametric restrictions, it is generally not possible to empirically test the identifying assumptions underlying different estimation approaches. For instance, in the portfolio choice framework, \citet[][Theorem~A.1]{Agarwal_Somaini(2018)ECTA} show that any ROL can be rationalized as optimal for some utility vector, as long as each ranked school has strictly positive assignment probability. This flexibility makes it impossible to separately identify heterogeneity in preferences from heterogeneity in strategic sophistication without additional structure. Likewise, it is always possible to find a utility vector that is consistent with both truth-telling and stability, making it difficult to discriminate between these assumptions using administrative data alone.

To address these limitations, several complementary strategies have emerged. One line of research uses laboratory experiments to examine how individuals report preferences in controlled matching environments, offering insights into the determinants of preference misrepresentation under both manipulable and strategy-proof mechanisms (see the reviews by \citealp{Hakimov_Kubler(2021)ExpEcon} and \citealp{Rees-Jones_Shorrer(2023)JPE:Micro}). Complementing this experimental evidence, a growing number of field studies exploit institutional features of real-world matching systems, such as the joint determination of admissions and financial aid, to detect preference misrepresentation using administrative data from strategy-proof mechanisms \citep[e.g.]{Hassidim_Romm_Shorrer(2021)MNSC, Artemov_Che_He(2023)JPE:Micro, Shorrer_Sovago(2023)JPE:Micro}.

A second approach relies on field surveys administered to participants in school choice or college admissions systems to elicit their preferences and beliefs \citep[e.g.,][]{Budish_Cantillon(2012)AER, Chen_Pereyra(2019)GEB, Kapor_Neilson_Zimmerman(2020)AER, Larroucau_Rios(2020)WP, Hassidim_Romm_Shorrer(2021)MNSC, Arteaga_et_al(2022)QJE, Bobba_Frisancho_Pariguana(2023)WP, De_Haan_et_al(2023)JPE, Corradini_Idoux(2025)NBER, Larroucau_et_al(2025)NBER}. These surveys have documented systematic biases in beliefs about admission probabilities, limited knowledge of available options and their characteristics, misunderstanding of the assignment mechanism, and evidence of optimization failures (see \Cref{subsubsec:information_frictions,subsubsec:bounded_rationality}). This survey evidence is increasingly being combined with administrative data to estimate structural models of application behavior and conduct counterfactual policy simulations \citep[e.g.,][]{Kapor_Neilson_Zimmerman(2020)AER,   De_Haan_et_al(2023)JPE, Ajayi_Sidibe(2024)WP, Ajayi_Friedman_Lucas(2025)EJ, Larroucau_et_al(2025)NBER}. 

As discussed in \Cref{subsec:preference_estimation_IA,subsec:preference_estimation_DA}, this body of work has directly informed the development of methods that relax the assumption of full rationality in preference estimation. These include approaches that allow for heterogeneous levels of strategic sophistication among applicants \citep{Hwang(2015)EAI, He(2017)WP, Agarwal_Somaini(2018)ECTA, Calsamiglia_Fu_Guell(2020)JPE, Kapor_Neilson_Zimmerman(2020)AER, Bayraktar_Hwang(2024)WP}; models in which agents fail to fully internalize continuation values when constructing their ROL \citep{Idoux(2023)WP, Wang_Wang_Ye(2025)NBER}; and frameworks that accommodate strategic mistakes even in strategy-proof environments \citep{Artemov_Che_He(2023)JPE:Micro, Che_Hahm_He(2023)WP}.

\paragraph{Leveraging the nesting structure of identifying assumptions.}

When relying solely on administrative data, the plausibility of the assumptions underlying different preference estimation approaches must be carefully evaluated in context. For instance, assuming rational expectations in the portfolio choice framework is more credible when participants have access to reliable information about the competitiveness of various schools than when this information is not publicly available. In DA settings, WTT is more plausible when students face significant uncertainty about their priorities and can rank an unlimited number of schools at no cost. Conversely, when list-length restrictions or application costs are present, stability may be more compelling than truth-telling, especially in large markets where admission cutoffs are easier to predict.

In some cases, the nesting structure of identifying assumptions allows for specification tests. \citet{Fack_Grenet_He(2019)AER} exploit this feature to propose a test of WTT against stability in strict-priority DA settings.\footnote{For a generalization of this test to settings with priority uncertainty, see \citet{Che_Hahm_He(2023)WP}. \citet{Fack_Grenet_He(2019)AER} also propose a test of stability against the undominated-strategies assumption, which involves checking whether the identified set from the moment (in)equality model \cref{eq:mei} is empty, using the specification test developed by \citet{Bugni_Canay_Shi(2015)JoE}. However, current moment (in)equalities techniques tend to yield conservative confidence sets, limiting the power of such tests.} Because WTT implies stability but not vice versa, WTT imposes over-identifying restrictions under the stability framework. This allows for a Hausman-type test: under the null hypothesis that applicants are weakly truth-telling, both the WTT and stability estimators, $\hat{\bm{\theta}}_{\text{WTT}}$ and $\hat{\bm{\theta}}_{\text{Stability}}$, are consistent; under the alternative---that WTT is violated but the assignment is stable---only $\hat{\bm{\theta}}_{\text{Stability}}$ remains consistent. A test statistic can thus be constructed as:
\begin{align*}
\left(\hat{\bm{\theta}}_{\text{Stability}} - \hat{\bm{\theta}}_{\text{WTT}}\right)^{\prime}
\left(\mathbb{V}(\hat{\bm{\theta}}_{\text{Stability}}) - \mathbb{V}(\hat{\bm{\theta}}_{\text{WTT}})\right)^{-1}
\left(\hat{\bm{\theta}}_{\text{Stability}} - \hat{\bm{\theta}}_{\text{WTT}}\right)\text{,}
\end{align*}
where $\mathbb{V}(\cdot)$ denotes the covariance matrix of the estimator. Under the null, the statistic is asymptotically $\chi^{2}_{|\bm{\theta}|}$ distributed, and rejection indicates that the data are inconsistent with WTT, assuming correct model specification.

This test has been applied across a range of DA settings and has generally led to rejection of WTT in favor of stability \citep{Fack_Grenet_He(2019)AER, Arslan(2021)EmpEcon, Che_Hahm_He(2023)WP, Anderson_et_al(2024)WP, Ngo_Dustan(2024)AEJ:App}. For example, using data from 15 Swedish school districts over multiple years, \citet{Anderson_et_al(2024)WP} find that WTT is rejected in 66 of 75 datasets. A caveat, however, is that rejection of WTT in favor of stability does not definitely rule out truth-telling, as the test is conditional on the model's parametric assumptions. In the presence of misspecification, the test may simply favor the less restrictive model. Developing non-parametric methods to formally compare and select among competing identification assumptions represents a promising direction for future research.


\section{Empirical Insights from Centralized School Choice}
\label{sec:empirical_insights}

The methods reviewed in the previous section, together with researchers' growing access to rich administrative data from centralized school assignment systems, have fueled a rapidly expanding empirical literature applying the tools of market design to school choice. This section provides an overview of some of the main insights that have emerged from this body of work.

While the broader empirical literature on school choice has traditionally focused on the effects of competition between public, private, charter, and voucher schools on school productivity, student sorting, and academic outcomes \citep[see][for surveys]{Hoxby(2003)book, Bettinger(2011)Handbook, Epple_Romano_Zimmer(2016)Handbook, Urquiola(2016)Handbook, Cohodes_Parham(2021)OxREconFin}, the empirical market design literature places greater emphasis on the organization of the choice process itself. This line of research builds on theoretical advances in matching theory to study the design of centralized assignment systems, how families respond to their core features, and how these design choices affect individual welfare and distributional outcomes. Methodologically, it combines reduced-form and structural approaches, leveraging detailed administrative data from these assignment systems to estimate student preferences, assess the performance of matching mechanisms, and conduct counterfactual analyses under alternative policy designs. Our review builds on earlier syntheses by \citet{Cantillon(2017)OxRep}, \citet{Pathak(2017)chapter}, and \citet{Agarwal_Somaini(2020)AnnuRev}, while incorporating more recent empirical contributions and policy applications. Where relevant, we also draw on evidence from centralized college admissions, as findings from these settings provide useful perspectives for the design of school choice systems.

We organize the discussion around three central themes. First, we examine how families form preferences and make application decisions, with particular attention to information frictions and behavioral deviations from full rationality (\Cref{subsec:preferences_application_behavior}). Second, we explore the design trade-offs in school choice systems, including the degree of centralization, the choice of assignment mechanisms, and the use of decision-support tools (\Cref{subsec:market_design_tradeoffs}). Third, we assess the broader effectiveness of centralized school choice reforms, focusing on their market-level impacts, distributional consequences, and interactions with other household decisions such as residential choice (\Cref{subsec:school_choice_effectiveness}).

\subsection{Preferences and Application Behavior}
\label{subsec:preferences_application_behavior}

A central question in the empirical school choice literature is how families form preferences and translate them into application decisions. This section reviews the evidence along three dimensions that jointly shape how families engage with centralized assignment systems and help explain observed disparities in outcomes: (i)~the determinants of parental preferences; (ii)~the distorting effects of information gaps on decision-making; and (iii)~the influence of bounded rationality on application behavior.

\subsubsection{Drivers of Preferences and Preference Heterogeneity}

A consistent finding across empirical studies of school choice is that families' reported preferences correlate with both proximity to schools and a range of school attributes \citep[e.g.,][]{Hastings_Kane_Staiger(2009)WP, Burgess_et_al(2015)EJ, Akyol_Krishna(2017)EER, Pathak_Shi(2021)JoE, Beuermann_et_al(2023)REStud}. These attributes include not only academic performance indicators and peer composition, but also more idiosyncratic features such as school denomination and educational philosophy \citep[e.g.,][]{Borghans_Glosteyn_Zolitz(2015)BEJEAP}.

A key trade-off parents face is between school quality and geographic proximity. Because the set of available alternatives varies in quality and distance, parents must often choose between a nearby school and one with stronger academic outcomes. This trade-off is particularly pronounced for low-socioeconomic status (low-SES) families, who are more constrained in their options due to residential sorting. In England, \citet{Burgess_et_al(2015)EJ} estimate that differences in geographic constraints 
explain approximately two-thirds of the observed SES gap in the quality of chosen schools. Geography thus remains a major obstacle to equalizing access to high-performing schools.

Even when proximity and school availability are held constant, substantial heterogeneity in parental preferences persists. Numerous studies document systematic differences in school preferences by socioeconomic status and race \citep[e.g.,][]{Hastings_Kane_Staiger(2009)WP, Borghans_Glosteyn_Zolitz(2015)BEJEAP, Burgess_et_al(2015)EJ, Oosterbeek_Sovago_van_der_Klauw(2021)JPubE, Ajayi(2024)JHR, Laverde(2024)WP, Corradini_Idoux(2025)NBER}, as well as by gender \citep[e.g.,][]{Ngo_Dustan(2024)AEJ:App}. For instance, in their landmark study of school choice in the Charlotte-Mecklenburg School Public School District (CMS), North Carolina, \citet{Hastings_Kane_Staiger(2009)WP} show that higher-SES families place greater weight on test scores, whereas lower-SES and minority families balance academic quality against preferences for schools with predominantly same-race peers. Vignette experiments in New York City \citep{Corradini_Idoux(2025)NBER} confirm that families favor schools with peer groups that match their own racial background. In Amsterdam, \citet{Oosterbeek_Sovago_van_der_Klauw(2021)JPubE} estimate that such preference heterogeneity accounts for 40\% of ethnic segregation and 25\% of income segregation across secondary schools. These findings underscore a core limitation of school choice as a tool for integration: even when access is formally equalized, segregation may persist or even intensify if preferences remain stratified by social group.

A central policy concern is whether parents value a school's causal effectiveness (value-added) or merely its peer composition. If families prioritize peers over instructional quality, \citet{Barseghyan_Clark_Coate(2014)AEJPol} show that school choice reforms may fail to generate competitive pressure for academic improvement, and could even reduce overall welfare. For choice to raise productivity, it must reward schools for effective teaching, not for attracting already high-achieving students. Yet this distinction is often blurred in practice. In New York City, \citet{Abdulkadiroglu_Pathak_Schellenberg_Walters(2020)AER} find that parental preferences correlate with both peer quality and value-added, but are uncorrelated with true effectiveness once peer composition is controlled for. This suggests that parents tend to penalize schools serving lower-achieving peers regardless of the actual instructional value they provide, weakening the intended incentive structure of school choice policies.

Whether parents value school effectiveness may depend on context, particularly the information environment and the multidimensional nature of school output. Evidence from \citet{Beuermann_et_al(2023)REStud} underscores this point. Using administrative and survey data on public secondary school choice in Trinidad and Tobago, the authors estimate the causal effects of attending specific schools on a wide array of academic and non-academic outcomes, including dropout, arrests, teen motherhood, and labor market participation. They find that schools' causal effects on student performance in the high-stakes tests taken at the end of secondary school are only weakly correlated with their effects on longer-term non-academic outcomes, highlighting the limitations of relying on test scores as the sole measure of school effectiveness. Linking these causal estimates to parents' school rankings, they show that parents do value schools that raise test scores, but also prefer those that reduce crime and teen births and improve labor market outcomes. Notably, families of low-achieving students place greater weight on non-academic impacts, whereas high-achieving students' families prioritize test score gains. These results indicate that, in some settings, parents recognize and act upon a broader conception of school quality---one that extends beyond test scores---which highlights the importance of accounting for a wider range of school contributions when designing and evaluating school choice policies.

Together, these findings raise a deeper question: to what extent do application patterns reflect genuine preferences, as opposed to differences in the information families possess and how they process it when forming and reporting school rankings? This distinction has spurred a growing body of work on informational frictions and bounded rationality in school choice. 

\subsubsection{Information Frictions}
\label{subsubsec:information_frictions}

Across a wide range of contexts, school choice participants are often found to be misinformed about key aspects of the matching process, such as the characteristics of schools, the full menu of available options, and their chances of admission. These information frictions carry important implications for both the efficiency and equity of assignment outcomes.

A consistent finding in the literature is that many families lack accurate information about school attributes, which can lead to suboptimal application decisions \citep[e.g.,][]{Hastings_Weinstein(2008)QJE, Ainsworth_et_al(2023)AER, Agte_et_al(2024)NBER, Campos(2024)NBER, Corradini(2024)WP,  Ajayi_Friedman_Lucas(2025)EJ, Corradini_Idoux(2025)NBER, Larroucau_et_al(2025)NBER}. In a seminal study leveraging both natural and field experimental variation in the CMS District, \citet{Hastings_Weinstein(2008)QJE} show that providing low-income families with information about school test scores increases their likelihood of applying to higher-scoring schools and leads to improved outcomes for their children. These findings illustrate that disadvantaged households may fail to consider high-performing schools simply due to lack of information, even when those schools are accessible.

Other studies highlight persistent racial and social gaps in the quality of information families possess. In New York City, \citet{Corradini(2024)WP} documents that Black and Hispanic families apply to lower value-added high schools than White and Asian families, and attributes part of this gap to inaccurate beliefs about school quality. This is evidenced by the fact that the introduction of a letter-grade rating system shifted their choices more than those of white and Asian families, reducing racial gaps in both enrollment at high-quality schools and subsequent academic achievement. Complementing this, \citet{Corradini_Idoux(2025)NBER} find that racial disparities in application behavior also stem from differences in school awareness, i.e., the extent to which families know that certain schools exist.

Social interactions play a critical role in shaping how families access and interpret school-related information. \citet{Campos(2024)NBER} finds that parents tend to underestimate school value-added and overestimate peer quality. But when provided with accurate information, both they and their neighbors adjust their preferences in favor of higher value-added schools. This pattern reveals not only a latent demand for school effectiveness but also the powerful mediating role of social networks in transmitting information and influencing choices. In a similar vein, \citet{Corradini_Idoux(2025)NBER} show that exposure to more diverse middle school peers reduces racial gaps in high school applications by broadening families' awareness of available school options and shifting preferences toward greater peer diversity. Their evidence suggests that racial disparities in school choice are not only driven by homophily in preferences but also by unequal information environments mediated by social exposure. Together, these findings highlight that social context and peer environments can reinforce or mitigate informational inequalities, with important implications for both segregation and access to high-quality schools.

Information frictions can also arise from the difficulty of navigating a complex menu of options. When families face high search and processing costs, they may make poorly informed decisions upfront and revise their preferences only after experiencing their initial match. \citet{Narita(2018)WP} provides evidence of this dynamic in New York City's high school choice system, where families dissatisfied with their initial match can participate in a discretionary reapplication process. He finds that over 70\% of reapplicants reverse their original school rankings, with most changes reflecting genuine updates in preferences driven by learning rather than strategic misreporting. These shifts in demand undermine the welfare performance of the initial match, which assumes stable and well-informed preferences. \citet{Narita(2018)WP} further shows that a centralized reapplication mechanism could significantly outperform the existing discretionary process, highlighting the importance of aftermarket design to accommodate evolving preferences.

Alongside misperceptions of school attributes, families also hold inaccurate beliefs about admission chances \citep[e.g.,][]{Kapor_Neilson_Zimmerman(2020)AER, Agte_et_al(2024)NBER, Ajayi_Sidibe(2024)WP, Arteaga_et_al(2022)QJE}. In environments that require strategic decision-making, such errors can be particularly costly because they directly influence how applicants rank schools. In New Haven, \citet{Kapor_Neilson_Zimmerman(2020)AER} find that families frequently misperceive their admission probabilities, with average absolute deviations from rational expectations reaching 37 percentage points. A common misconception is failing to understand that placing a school lower in one's ROL reduces the chance of admission. These belief errors are not merely anecdotal: Kapor et al.\ estimate that they lead to a 24\% welfare loss compared to a best-case scenario in which families are fully informed and play the Bayes--Nash equilibrium in the game induced by the New Haven mechanism.

Importantly, such distortions also arise under strategy-proof mechanisms and can induce suboptimal behavior by distorting how applicants conduct their school search. Using data from a large-scale survey of participants in Chile's centralized school choice system, \citet{Arteaga_et_al(2022)QJE} find that applicants systematically overestimate their admission probabilities, especially when the true likelihood of placement is low. This leads to premature termination of the search process and the submission of overly narrow application portfolios. More broadly, \citet{Agte_et_al(2024)NBER} show that misperceptions about school availability, quality, and admission probabilities interact to reduce search effort and lower match quality relative to a full-information benchmark.

These findings have motivated the development of school choice platforms and interventions aimed at helping families navigate the complexity of the application process and conduct a more effective school search. We return to these policy implications and design innovations in \Cref{subsec:centralized_vs_decentralized}.

\subsubsection{Bounded Rationality}
\label{subsubsec:bounded_rationality}

While information frictions contribute to suboptimal decisions in centralized school choice, a growing body of evidence suggests that participants also deviate from fully rational decision-making in ways that go beyond informational constraints.

A longstanding concern in school choice design is that families differ in their ability to understand the assignment rules and optimize their application strategies accordingly. This heterogeneity in strategic sophistication has motivated the widespread adoption of strategy-proof mechanisms, which aim to protect participants who may lack the cognitive resources for complex reasoning. As discussed in \Cref{subsec:cardinal_welfare}, mechanisms that reward strategic play may yield inequitable outcomes when some families are better equipped to navigate the strategic environment than others. This concern is empirically illustrated by \citet{Abdulkadiroglu_Pathak_Roth_Sonmez(2006)NBER, Abdulkadiroglu_Pathak_Roth_Sonmez(2006)WP}, who show that under Boston's former IA mechanism, families differed markedly in their ability to anticipate capacity constraints. While some adjusted their submissions strategically, others failed to do so, leading to avoidable non-assignments. Further evidence comes from lab experiments by \citet{Basteck_Mantovani(2018)GEB}, who find that under manipulable mechanisms, participants with lower cognitive ability are more likely to be matched to lower-quality schools.

Notably, even in strategy-proof settings with readily accessible information, many applicants continue to misrepresent their preferences. This pattern has been documented in both laboratory and field settings (see \Cref{sec:TT_approach}), and in some cases results in meaningful welfare losses. These findings raise concerns about the real-world effectiveness of non-manipulable mechanisms and suggest that their theoretical benefits may not always be realized in practice. In response, a growing literature in behavioral market design seeks to identify the sources of preference misrepresentation in strategy-proof environments. While the full policy implications have yet to be fully explored, incorporating behavioral insights into system design is increasingly seen as essential for ensuring that school choice mechanisms are not only theoretically robust but also practically effective and equitable.

\citet{Rees-Jones_Shorrer(2023)JPE:Micro} provide a comprehensive review of this emerging literature and identify three broad categories of behavioral explanations for preference misrepresentation in strategy-proof environments. 

The first centers on incorrect processing of incentives, where participants fail to fully grasp the strategic properties of the mechanisms. This includes failures of contingent reasoning \citep{Niederle_Vespa(2023)AnnuRev}, reliance on faulty heuristics \citep{Hassidim_Romm_Shorrer(2017)AER:PP, Rees-Jones_Shorrer(2023)JPE:Micro}, and misunderstandings of continuation value \citep{Idoux(2023)WP, Wang_Wang_Ye(2025)NBER}.

The second category involves non-standard preferences. These include report-dependent preferences, where individuals are averse to rejection and derive utility from being assigned to a school they actively ranked, leading them to apply to schools with higher admission likelihoods \citep{Meisner(2022)MNSC, Kloosterman_Troyan(2024)WP}; and expectations-based reference dependence, where students form psychological attachments to schools they expect to attend and avoid ranking high-value but unlikely options to minimize potential disappointment \citep{Dreyfuss_Heffetz_Rabin(2022)AEJ:Micro, Meisner_Wangenheim(2023)JET}.

The third category stems from models of incorrect processing of admission probabilities. A common example is overconfidence in admission chances \citep{Pan(2019)GEB, Kapor_Neilson_Zimmerman(2020)AER, Arteaga_et_al(2022)QJE, Larroucau_et_al(2025)NBER}, which can lead students to under- or over-apply to certain schools. Another documented bias is correlation neglect, where applicants fail to recognize that admission probabilities are correlated across schools, causing them to omit viable fallback options from their application portfolios \citep{Rees-Jones_Shorrer_Tergiman(2024)AEJ:Micro}.

These behavioral insights are increasingly being integrated into empirical methods for preference estimation, either by allowing for certain types of mistakes in the stability-based approach \citep{Artemov_Che_He(2023)JPE:Micro, Che_Hahm_He(2023)WP}, or by embedding psychological structure into the portfolio choice models \citep{Idoux(2023)WP, Wang_Wang_Ye(2025)NBER}. Beyond estimation, they have begun to inform the design of interventions, such as simplified application platforms, nudges, and targeted information provision, aimed at mitigating welfare losses from cognitive and behavioral frictions. As discussed in the next section, this research agenda opens promising avenues for policy innovation in the design and implementation of school choice systems.

\subsection{Market Design Trade-offs}
\label{subsec:market_design_tradeoffs}

Closely linked to the extensive theoretical literature on matching mechanisms in education, a substantial body of empirical research has investigated the trade-offs involved in the design of centralized assignment systems.

\subsubsection{Centralized vs.\ Decentralized Choice}
\label{subsec:centralized_vs_decentralized}

One of the main theoretical motivations for adopting centralized, single-offer school assignment systems is their ability to reduce misallocation caused by congestion and matching frictions in decentralized and uncoordinated processes. However, empirically comparing the allocative efficiency of centralized and decentralized systems is challenging, as the latter rarely produce the kind of systemic data needed to directly assess student welfare across regimes.

\citet{Abdulkadiroglu_Agarwal_Pathak(2017)AER} address this challenge by analyzing New York City's transition from an uncoordinated high school admissions system to one based on the DA mechanism in 2003. Using data from the new centralized system to estimate student preferences, they find that it achieved approximately 80\% of the welfare gains attainable by moving from a neighborhood-based assignment to a utilitarian optimal allocation, while the uncoordinated system achieved at most 35\% of these gains. The largest improvements accrued to students who were most likely to remain unassigned in the previous system and were administratively placed after the main round. These results suggest that congestion and ad hoc placement were key sources of inefficiency in the uncoordinated system. Notably, the efficiency gains occurred in a context where families were already familiar with school choice, as both systems used a common application. This implies that even greater benefits from coordination may be possible in settings with less structured markets or less experienced participants.

Further support for centralization comes from studies of the welfare effects of off-platform options in centralized assignment mechanisms—i.e., alternatives available to participants but operating outside the centralized platform. The existing evidence comes primarily from college admissions. In Chile, \citet{Kapor_Karnani_Neilson(2024)JPE} show that off-platform programs generate negative externalities, as applicants forgo centralized offers in favor of alternatives, leading to greater waitlisting and aftermarket frictions. Exploiting a 2012 reform that expanded the centralized platform's capacity by 40\%, they find substantial welfare gains and higher graduation rates conditional on enrollment. In Albania, \citet{Vrioni(2023)WP} studies the integration of private colleges into the centralized platform and reports more nuanced effects: while the reform improved the matching outcomes of low-SES students relative to their high-SES peers, it was associated with an overall welfare loss, driven by strict list-length constraints and admission uncertainty that encouraged overly cautious application portfolios. These findings highlight the importance of relaxing list-length limits and improving information about admission probabilities to fully realize the benefits of centralization.

These lessons, however, may not be directly applicable to K--12 school choice, where institutions and participation margins differ. In the United States, many districts operate opt-in systems, whereby families default to a neighborhood school unless they actively participate in a centralized process. This design resembles an off-platform setting and raises important but understudied questions: Which schools are included on the platform, who chooses to opt in, and how do information frictions and application costs shape participation? More evidence is needed to assess how these factors affect both equity and efficiency, and to clarify the trade-offs between partial and full centralization in K--12 settings.

\subsubsection{Trade-offs between Assignment Mechanisms} 

A central insight from the theoretical literature on market design is that selecting a student assignment mechanism involves inherent trade-offs between efficiency, fairness, and strategy-proofness (see \Cref{sec:theory}). 

A prominent illustration of these trade-offs is the long-standing debate surrounding the use of IA vs.\ DA in centralized assignment systems. While IA has been widely criticized for its poor incentive properties compared to DA, its comparative efficiency remains theoretically ambiguous. As discussed in \Cref{{subsec:cardinal_welfare}}, \citet{Abdulkadiroglu_Che_Yasuda(2011)AER} show that when applicants have correlated ordinal preferences over schools and priorities are determined purely by lotteries, IA can Pareto dominate DA in terms of cardinal welfare. The intuition is that IA enables applicants to signal the intensity of their preferences through strategic play: only those with strong preferences will apply to highly competitive schools. As a result, IA tends to allocate seats to applicants who value them the most, leading to higher average welfare ex ante (i.e., before lottery tie-breakers are realized). However, this potential efficiency advantage must be weighed against the risks that IA poses to participants who lack the information or the sophistication to strategize effectively.

To inform this debate, the empirical literature has evaluated the performance of IA in real-world settings---including Amsterdam, Barcelona, Beijing, Boston, New Haven, and Wake County---focusing on its implications for efficiency and equity relative to DA. Comparing these mechanisms is empirically challenging due to the difficulty of recovering preferences from data generated under manipulable mechanisms. Researchers have addressed this challenge using a range of strategies (see \Cref{subsec:preference_estimation_IA}). 

 Studies that estimate preferences using submitted ROLs generally find that IA dominates DA in terms of average welfare (typically measured in distance units). This result holds under a range of modeling assumptions: when applicants are assumed to have rational or adaptive expectations \citep{Agarwal_Somaini(2018)ECTA}, when they are modeled as a mix of naive and sophisticated players \citep{Agarwal_Somaini(2018)ECTA, Calsamiglia_Fu_Guell(2020)JPE, Xu_Hammond(2024)EcIn}, or when restrictions are imposed based on the assumption that participants do not play dominated strategies \citep{Hwang(2015)EAI, He(2017)WP, Bayraktar_Hwang(2024)WP}. 

However, a key limitation in comparing the welfare implications of IA and DA is the sensitivity of these analyses to assumptions about applicants' beliefs and decision-making behavior. To address this, a complementary line of research incorporates survey data on beliefs and stated preferences \citep{Kapor_Neilson_Zimmerman(2020)AER, De_Haan_et_al(2023)JPE}, or relies on direct measures of strategic sophistication \citep{Dur_Hammond_Morrill(2018)AEJ:Pol}. These studies consistently find that in IA settings, applicants frequently hold inaccurate beliefs about their admission chances, and that socially disadvantaged students are less likely to behave strategically. Such belief and optimization errors are critical to the welfare comparison between IA and DA. In New Haven, \citet{Kapor_Neilson_Zimmerman(2020)AER} show that accounting for applicants' subjective belief errors reverses the welfare ranking: switching from IA to DA would increase average welfare, whereas assuming rational expectations would imply the opposite. Similarly, \citet{De_Haan_et_al(2023)JPE} find that in Amsterdam, incorporating strategic mistakes---whether due to biased beliefs or suboptimal choices given those beliefs—--reverses the welfare ordering of IA and DA relative to the rational expectations benchmark. These findings suggest that, in practice, the costs of strategic mistakes under manipulable mechanisms tend to outweigh the benefits of increased expressiveness. Such concerns have been central to the adoption of strategy-proof mechanisms in districts such as Amsterdam, Boston, New Haven, and Wake County, all of which transitioned from IA to DA. 

Theoretical work by \citet{Shorrer(2019)EC19}, \citet{Calsamiglia_Martinez-Mora_Miralles(2021)EJ}, and \citet{Akbarpour_et_al(2022)JPubE} demonstrates that manipulable mechanisms such as IA raise equity concerns not only because participants differ in strategic sophistication, but also because these mechanisms systematically favor those with better outside options (such as access to private school alternatives), thereby fostering endogenous segregation (see \Cref{subsec:cardinal_welfare}). Empirical evidence supports these predictions. For example, \citet{Calsamiglia_Guell(2018)JPubE} show that while only 4\% of schools in Barcelona are private, 14\% of families who adopt risky application strategies end up enrolling in a private school when not admitted to their preferred public option. This suggests that access to outside options is disproportionately concentrated among families willing to take greater strategic risks. Likewise, \citet{Akbarpour_et_al(2022)JPubE} document that prior to New Haven's switch from IA to DA, students with outside options were significantly more likely to rank oversubscribed, highly rated schools. This behavioral gap disappeared after the reform, reinforcing the idea that access to outside options amplifies strategic distortions under IA and exacerbates existing inequalities among participants. 

However, evidence from \citet{Terrier_Pathak_Ren(2025)AEJ:App} underscores that the equity implications of strategic behavior are highly context-dependent. In particularly competitive environments, such as local authorities in England with selective schools, IA can lead high-SES families to strategically avoid ranking overdemanded schools where they could, in fact, have gained admission. Consistent with the example in \cref{{fn:example_naive}}, this behavior lowers competition for the most sought-after schools, thereby improving access for low-SES students. Because of this ``competition-for-top-schools'' effect, replacing IA with DA can, in some settings, reduce access to high-quality options for disadvantaged groups by intensifying competition. The experience in England illustrates this risk: Terrier et al.\ find that banning IA led to a deterioration in both school quality and peer composition for low-SES students.

While much of the empirical literature has focused on IA and DA, comparatively less is known about TTC, which has rarely been implemented in practice. To our knowledge, the New Orleans Recovery School District is the only system to have adopted TTC, before switching to DA after one year. According to \citet{Pathak(2017)chapter}, this change was partly motivated by the difficulty of explaining how TTC handled priorities, which created confusion among families and administrators. Existing evaluations, based primarily on counterfactual simulations \citep{Abdulkadiroglu_Agarwal_Pathak(2017)AER, Abdulkadiroglu_et_al(2020)AER:Insights, Calsamiglia_Fu_Guell(2020)JPE, Xu_Hammond(2024)EcIn}, generally find that TTC yields only modest welfare gains relative to DA. For example, \citet{Abdulkadiroglu_Agarwal_Pathak(2017)AER} show that while TTC outperforms DA in some respects, the bulk of welfare improvements in New York City's school choice reform came from moving to a coordinated system---highlighting that the gains from centralization often exceed those from switching assignment mechanisms within a centralized framework. 

As DA has gained widespread adoption, recent research has increasingly turned to how specific design features shape its performance. One key consideration is the use of list-length constraints, which are common in practice. Beyond undermining the strategy-proofness of DA, such constraints can generate significant welfare losses by increasing the risk of non-assignment and reducing match quality. This pattern has been documented in both school choice contexts \citep{Ajayi_Sidibe(2017)WP} and college admissions \citep{Luflade(2019)WP, Hernandez-Chanto(2021)WP, Ekbatani(2022)WP}, underscoring the potential gains from relaxing list-size limits or helping students better target feasible options.

Another influential design choice concerns the implementation of DA in sequential rather than static form. A growing body of theoretical and experimental work shows that sequential versions of strategy-proof mechanisms encourage more truthful reporting \citep{Li(2017)AER, Klijn_Pais_Vorsatz(2019)GEB, Bo_Hakimov(2020)EJ} and reduce the impact of behavioral biases such as loss aversion \citep{Meisner_Wangenheim(2021)WP, Dreyfuss_Glicksohn_Heffetz_Romm(2022)NBER}. Moreover, building on an emerging literature that investigates the implications of costly information acquisition for the design of matching markets \citep[e.g.,][]{Bade(2015)TE, Immorlica_et_al(2020)WP, Artemov(2021)JET, Bucher_Caplin(2021)NBER, Chen_He(2021)JET, Chen_He(2022)EconTheory, Noda(2022)GEB, Hakimov_Kubler_Pan(2023)QE, Koh_Lim(2025)WP}, there is increasing evidence that sequential implementations of DA facilitate more effective learning by allowing students to update their admission chances and refine their preferences accordingly.

Empirical studies support this view across a variety of contexts. \citet{Luflade(2019)WP} shows that in Tunisia's sequential DA college match, where students are assigned by priority tiers over multiple rounds, disclosing program vacancies between rounds enables applicants to update their expectations about admission chances and improves final allocations. In Chile, \citet{Larroucau_et_al(2025)NBER} find that providing real-time personalized
information about college applicants' admission probabilities, alongside warning messages and cutoff scores for all programs in the centralized system---resembling a one-shot implementation of the iterative DA algorithm described by \citet{Bo_Hakimov(2022)GEB}---reduces application mistakes and enhances student matching outcomes. Using data from the German college match---where the initial stages of the college-proposing DA mechanism are implemented in real time, allowing students to hold multiple offers---\citet{Grenet_He_Kubler(2022)JPE} show that multioffer mechanisms raise expected utility by enabling applicants to learn their preferences more efficiently.

Sequential assignment procedures can also mitigate inefficiencies caused by off-platform options, which arise when students reject centralized offers in favor of outside opportunities, leaving vacancies in programs that other applicants would have preferred over their own match. Using data from France's former three-round centralized college admissions system, which reallocated offers declined by students opting for off-platform alternatives, \citet{DeGroote_Fabre_Luflade_Maurel(2025)WP} estimate a dynamic model of application and acceptance decisions that captures students' trade-off between waiting for a better offer and the cost of delayed certainty. Despite substantial waiting costs, counterfactual simulations indicate that the sequential mechanism outperforms a standard one-round alternative in terms of match quality, graduation rates, and student welfare.

\subsubsection{Choice Support} 

While the design of matching mechanisms is central to the functioning of centralized assignment systems, optimizing the mechanism alone is not sufficient. Equally important is ensuring that participants can navigate the system effectively and make informed decisions. The information frictions documented in \Cref{subsubsec:information_frictions} have spurred growing interest in integrating informational interventions into school choice and college admissions platforms. These developments raise important questions for market design: what kinds of information best support families' decision-making, and how should that information be communicated?

A substantial empirical literature highlights the value of interventions that inform families about key school characteristics. Field experiments that provide accessible data on test scores \citep{Hastings_Weinstein(2008)QJE, Andrabi_Das_Khwaja(2017)AER} or graduation rates \citep{Corcoran_et_al(2018)NBER, Cohodes_et_al(2022)NBER} show that families are generally responsive to such information, shifting applications toward higher-performing schools. In several cases, these shifts translate into improved student outcomes. For example, \citet{Hastings_Weinstein(2008)QJE} find that access to school performance data increases academic achievement, while in Pakistan, \citet{Andrabi_Das_Khwaja(2017)AER} document gains in both primary school enrollment and student test scores.

Much of the early focus in this literature was on raw performance metrics, rather than school effectiveness as measured by value-added. This distinction matters, since efficiency gains in school choice hinge more on rewarding instructional quality than on sorting by peer characteristics. Moreover, because school value-added is only weakly correlated with the racial and socioeconomic composition of schools \citep{Angrist_Hull_Pathak_Walters(2024)AER:Insights}, providing families with effectiveness data may also help mitigate segregation in enrollment patterns. Recent evidence supports this potential. In Romania, \citet{Ainsworth_et_al(2023)AER} show that access to accurate value-added data leads families, particularly those of low-achieving students, to make higher-quality school choices. \citet{Campos(2024)NBER} finds similar patterns, with value-added information more influential when disseminated through social networks. Still, \citet{Ainsworth_et_al(2023)AER} emphasize that even under full information, families place substantial weight on peer composition and curriculum, limiting the extent to which they prioritize school effectiveness. In their setting, improved information closed only about one-quarter of the gap between actual and optimal school choices in terms of value-added.

Beyond value-added metrics, another promising line of research focuses on supporting applicants through real-time feedback on admission probabilities. Using experimental and quasi-experimental designs in the Chilean and New Haven centralized school choice systems, \citet{Arteaga_et_al(2022)QJE} show that ``smart matching platforms,'' which provide live risk assessments based on the evolving distribution of submitted applications, significantly reduce non-placement rates and increase the value-added of students' final assignments. Similarly, \citet{Larroucau_et_al(2025)NBER} find that personalized admission probability feedback in Chilean college admissions reduces application errors and substantially improves matching outcomes, particularly for students at risk of going unmatched or undermatched. Moreover, using randomized interventions and administrative data from Chile's nationwide
school choice process, \citet{Agte_et_al(2024)NBER} show that combining personalized feedback on admission chances with information on the quality of nearby schools leads students to match to higher-quality schools and closes the  gap between low- and high-SES students.

Providing performance feedback also shows promise. In Mexico, \citet{Bobba_Frisancho_Pariguana(2023)WP} show that offering students from low socioeconomic backgrounds individualized information about their academic standing helps align their school choices and placement with their demonstrated abilities. Three years later, students who received this feedback were 4~percentage points more likely to graduate from high school.

Yet while these interventions can improve decision-making and match quality, they do not necessarily reduce opportunity gaps. Studies by \citet{Corcoran_et_al(2018)NBER} and \citet{Cohodes_et_al(2022)NBER} show that both advantaged and disadvantaged students benefit from improved information, and in some cases, more advantaged students benefit disproportionately. \citet{Corradini(2024)WP} emphasizes the importance of general equilibrium considerations in assessing the distributional consequences of information. In her structural analysis of New York City's 2007 introduction of a high school letter-grade system, she finds that enrollment outcomes improved for Black and Hispanic students largely because white and Asian families---due to their stronger preferences for peer demographics---were less responsive to the new quality ratings. However, her counterfactual simulations reveal that if all families responded primarily to academic quality and commute time, the distributional effects would reverse: the quality of school offers would improve for white and Asian students, while minority students would receive lower-quality offers due to intensified competition for high-rated schools. To mitigate this displacement risk, Corradini proposes a targeted information design that releases coarse quality ratings (e.g., letter grades) while providing more precise signals only for schools that are disproportionately chosen by disadvantaged families. Such a design can help preserve access by dampening excess competition for the highest-performing schools.

These findings underscore the need to account for general equilibrium effects when assessing the welfare consequences of large-scale information interventions. Addressing this challenge, \citet{Allende_Gallego_Neilson(2019)WP} embed a field experiment targeting public Pre-K families in Chile within a structural model of school choice and school-side responses. Their counterfactual simulations show that while capacity constraints cut in half the average quality gains induced by information provision for low-SES families, endogenous responses by schools---in terms of quality improvement, expansion, and pricing---partially offset these effects, resulting in a net positive average treatment effect.

Finally, it is important to recognize that even well-designed informational interventions may fall short when other structural barriers constrain behavior. This is especially true in developing-country contexts, where geographic, financial, and logistical challenges can limit families' ability to act on improved information. In Ghana, for instance, \citet{Ajayi_Friedman_Lucas(2025)EJ} show that despite receiving detailed information on high school quality and selectivity, many students ultimately fail to matriculate due to constraints such as travel distance, cost, or prior commitments made before the information arrived.

In sum, information interventions have considerable promise for improving match quality and expanding access to high-performing schools, particularly when they incorporate value-added measures and personalized guidance. But their impact depends critically on thoughtful design, attention to equilibrium effects, and complementary policies that address structural barriers. Understanding how to best deploy these interventions remains an important frontier for research and policy.

\subsection{School Choice Effectiveness}
\label{subsec:school_choice_effectiveness}

Having examined how families form school preferences and how market design can help improve the allocation of students to schools, we now turn to the broader question of school choice effectiveness: to what extent do school choice policies expand access to high-quality education, and for whom? Addressing this requires a market-level perspective that considers not only the behavior of families, but also schools' supply-side responses, the distributional consequences across socioeconomic groups, and the interactions between overlapping markets. These include both vertical linkages across educational levels (e.g., primary, secondary, and higher education) and horizontal linkages with other domains, such as housing.

\subsubsection{Market-Level Effects of School Choice and Supply-Side Responses} 

A central rationale for moving away from traditional neighborhood-based assignment toward centralized school choice is its potential to enhance both the efficiency of student--school matches and the incentives faced by schools. By allowing families to select from a wider range of options, choice policies aim to better align students with schools that fit their needs and preferences, thereby improving allocative efficiency. At the same time, expanding choice introduces competitive pressure that may encourage schools to raise academic performance, adopt innovations, and respond more effectively to parental demand \citep{Friedman(1955)chapter, Hoxby(2000)AER, Hoxby(2003)chapter}. From this dual perspective, school choice offers the prospect not only of more effective matching but also of dynamic improvements in school quality and overall system performance.

Empirically identifying these market-level effects, however, remains a significant challenge, largely due to the difficulty of observing credible counterfactuals across different market structures. Much of the existing evidence therefore focuses either on the individual-level impacts of school choice \citep[e.g.,][]{Cullen_Jacob_Levitt(2006)ECMA, Abdulkadiroglu_et_al(2011)QJE, Deming_Hastings_Kane_Staiger(2014)AER, Abdulkadiroglu_Pathak_Walters(2018)AEJ:App}, or on cross-district and cross-municipality comparisons \citep{Hoxby(2000)AER, Hoxby(2003)chapter, Hsieh_Urquiola(2006)JPubE, Rothstein(2006)AER, Rothstein(2007)AER}, which yield mixed results. A related body of research examines competition across school sectors---such as Catholic, charter, or voucher schools---which often operate alongside traditional public schools and compete for students across district lines \citep[e.g.,][]{Neal(1997)JOLE, Card_Dooley_Payne(2010)AEJ:App}. While informative, these studies typically do not capture the systemic, equilibrium effects of introducing centralized school choice within an integrated public school system.

\citet{Campos_Kearns(2024)QJE} directly address this gap by exploiting a natural experiment within the Los Angeles Unified School District's Zones of Choice (ZOC) program. Unlike much of the existing literature, which focuses on individual-level impacts or cross-district comparisons, this study evaluates the market-level consequences of introducing centralized choice within a single, unified system. In 2012, approximately one-third of the district was restructured into local choice zones, while the remainder continued to operate under traditional neighborhood-based assignment. Using a difference-in-differences strategy embedded in a structural model of school competition and family preferences, the authors find that the ZOC program led to substantial gains in both student achievement and college enrollment. The largest improvements occurred in previously low-performing schools, driven not only by improved student--school matching but also by increases in school value-added, consistent with quality improvements spurred by competitive pressure. The study further documents a reduction in between-neighborhood inequality in educational outcomes and shows that families in high-choice settings tend to prioritize school effectiveness over peer composition. This latter finding may partly reflect the relative demographic homogeneity across schools within ZOC zones, where student populations are predominantly from socially disadvantaged backgrounds, which likely reduces parental concerns about peer composition. Taken together, these results offer some of the strongest empirical evidence to date that centralized school choice can generate meaningful gains in both allocative efficiency and school quality via supply-side responses.

While these findings are compelling, further research is needed to evaluate how broadly they generalize across districts with differing demographic and institutional environments. The increasing adoption of centralized assignment systems presents new opportunities in this regard, enabling researchers to exploit variation in policy implementation and leverage detailed administrative data to gain deeper insight into both parental preferences and the incentive structures shaping supply-side responses.

\subsubsection{Distributional Effects}

As illustrated in \Cref{sec:toy}, a central motivation for school choice reforms is to decouple educational opportunity from residential location, with the aim of promoting more equitable access to high-quality schools. Proponents argue that giving families the freedom to choose schools can reduce socioeconomic and racial segregation and potentially close achievement gaps by expanding access to better educational environments. Yet empirical evidence on these distributional effects remains mixed, highlighting the need to better understand how school choice systems can be designed to more effectively achieve their equity objectives.

Geographic constraints continue to pose a major barrier to equitable access. An illustrative case is provided by \citet{Laverde(2024)WP}, who shows that in Boston's pre-kindergarten system, Black students are disproportionately assigned to lower-rated schools, even under a choice-based assignment. Cross-racial gaps in assigned school quality are no smaller than under a hypothetical neighborhood-based allocation, suggesting that commuting burdens, especially for Black families, limit the effectiveness of choice systems in expanding access to high-quality options. \citet{Ajayi(2024)JHR} documents similar dynamics in a developing-country context: in Ghana, disadvantaged students are more likely to apply to lower-performing secondary schools than more advantaged peers with similar exam scores. This pattern is partly driven by longer travel distances to high-quality schools, coupled with a stronger preference for proximity among disadvantaged families. Taken together, these studies underscore that residential sorting and mobility costs can severely constrain the equalizing potential of school choice. Formal equalization of access may have limited effects on enrollment outcomes or segregation unless such structural barriers are explicitly addressed.

Within the constraints imposed by geography, priority rules nonetheless play a central role in affecting distributional outcomes \citep[e.g.,][]{Kessel_Olme(2018)WP, Burgess_Greaves_Vignoles(2020)report, Nguyen(2021)WP, Oosterbeek_Sovago_van_der_Klauw(2021)JPubE, Calsamiglia_Miralles(2023)IER, Idoux(2023)WP, Gortazar_Mayor_Montalban(2023)ECOEDU, Pariguana_Ortega-Hesles(2025)WP}. In New York City, \citet{Idoux(2023)WP} estimates that roughly half of the racial segregation among public middle schools can be attributed to admission criteria---primarily those based on geographic proximity and, to a lesser extent, academic screening---with the remainder explained by parental preferences and residential sorting. While these findings point to the potential for admissions policies to reduce segregation, the analysis also stresses that their effectiveness depends critically on how families respond. Evidence from two local middle school reforms in NYC, which curtailed academic screening and expanded access to selective schools, illustrates the importance of behavioral reactions: Idoux shows that in both cases, families adjusted their application strategies in ways that reinforced the reforms' integrative effects, resulting in school offers that were three times less segregated. However, these gains were partially offset by increased exit from the public system among white and high-income families. These findings highlight the importance of anticipating demand-side responses when designing equity-oriented admissions policies, as family reactions may amplify or undermine policy goals.

One potential tool for promoting equity through admission criteria is the use of priority rules that give preferential treatment to disadvantaged groups, such as quotas and reserves. While such policies aim to expand access for underrepresented students, theoretical insights show that, under certain combinations of preferences and priorities, instruments like majority quotas or minority reserves can unintentionally make some disadvantaged students worse off by assigning them to less preferred schools (see \Cref{subsec:quotas} for details). A separate concern arises from the ``mismatch hypothesis,'' which posits that students placed in more selective schools or colleges than they would otherwise attend may struggle academically if the environment is not well aligned with their prior preparation \citep{Sander(2004)StanfordLawReview, Sander_Taylor(2012)Book, Arcidiacono_Lovenheim(2016)JEL}. 

These concerns highlight the need for empirical evaluation to assess whether such policies achieve their equity objectives in practice. In higher education, evidence from large-scale implementations of quota-based affirmative action in India and Brazil suggests broadly positive results. Studies of engineering schools in India \citep{Bagde_Epple_Taylor(2016)AER} and federal universities in Brazil \citep{Francis-Tan_Tannuri-Pianto(2018)JDE, Mello(2022)AEJ:Pol, Barahona_Dobbin_Otero(2023)WP} find that these policies significantly improved access for marginalized students without compromising overall efficiency. In particular, \citet{Barahona_Dobbin_Otero(2023)WP} show that a nationwide policy reserving 50\% of seats for disadvantaged groups substantially boosted their enrollment in selective programs, with no evidence of negative spillovers for non-targeted students. 

In contrast, evidence from K-12 settings is more mixed. For instance, studies of Chicago's place-based affirmative action policy in elite public high schools document unintended negative effects on low-SES students' academic performance and four-year college enrollment. Possible explanations include within-school rank effects \citep{Barrow_Sartain_de-la-Torre(2020)AEJ:App} and the diversion of students from higher-performing alternatives such as charter schools \citep{Angrist_Pathak_Zarate(2023)JPubE}. These findings caution that the success of affirmative action policies depends on the broader institutional context and the set of counterfactual options available to targeted students. This calls for further empirical investigation across a wider range of implementation settings, especially in school choice systems where the use of affirmative action policies within centralized assignment mechanisms has received less attention than in higher education.

While mobility costs and admission criteria are important levers for promoting integration through school choice reforms, empirical evidence suggests that even when such reforms lead to greater integration, they may not substantially reduce racial or socioeconomic achievement gaps. This concern is illustrated by \citet{Angrist_Gray-Lobe_Pathak_Idoux(2024)WP}, who study the effects of non-neighborhood school enrollment in Boston and New York City. Using conditionally randomized school offers as instruments, they find that although commuting to non-neighborhood schools reduces minority students' isolation and increases their enrollment in schools with fewer minority peers, these integrative effects translate into minimal gains in test scores. In New York City, they even observe negative effects on on-time high school graduation and college enrollment. These disappointing outcomes are attributed to two main factors: the destination schools offer little added value relative to nearby alternatives, and in New York, the additional travel burden appears to negatively affect longer-term outcomes, independent of school quality.

These findings highlight the limitations of relying on school mobility alone as a lever for equity, and point to the need for policies that both raise the quality of schooling options and help families identify schools that offer genuine value. In this context, information frictions emerge as a critical constraint. Building on this idea, \citet{Lee_Son(2024)WP} use data from New York City's high school choice system to show how limited awareness of available options and inaccurate beliefs about admission chances influence students' application behavior. Their analysis reveals that such frictions disproportionately affect Black and Hispanic students, limiting their access to high-performing schools and reducing the overall welfare gains from school choice. Counterfactual simulations suggest that personalized school recommendations, designed to align with families' preferences and awareness, could recover between 20\% and 36\% of the lost welfare. These findings are consistent with the previously discussed evidence from \citet{Allende_Gallego_Neilson(2019)WP} and \citet{Corradini(2024)WP}, which highlight the potential for well-designed information interventions to shift demand toward more effective schools, particularly among disadvantaged students.

\subsubsection{Dynamic Effects of School Choice and Interactions with Other Markets}

A growing body of empirical research is opening up a new agenda by moving beyond static models of school choice to recognize the dynamic and interdependent nature of families' decisions. Rather than selecting schools in isolation, households make a series of interconnected choices over time and across markets, considering how residential location shapes access to schools, how current school choices influence future academic options, and how public school choice options interact with alternatives such as private schooling. These interdependencies call for a dynamic framework to more fully capture the welfare and distributional consequences of school choice reforms.

One important line of research investigates the interaction between school choice and residential location. While most empirical models in the school choice literature treat housing location as fixed, households often select neighborhoods strategically, in part to gain access to preferred schools. Several theoretical studies have formalized this joint decision-making process, embedding school choice within models of neighborhood sorting \citep[e.g.,][]{Xu(2019)JPET, Avery_Pathak(2021)AER, Grigoryan(2023)WP}. Empirical work on this interaction, however, has only begun to take shape.

A notable contribution is \citet{Park_Hahm(2023)WP}, who develop a model in which both residential location and the decision to pursue alternative schooling options are endogenously determined. Using boundary discontinuity variation in New York City, they show that up to 30\% of the racial gap in school quality reflects differences in residential choices driven by the expected utility of accessible schools, with non-minority households being more sensitive to the average test score of nearby schools than minority households. These endogenous location decisions have important implications for both policy design and empirical identification. On the methodological front, the study cautions against treated school proximity as orthogonal to unobserved preferences. In their setting, families choose where to live in part to improve their children's chances of getting into desirable schools. If this self-selection is ignored, researchers may mistakenly attribute too much importance to travel distance in determining school preferences. Indeed, the authors find that failing to account for this leads to a 15\% overestimate of commuting costs, because proximity affects both travel time and the likelihood of admission. From a policy perspective, Park and Hahm's results suggest that residential mobility can substantially undermine the effectiveness of admissions reforms. In a counterfactual exercise, they show that replacing priority-based admissions with a pure lottery at high-performing Manhattan schools would reduce racial disparities in assignments by 7\%, but this impact is cut in half once location responses are taken into account. Ignoring the interaction between school choice and residential sorting thus risks overstating the equity gains from centralized assignment reforms.

\citet{Agostinelli_Luflade_Martellini(2024)NBER} reach similar conclusions in a different context. They develop a spatial equilibrium model of residential sorting and school choice in Wake County, North Carolina, to evaluate two common equity-oriented policies: expanding the set of school options available to disadvantaged families (a ``place-based'' approach) and offering housing vouchers to help them relocate to more affluent neighborhoods (``people-based''). While both policies modestly improve peer environments for targeted students, they result in aggregate welfare losses. These losses stem in part from negative spillovers: families who previously paid to reside near high-performing schools experience a decline in peer quality and face limited options to avoid it---especially low-income families just able to afford these neighborhoods. The model also sheds light on how geography and zoning regulations affect who bears the costs of integration: schools in remote or tightly zoned areas are largely insulated from new inflows, shifting the integration burden onto a narrower set of schools. These results highlight the importance of accounting for neighborhood structure and general equilibrium effects when evaluating the broader consequences of school choice expansions.

Another strand of work explores the dynamic linkages between educational stages, showing that earlier school choices influence later ones, and that policies can leverage this sequentiality to enhance their overall impact. Building on this idea, \citet{Hahm_Park(2024)WP} develop a dynamic framework to analyze how middle school assignments affect subsequent high school choices and outcomes. Using quasi-random variation from New York City's middle school admissions process, they show that students who attend top-rated middle schools are more likely to subsequently apply to and be assigned to high-performing high schools. Through structural modeling, the authors disentangle two mechanisms: a ``priority channel,'' through which middle school affects a student's admissions chances, and an ``application channel,” through which it changes preferences and application behavior. Strikingly, the application channel accounts for 80\% of the observed effect on high school quality. Their counterfactual simulations show that eliminating geography-based eligibility rules for top middle schools would improve both efficiency and equity at the high school level, and that these effects are amplified by up to 50\% when combined with parallel reforms in high school admissions. These results suggest that policies targeting earlier educational stages can generate meaningful downstream benefits, even when reforms at later stages appear less effective in isolation.

A further dimension of dynamic interaction involves the relationship between centralized public school choice systems and the private education sector. As public choice systems expand, private providers may respond by adjusting tuition, admissions criteria, or even their market participation. \citet{Dinerstein_Smith(2021)AER} offer compelling evidence of such responses in the context of a large-scale public school funding reform in New York City: increased resources for public schools triggered a significant contraction in private school enrollment. While this reform was not directly related to school choice, it illustrates how public sector policies can reshape the education market as a whole. Previous studies have documented private school responses to voucher programs \citep{Hsieh_Urquiola(2006)JPubE, Bohlmark_Lindahl(2015)Economica} and charter school growth \citep{Chakrabarti_Roy(2016)JUE}, but there is limited evidence on how private providers react specifically to the expansion of centralized public school choice. This gap presents a promising direction for future research, particularly given its implications for equity and competition in education markets.

Finally, school choice policies may affect other connected markets, such as the labor market for teachers. By shifting enrollment patterns and redistributing students across schools, these policies can alter teacher demand across schools, potentially influencing teacher sorting, compensation, and quality. Although empirical evidence in this area remains limited, it represents a fruitful avenue for future research with important implications for the long-term effectiveness of school choice systems. A fuller understanding of these broader market interactions will be essential to informed policy design.

\newpage%

\newcommand{\study}[9]{%
  \multirow{1}{\unita}{#1} &%
  \multirow{1}{\unitb}{\centering #2} &%
  \multirow{1}{\unitc}{\centering #3} &%
  \multirow{1}{\unitd}{\centering #4} &%
  \multirow{1}{\unite}{\centering #5} &%
  \multirow{1}{\unitf}{\centering #6} &%
  \multirow{1}{\unitg}{\centering #7} &%
  \multirow{1}{\unith}{\centering #8} &%
  \multirow{1}{\uniti}{\centering #9}  %
}

\newcommand{\separator}{\addlinespace[3pt]\midrule\addlinespace[3pt]}

\renewcommand{\arraystretch}{0.9}
\makeatletter%
\global\pdfpageattr\expandafter{\the\pdfpageattr/Rotate 90}%
\makeatother%
\begin{sidewaystable}[t!]
\setlength\tabcolsep{3pt}
{\fontsize{6pt}{8pt}\selectfont
\begin{threeparttable}
\caption{Preference Estimation in Centralized School Choice and College Admissions: Survey of Empirical Studies}
\label{tab:survey}
\newcommand\unita{2.1\linewidth}
\newcommand\unitb{0.7\linewidth}
\newcommand\unitc{1.3\linewidth}
\newcommand\unitd{0.8\linewidth}
\newcommand\unite{0.5\linewidth}
\newcommand\unitf{1.2\linewidth}
\newcommand\unitg{1\linewidth}
\newcommand\unith{0.9\linewidth}
\newcommand\uniti{0.5\linewidth}
\begin{tabularx}{\textwidth}{@{}
>{\raggedright\arraybackslash\hsize=2.1\hsize}X%
>{\raggedright\arraybackslash\hsize=0.7\hsize}X%
>{\raggedright\arraybackslash\hsize=1.3\hsize}X%
>{\raggedright\arraybackslash\hsize=0.8\hsize}X%
>{\raggedright\arraybackslash\hsize=0.5\hsize}X%
>{\raggedright\arraybackslash\hsize=1.2\hsize}X%
>{\raggedright\arraybackslash\hsize=1\hsize}X%
>{\raggedright\arraybackslash\hsize=0.9\hsize}X%
>{\raggedright\arraybackslash\hsize=0.5\hsize}X%
}
\toprule
\study%
{\textbf{Study}}%
{\centering \textbf{Matching market}}%
{\centering \textbf{Country or city}}%
{\centering \textbf{Assignment mechanism}}%
{\centering \textbf{List-length restrictions}}%
{\centering \textbf{Priorities}}%
{\centering \textbf{Empirical approach(es)}}%
{\centering \textbf{Model(s)}}%
{\centering \textbf{Estimation method(s)}}
\\
\\
\addlinespace[2pt]
\midrule
\addlinespace
\multicolumn{9}{@{}c}{\textbf{\textit{Panel A. Immediate Acceptance and Other Strategic Mechanisms}}}\\
\addlinespace
\midrule
\addlinespace[3pt]
\study%
{\citet{Agarwal_Somaini(2018)ECTA}}%
{High S}%
{Cambridge (MA)}%
{Cambridge$^a$}%
{Yes}
{Coarse+lottery}%
{WTT/Porfolio}%
{Probit}%
{Gibbs}
\\
\separator
\study%
{\citet{Hastings_Kane_Staiger(2009)WP}}
{Elem.\ S}%
{Charlotte-Meckl.\ (NC)}%
{Charlotte$^b$}%
{Yes}%
{Coarse+lottery}%
{WTT}%
{Mixed logit}%
{MSL}
\\
\separator
\study%
{\citet{Campos(2024)NBER}}%
{High S}%
{Los Angeles (CA)}%
{IA}%
{No}
{Coarse+lottery}%
{WTT}%
{Logit}%
{MLE}
\\
\separator
\study%
{\citet{Campos_Kearns(2024)QJE}}%
{High S}%
{Los Angeles (CA)}%
{IA}%
{No}
{Coarse+lottery}%
{WTT/Portfolio}%
{Logit/Probit}%
{MLE/Gibbs}
\\
\separator
\study%
{\citet{Kapor_Neilson_Zimmerman(2020)AER}}%
{High S}%
{New Haven (CT)}%
{New Haven$^b$}%
{Yes}
{Coarse+lottery}
{Porfolio}%
{Probit}%
{Gibbs}
\\
\separator
\study%
{\citet{Xu_Hammond(2024)EcIn}}%
{Mid./High S}%
{Wake County (NC)}%
{IA}%
{Yes}
{Coarse+lottery}%
{Portfolio}%
{Probit}%
{Gibbs}
\\
\separator
\study%
{\citet{De_Haan_et_al(2023)JPE}}%
{Mid.\ S}%
{Amsterdam (pre 2015)}%
{Adaptive IA$^c$}%
{Yes}
{Coarse+lottery}%
{Portfolio}%
{Logit}%
{MSL}
\\
\separator
\study%
{\citet{Calsamiglia_Fu_Guell(2020)JPE}}%
{Elem.\ S}%
{Barcelona}%
{IA}%
{Yes}
{Coarse+lottery}%
{Portfolio}%
{Probit}%
{MSL}
\\
\separator
\study%
{\citet{He(2017)WP}}%
{Mid.\ S}%
{Beijing}%
{IA}%
{No}
{Lottery}
{WTT/Portfolio/UD}%
{Logit}%
{MLE}
\\
\separator
\study%
{\citet{Hwang(2015)EAI}; \citet{Bayraktar_Hwang(2024)WP}}%
{High S}%
{Seoul}%
{IA}%
{Yes}
{Coarse+lottery}
{UD}%
{Logit/Probit}%
{MEI}
\\
\addlinespace[3pt]
\midrule
\addlinespace
\multicolumn{9}{@{}c}{\textbf{\textit{Panel B. Deferred Acceptance}}}\\
\addlinespace
\midrule
\addlinespace[3pt]
\study%
{\citet{Laverde(2024)WP}}%
{Pre-K}%
{Boston}%
{DA}%
{No}%
{Coarse+lottery}%
{STT}%
{Logit}%
{MLE}
\\
\separator
\study%
{\citet{Pathak_Shi(2021)JoE}}%
{Elem.\ S}%
{Boston}%
{DA}%
{No}%
{Coarse+lottery}%
{WTT}%
{Mixed logit}%
{Gibbs/HMC}
\\
\separator
\study%
{\citet{Idoux(2023)WP}}%
{Mid.\ S}%
{NYC}%
{DA}%
{Yes}
{Coarse+lottery/Strict}%
{Portfolio}%
{Probit}%
{Gibbs}
\\
\separator
\study%
{\citet{Park_Hahm(2023)WP}}%
{Mid.\ S}%
{NYC}%
{DA}%
{No$^d$}
{Coarse+lottery/Strict}%
{WTT}%
{Mixed logit}%
{EM}
\\
\separator
\study%
{\citet{Abdulkadiroglu_Agarwal_Pathak(2017)AER}}%
{High S}%
{NYC}%
{DA}%
{Yes}
{Coarse+lottery/Strict}%
{WTT}%
{Probit}%
{Gibbs}
\\
\separator
\study%
{\citet{Narita(2018)WP}}%
{High S}%
{NYC}%
{DA}%
{Yes}
{Coarse+lottery/Strict}%
{WTT}%
{Logit}%
{MSL}
\\
\separator
\study%
{\citet{Abdulkadiroglu_Pathak_Schellenberg_Walters(2020)AER}}%
{High S}%
{NYC}%
{DA}%
{Yes}
{Coarse+lottery/Strict}%
{WTT}%
{Logit}%
{MLE}
\\
\separator
\study%
{\citet{Che_Hahm_He(2023)WP}}%
{High S}%
{NYC}%
{DA}%
{Yes}
{Coarse+lottery/Strict}%
{WTT/Stability}%
{Probit}%
{Gibbs}
\\
\separator
\study%
{\citet{Nguyen(2021)WP}}%
{High S}%
{NYC}%
{DA}%
{Yes}
{Coarse+lottery/Strict}%
{WTT}%
{Logit}%
{MLE}
\\
\separator
\study%
{\citet{Corradini(2024)WP}}%
{High S}%
{NYC}%
{DA}%
{Yes}
{Coarse+lottery/Strict}%
{WTT}%
{Logit}%
{MLE}
\\
\separator
\study%
{\citet{Lee_Son(2024)WP}}%
{High S}%
{NYC}%
{DA}%
{Yes}
{Coarse+lottery/Strict}%
{Portfolio}%
{Probit}%
{GMM}
\\
\separator
\study%
{\citet{Hahm_Park(2024)WP}}%
{Mid./High S}%
{NYC}%
{DA}%
{No/Yes}%
{Coarse+lottery/Strict}%
{Stability}%
{Mixed logit}%
{EM}
\\
\separator
\study%
{\citet{Oosterbeek_Sovago_van_der_Klauw(2021)JPubE}}%
{Mid.\ S}%
{Amsterdam (post 2015)}%
{DA}%
{No}
{Coarse+lottery}%
{WTT}%
{Logit}%
{MLE}
\\
\separator
\study%
{\citet{Agte_et_al(2024)NBER}}%
{Elem.\ S}%
{Chile}%
{DA}%
{Yes}
{Strict}%
{WTT$^e$}%
{Probit}%
{Gibbs}
\\
\separator
\study%
{\citet{Burgess_et_al(2015)EJ}}%
{Elem.\ S}%
{England}%
{DA/IA}
{Yes}
{Strict}%
{Stability}%
{Logit}%
{MLE}
\\
\separator
\study%
{\citet{Ajayi(2024)JHR}}%
{High S}%
{Ghana}%
{DA}%
{Yes}
{Strict}%
{Portfolio/Stability}%
{Logit}%
{MLE}
\\
\separator
\study%
{\citet{Ajayi_Sidibe(2024)WP}}%
{High S}%
{Ghana}%
{DA}%
{Yes}
{Strict}%
{Portfolio}%
{Probit}%
{MSM}
\\
\separator
\study%
{\citet{Bobba_Frisancho_Pariguana(2023)WP}}%
{High S}%
{Mexico City}%
{SD}%
{Yes}
{Strict}%
{Stability}%
{Logit}%
{MLE}
\\
\separator
\study%
{\citet{Ngo_Dustan(2024)AEJ:App}}%
{High S}%
{Mexico City}%
{SD}%
{Yes}
{Strict}%
{WTT/Stability}%
{Logit}%
{MLE}
\\
\separator
\study%
{\citet{Pariguana_Ortega-Hesles(2025)WP}}%
{High S}%
{Mexico City}%
{SD}%
{Yes}
{Strict}%
{Stability}%
{Logit}%
{MLE}
\\
\addlinespace[2pt]
\bottomrule
\addlinespace[3pt]
\multicolumn{9}{@{}r}{\footnotesize (Continued on next page)}
\end{tabularx}
\end{threeparttable}}%
\end{sidewaystable}%
\clearpage%
\global\pdfpageattr\expandafter{\the\pdfpageattr/Rotate 0}

\newpage%
\addtocounter{table}{-1}
\renewcommand{\arraystretch}{0.91}
\makeatletter%
\global\pdfpageattr\expandafter{\the\pdfpageattr/Rotate 90}%
\makeatother%
\begin{sidewaystable}[t!]
\setlength\tabcolsep{3pt}
{\fontsize{6pt}{8pt}\selectfont
\begin{threeparttable}
\caption{Preference Estimation in Centralized School Choice and College Admissions: Survey of Empirical Studies (Continued)}
\newcommand\unita{2.1\linewidth}
\newcommand\unitb{0.7\linewidth}
\newcommand\unitc{1.3\linewidth}
\newcommand\unitd{0.8\linewidth}
\newcommand\unite{0.5\linewidth}
\newcommand\unitf{1.2\linewidth}
\newcommand\unitg{1\linewidth}
\newcommand\unith{0.9\linewidth}
\newcommand\uniti{0.5\linewidth}
\begin{tabularx}{\textwidth}{@{}
>{\raggedright\arraybackslash\hsize=2.1\hsize}X%
>{\raggedright\arraybackslash\hsize=0.7\hsize}X%
>{\raggedright\arraybackslash\hsize=1.3\hsize}X%
>{\raggedright\arraybackslash\hsize=0.8\hsize}X%
>{\raggedright\arraybackslash\hsize=0.5\hsize}X%
>{\raggedright\arraybackslash\hsize=1.2\hsize}X%
>{\raggedright\arraybackslash\hsize=1\hsize}X%
>{\raggedright\arraybackslash\hsize=0.9\hsize}X%
>{\raggedright\arraybackslash\hsize=0.5\hsize}X%
}
\toprule
\study%
{\textbf{Study}}%
{\centering \textbf{Matching market}}%
{\centering \textbf{Country or city}}%
{\centering \textbf{Assignment mechanism}}%
{\centering \textbf{List-length restrictions}}%
{\centering \textbf{Priorities}}%
{\centering \textbf{Empirical approach(es)}}%
{\centering \textbf{Model(s)}}%
{\centering \textbf{Estimation method(s)}}
\\
\\
\addlinespace[2pt]
\midrule
\addlinespace
\multicolumn{9}{@{}c}{\textbf{\textit{Panel B. Deferred Acceptance (continued)}}}\\
\addlinespace
\midrule
\addlinespace[3pt]
\study%
{\citet{Fack_Grenet_He(2019)AER}}%
{High S}%
{Paris}%
{DA}%
{Yes}
{Strict}%
{WTT/Stability/UD}%
{Logit}%
{MLE/MEI}
\\
\separator
\study%
{\citet{Kessel_Olme(2018)WP}}%
{Elem.\ S}%
{Botkyrka (Sweden)}%
{DA}%
{Yes}
{Strict}%
{WTT}%
{Mixed logit}%
{MSL}
\\
\separator
\study%
{\citet{Anderson_et_al(2024)WP}}%
{Elem./Mid.\ S}%
{Sweden}%
{DA}%
{Yes/No}%
{Coarse+lottery/Strict}%
{WTT/Stability}%
{Logit}%
{MLE}
\\
\separator
\study%
{\citet{Akyol_Krishna(2017)EER}}%
{High S}%
{Turkey}%
{SD}
{Yes}
{Strict}%
{Stability}%
{Nested logit}%
{MSL}
\\
\separator
\study%
{\citet{Beuermann_et_al(2023)REStud}}%
{Mid.\ S}%
{Trinidad and Tobago}%
{SD}
{Yes}
{Strict}%
{Portfolio}%
{logit}%
{MLE}
\\
\separator
\study%
{\citet{Vrioni(2023)WP}}%
{Higher Ed}%
{Albania}%
{DA}%
{Yes}
{Strict}%
{Portfolio}%
{Logit}%
{MSL}
\\
\separator
\study%
{\citet{Barahona_Dobbin_Otero(2023)WP}}%
{Higher Ed}%
{Brazil}%
{Iterative DA$^f$}%
{Yes}
{Strict}%
{Stability}%
{Logit}%
{MLE}
\\
\separator
\study%
{\citet{Bordon_Fu(2015)REStud}}%
{Higher Ed}%
{Chile}%
{DA}%
{Yes}
{Strict}%
{Stability}%
{Probit}%
{MSM}
\\
\separator
\study%
{\citet{Bucarey(2018)WP}}%
{Higher Ed}%
{Chile}%
{DA}%
{Yes}
{Strict}%
{Stability}%
{Mixed Logit}%
{MSM}
\\
\separator
\study%
{\citet{Larroucau_Rios(2020)WP}}%
{Higher Ed}%
{Chile}%
{DA}%
{Yes}
{Strict}%
{Portfolio}%
{Probit}%
{Gibbs}
\\
\separator
\study%
{\citet{Larroucau_Rios(2023)WP}}%
{Higher Ed}%
{Chile}%
{DA}%
{Yes}
{Strict}%
{Portfolio}%
{Mixed Logit}%
{GII}
\\
\separator
\study%
{\citet{Kapor_Karnani_Neilson(2024)JPE}}%
{Higher Ed}%
{Chile}%
{DA}%
{Yes}
{Strict}%
{Stability}%
{Probit}%
{Gibbs}
\\
\separator
\study%
{\citet{Campos_Munoz_Bucarey_Contreras(2025)WP}}%
{Higher Ed}%
{Chile}%
{DA}%
{Yes}
{Strict}%
{WTT}%
{Logit}%
{MLE}
\\
\separator
\study%
{\citet{Wang_Wang_Ye(2025)NBER}}%
{Higher Ed}%
{Ningxia (China)}%
{SD}%
{Yes}
{Strict}%
{Portfolio}%
{Probit}%
{MSM}
\\
\separator
\study%
{\citet{Hernandez-Chanto(2021)WP}}%
{Higher Ed}%
{Costa Rica}%
{SD}%
{Yes}
{Strict}%
{Portfolio}%
{Probit}%
{Gibbs}
\\
\separator
\study%
{\citet{Chrisander_Bjerre-Nielsen(2023)WP}}%
{Higher Ed}%
{Denmark}%
{DA}%
{Yes}
{Strict}%
{WTT/Stability}%
{Logit}%
{MLE}
\\
\separator
\study%
{\citet{Gandil(2025)WP}}%
{Higher Ed}%
{Denmark}%
{DA}%
{Yes}
{Strict}%
{Stability}%
{Logit}%
{MLE}
\\
\separator
\study%
{\citet{DeGroote_Fabre_Luflade_Maurel(2025)WP}}%
{Higher Ed}%
{France}%
{Multi-round DA$^h$}%
{Yes}%
{Coarse+lottery/Strict}%
{WTT}%
{Logit}%
{MLE}
\\
\separator
\study%
{\citet{Grenet_He_Kubler(2022)JPE}}%
{Higher Ed}%
{Germany}%
{Multi-offer DA$^g$}%
{Yes}%
{Strict}%
{Stability}%
{Logit}%
{MLE}
\\
\separator
\study%
{\citet{Kirkeboen(2012)WP}}%
{Higher Ed}%
{Norway}%
{DA}%
{Yes}
{Strict}%
{WTT}%
{Logit}%
{MLE}
\\
\separator
\study%
{\citet{Hallsten(2010)AmJSociol}}%
{Higher Ed}%
{Sweden}%
{DA}%
{Yes}
{Strict}%
{WTT}%
{Logit}%
{MLE}
\\
\separator
\study%
{\citet{Luflade(2019)WP}}%
{Higher Ed}%
{Tunisia}%
{DA}%
{Yes}
{Strict}%
{WTT/Portfolio}%
{Mixed logit}%
{MSL/MSM}
\\
\separator
\study%
{\citet{Arslan(2021)EmpEcon}}%
{Higher Ed}%
{Turkey}%
{DA}%
{Yes}
{Strict}%
{WTT/Stability}%
{Logit}%
{MLE}
\\
\separator
\study%
{\citet{Akyol_Krishna_Lychagin(2024)NBER}}%
{Higher Ed}%
{Turkey}%
{DA}%
{Yes}
{Strict}%
{Stability}%
{Latent class logit}%
{MLE}
\\
\addlinespace[2pt]
\bottomrule
\end{tabularx}
\begin{tablenotes}\labelsep0.0em\fontsize{7pt}{9pt}\selectfont
\item \emph{Notes:} 
Matching markets: 
Pre-K: Pre-kindergarten; 
Elem.\ S: Elementary school; 
Mid.~S: Middle school; 
High~S: High school; 
Higher~Ed: Higher education; 
Assignment mechanisms: 
IA: Immediate Acceptance; 
DA: Deferred Acceptance; 
SD: Serial dictatorship. 
Priority rules: Coarse+lottery: coarse priorities (e.g., walk-zone, sibling attendance) combined with lottery tie-breaking; 
Strict: Strict priorities (e.g., test scores, home-school distance). 
Empirical approaches to estimating preferences: 
STT: Strict Truth-Telling;
WTT: Weak Truth-Telling; 
Portfolio: Optimal portfolio choice; 
UD: Undominated strategies. 
Estimation methods: 
EM: Expectation-Maximization algorithm; 
Gibbs: Gibbs sampling; 
GII: Generalized Indirect Inference \citep{Bruins_et_al(2018)JoE}; 
GMM: Generalized Method of Moments; 
HMC: Hamiltonian Monte Carlo; 
MEI: Moment Equalities and Inequalities; 
MLE: Maximum Likelihood Estimation; 
MSL: Maximum Simulated Likelihood; 
MSM: Method of Simulated Moments.
$^a$~Cambridge (MA) uses a modified IA mechanism with dual capacities---for both programs and schools----limiting assignment feasibility.
$^b$~Charlotte and New Haven use IA-like mechanisms where priority group takes precedence over ROL rank in determining assignments.
$^c$~The adaptive IA mechanism allows students to apply to their highest-ranked school with remaining seats in each round \citep{Mennle_Seuken(2017)WP, Dur(2019)MathSocSci}.
$^d$~NYC data used by \citet{Park_Hahm(2023)WP} reflects a period (2014--2015) when ROLs were unrestricted for middle schools; from 2017 onward, lists were capped at 12 programs.
$^e$~\cite{Agte_et_al(2024)NBER} estimate parents' preferences under WTT while accommodating information frictions by modeling limited awareness of schools and misperceptions of their characteristics.
$^f$~In Brazil, college admissions use an iterative version of DA in which students can revise up to two choices daily, guided by cutoff score updates \citep{Bo_Hakimov(2020)EJ, Bo_Hakimov(2022)GEB}.
$^g$~Germany's college admissions process begins with a multi-offer DA phase, allowing students to hold offers and update ROLS, before entering a final single-offer DA phase.
$^h$~France's former centralized college admission system (\emph{Admission Post-Bac}) featured a three-round sequential implementation of the college-proposing DA, redistributing in each round the offers declined by students who opted for off-platform alternatives.
\end{tablenotes}
\end{threeparttable}}%
\end{sidewaystable}%
\clearpage%
\global\pdfpageattr\expandafter{\the\pdfpageattr/Rotate 0}


\clearpage
\phantomsection
\setstretch{1}
\setlength{\bibsep}{3pt plus 0.3ex}
\bibliographystyle{econ.bst}
\bibliography{bibliography}
\addcontentsline{toc}{section}{References}

\end{document}